\documentclass[11pt,a4paper,oneside]{article}
\usepackage[top=3cm, bottom=3cm, left=2cm, right=2cm]{geometry} %IMPORTANT
\usepackage[english]{babel}
\usepackage[utf8]{inputenc}
\usepackage[T1]{fontenc}

\usepackage{amsthm,amsmath,amssymb,mathrsfs,dsfont,mathtools}
\usepackage{bbm}
\usepackage{bm}
\usepackage{graphicx}
\usepackage[colorlinks=true, allcolors=blue]{hyperref}
\usepackage{comment}
\usepackage{microtype} %IMPORTANT
\usepackage[all]{nowidow} %IMPORTANT
\usepackage{booktabs}

\usepackage[authoryear, round]{natbib}
\def\bSig\mathbf{\Sigma}

\newcommand{\indic}{\mathds{1}}
\newcommand{\E}{\mathds{E}}

\renewcommand{\P}{\mathds{P}}
\newcommand{\X}{\mathds{X}}

\newcommand{\R}{\mathds{R}}
\newcommand{\N}{\mathds{N}}

\newcommand{\Ztilde}{\Tilde{Z}}
\newcommand{\ztilde}{\Tilde{z}}
\newcommand{\mutilde}{\Tilde{\mu}}

\newcommand{\Ga}{\Gamma}

\newcommand{\Bcr}{\mathscr{B}}

\newcommand{\Ncr}{\mathscr{N}}

\newcommand{\BP}{\textsc{bp}}

\newcommand{\NB}{\textsc{nb}}

\newcommand{\ld}{\ldots}

\newcommand{\Var}{\operatornamewithlimits{Var}}

\newcommand{\suml}{\sum_{\ell=1}^k}
\newcommand{\sumj}{\sum_{j=1}^{N}}
\newcommand{\prodl}{\prod_{\ell=1}^k}

\newcommand{\dconv}{\overset{\rm d}{\longrightarrow}}
\newcommand{\dequal}{\overset{\rm d}{=}}
\newcommand{\dgamma}{{\rm Gamma}}
\newcommand{\dbeta}{{\rm Beta}}

\newcommand{\iid}{\stackrel{\textup{iid}}{\sim}}
\newcommand{\ind}{\stackrel{\textup{ind}}{\sim}}

\newcommand{\D}{{\rm d}}
\def\simind{\stackrel{\mbox{\scriptsize{ind}}}{\sim}}
\def\simiid{\stackrel{\mbox{\scriptsize{iid}}}{\sim}}

\usepackage{authblk}
\date{}
\title{Bayesian analysis of product feature allocation models}
\author[1]{Lorenzo Ghilotti}
\author[1]{Federico Camerlenghi}
\author[1]{Tommaso Rigon}

\affil[1]{Department of Economics, Management, and Statistics, University of Milano--Bicocca, 20126 Milano, Italy}

\providecommand{\keywords}[1]{
  \small 
  \textbf{\textit{Keywords:}} #1
  \normalsize
}

\usepackage[sf]{titlesec}
\usepackage{sectsty}
\allsectionsfont{\centering \normalfont\scshape} % Make all sections centered, the default font and small caps

\newtheorem{theorem}{Theorem}
\newtheorem{corollary}{Corollary}
\newtheorem{lemma}{Lemma}
\newtheorem{proposition}{Proposition}
\theoremstyle{definition}

\theoremstyle{remark}

\begin{document}

\maketitle

\begin{abstract}
Feature allocation models are an extension of Bayesian nonparametric clustering models, where individuals can share multiple features. We study a broad class of models whose probability distribution has a product form, which includes the popular Indian buffet process. This class plays a prominent role among existing priors, and it shares structural characteristics with Gibbs-type priors in the species sampling framework. We develop a general theory for the entire class, obtaining closed form expressions for the predictive structure and the posterior law of the underlying stochastic process. Additionally, we describe the distribution for the number of features and the number of hitherto unseen features in a future sample, leading to the $\alpha$-diversity for feature models. We also examine notable novel examples, such as mixtures of Indian buffet processes and beta Bernoulli models, where the latter entails a finite random number of features. This methodology finds significant applications in ecology, allowing the estimation of species richness for incidence data, as we demonstrate by analyzing plant diversity in Danish forests and trees in Barro Colorado Island.
\end{abstract}


\section{Introduction}

Random feature allocations have emerged as an important area of Bayesian nonparametrics. The pioneering work of \citet{Gri06} introduced the Indian buffet process (\textsc{ibp}), a stochastic mechanism for binary matrices, which is obtained by considering the infinite limit of a beta Bernoulli (\textsc{bb}) model. Unlike clustering problems, where each individual belongs to a single group, in feature allocation models every observation may possess a finite set of \emph{features} or \emph{attributes}. Shortly after its proposal, \citet{Thi07} demonstrated that de Finetti's celebrated theorem could be applied to the \textsc{ibp}, establishing its connection to the beta process of \citet{Hjo90}. This pivotal finding linked \textsc{ibp}s to the theory of completely random measures \citep{Kin67}, laying the groundwork for a new branch of Bayesian nonparametrics. These initial investigations sparked a rich stream of research, particularly within the machine learning community. \textsc{ibp} models have found widespread applicability across various domains, including Bayesian factor analysis and nonnegative matrix factorization \citep{Gri06,Kno11, Aye21}, topic modeling \citep{Wil10}, relational models \citep{Miller2009, Pal12}, and object recognition \citep{Bro15}. We refer to \citet{Teh10} and \citet{Gri11} for a comprehensive review of the early contributions.  The relevance of \textsc{ibp}s as a statistical tool is now firmly established; however, their limitations have also become apparent, as well as the need for a deeper theoretical understanding. For example, the logarithmic growth rate of the number of features in the two-parameter \textsc{ibp} \citep{Gri11} spurred the proposal of a three-parameter generalization by \cite{Teh09}, which exhibits a power-law behavior; see also \citet{Broderick2012}. Another major step was pursued by \cite{Broderick2013two}, who showed that most feature models are characterized by a combinatorial entity known as the \emph{exchangeable feature probability function} (\textsc{efpf}). More recently, \citet{Jam17} studied a general class of feature models based on completely random measures, while another extensive class, relying on random scaling of the underlying process, is investigated in \citet{Cam24}. Another recent construction is discussed in \citet{Heaukulani2020}. 

In this paper, we study the broad class of feature models defined by \citet{Bat18}, who holds the merit of characterizing all \textsc{efpf}s with a product form as mixtures of the two most widely used feature models, namely the Indian buffet process \citep[\textsc{ibp}, ][]{Teh09} and the beta Bernoulli \citep[\textsc{bb}, ][]{Gri06}. However, apart from this important representation theorem, a comprehensive statistical investigation is still lacking. We develop a general theory for this class of models, encompassing (i) the predictive structure, leading to a generalized Indian buffet metaphor, (ii) the posterior distribution of the underlying process, (iii) prior and posterior properties regarding the number of features, and (iv) the asymptotic behavior. Our findings are available in closed form, lead to computationally efficient inferential procedures, and enjoy a transparent interpretation. Moreover, this theoretical investigation allows us to identify three novel feature allocation models that stand out for their tractability: the gamma mixture of \textsc{ibp}s, and the Poisson and negative binomial mixture of \textsc{bb}s. The latter two models entail a random but finite number of possible features, therefore being structurally different from existing \textsc{ibp}-like specifications that involve infinitely many features. Finally, we highlight several remarkable parallelisms between the class of \citet{Bat18} and Gibbs-type priors for species sampling models \citep{Gne05,DeB15}. In light of these similarities, we will refer to this class as \emph{Gibbs-type feature models}. In species sampling problems, Gibbs-type priors are perhaps the most natural generalization of the Dirichlet process of \citet{Fer73}, owing to their balance between flexibility and analytical tractability. Notable examples are the Pitman--Yor process \citep{Pit97}, the normalized generalized gamma process \citep{Lijoi2007}, and mixtures of Dirichlet multinomial processes \citep{Gne10,DeBlasi2013}. For similar reasons, we argue that Gibbs-type feature models are one of the most natural extensions of the \textsc{ibp} and the \textsc{bb}.

We demonstrate here the usefulness of feature models in ecological problems as a tool to measure biodiversity. There exists a rich literature about the quantification of biodiversity \citep{Colwell2009, Magurran2011} with taxon richness, i.e., the number of different taxa present in a community, being perhaps the simplest and most natural definition. Richness estimation is, in turn, related to the notion of \emph{taxon accumulation curves} \citep{Got2001}. 
A Bayesian nonparametric inferential framework for predicting unseen species has been laid down by \citet{Lij07} for Gibbs-type priors. See also \citet{Favaro2009} for the Pitman--Yor special case and \citet{Zito2023} for a related model-based approach. These Bayesian methods are suitable for individual-based accumulation curves, that is, when species are observed one at a time. However, species are often captured or collected in chunks, and hence, each observation takes the form of a vector of binary variables accounting for the presence or absence of a species. Feature models are well-suited for this kind of data, called \emph{incidence data}, leading to a Bayesian analysis of sample-based accumulation curves. Despite the development of classical estimators for this setting \citep[e.g.][]{Col12, Chi14, Cha14, Chak19, Chi22, Chiu2023}, the Bayesian nonparametric literature remains much more limited, except for the recent works of \citet{Mas22, Cam24}. Our theoretical investigation allows for the prediction of the number of unseen species, the modeling of accumulation curves, and the quantification of biodiversity. For instance, an important theoretical result of this paper, particularly relevant for ecological applications, is the definition of the $\alpha$-diversity, a biodiversity measure that extends the notion of \citet{Pitman2003} to sample-based designs. In the proposed Poisson and negative binomial mixture of \textsc{bb} models, the $\alpha$-diversity coincides with the taxon richness, and its posterior distribution follows a Poisson and a negative binomial distribution, respectively. This leads to straightforward Bayesian estimators for the taxon richness whose uncertainty can be formally and easily quantified. 
{Although this work focuses on ecological applications, the proposed methodology is broadly applicable across various domains. %, highlighting the generality of our framework. 
For instance, in biological sciences, estimating  the number of unseen or rare genetic variants in the human
genome can help the understanding of human diseases or guide the design of effective clinical procedure \citep{ionita2009estimating, Gra(14), Zou16}.
In single-cell sequencing data, predicting the number and frequency of somatic mutations at the cellular level is essential for characterizing tumor heterogeneity, which is a key factor in cancer progression and resistance to therapy. Since the expense of sequencing is nontrivial, accurate prediction  is crucial to allocate limited sequencing budget \citep{Zhang(20)}. 
Other applications include cancer biology \citep{Chak19}, precision medicine \citep{momozawa2020unique} and microbiome analysis \citep{sanders2019optimizing}.}

The paper is structured as follows. In Section \ref{sec:feature_allocation}, we review feature allocation models. In Section \ref{sec:general_theory} we develop general theory for the class of Gibbs-type feature models. In Sections~\ref{sec:mixture_IBP}-\ref{sec:BB_mixtures} we propose and study novel examples of Gibbs-type feature allocation models, distinguishing between models with an infinite number of features (mixtures of \textsc{ibp}s) and those assuming finitely many features (mixtures of \textsc{bb}s). Simulation studies are discussed in Section \ref{sec:comparison_models}, while Section~\ref{sec:real_data} illustrates our methodology by analyzing two real datasets. The paper ends with a discussion; proofs, additional theorems, simulation studies {and additional details about the applications} are collected in the supplementary material.

\section{Feature allocation models}  \label{sec:feature_allocation}
\subsection{Preliminary concepts}

\begin{figure}
    \centering
    \includegraphics[width=0.6\textwidth]{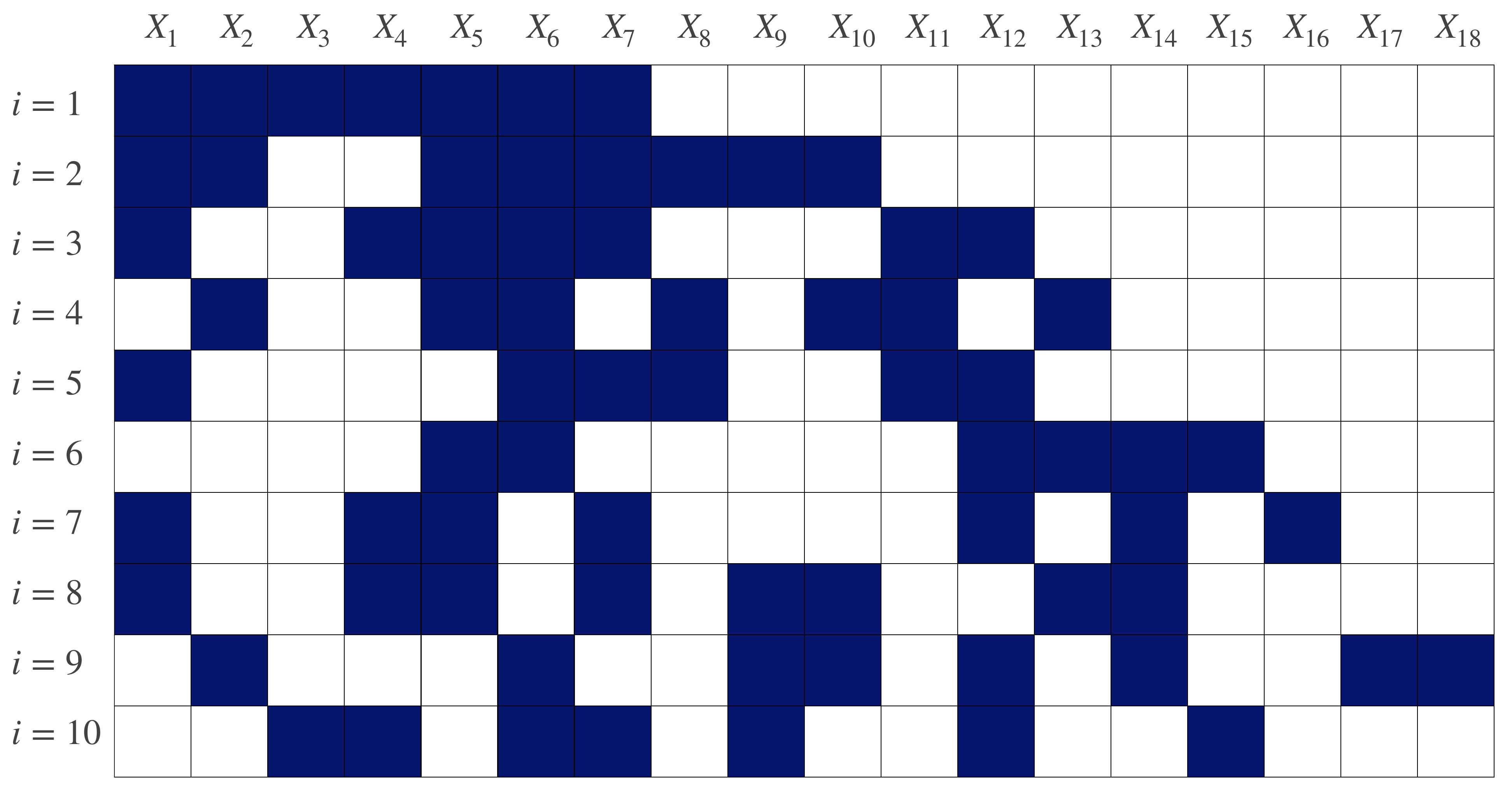}
    \caption{Matrix representation of a feature allocation, with $n=10$ individuals and $K_n = 18$ observed features. Features are in \emph{order of appearance} \citep{Broderick2013two}. The blue squares correspond to  $A_{i,\ell} = 1$ (the $i$th individual displays feature $X_\ell$), the white squares correspond to  $A_{i,\ell} = 0$ (the $i$th does not express feature~$X_\ell$). }
    \label{fig:matrix_representation}
\end{figure}

Feature allocation models describe how features are distributed among a sample of $n$ individuals \citep{Broderick2013two}. Let $[n]=\{1, \ldots, n\}$ be a set comprising the first $n$ natural numbers. An ordered \emph{feature allocation} is a sequence of non-empty sets $B_{n,\ell} \subseteq [n]$, for $\ell=1,\ldots, K_n$, such that $B_{n,\ell}$ identifies the set of individuals exhibiting the $\ell$th feature. To distinguish these sets, we assign them a \emph{label}. More precisely, suppose $\X$ is the space of feature labels, and let $X_\ell \in \X$ be the label associated with the set $B_{n,\ell}$, for $\ell = 1,\ldots, K_n$. The labels $X_\ell$ are independent and identically distributed (i.i.d.) draws, that is $X_\ell \iid P_0$, with $P_0$ being a diffuse distribution on $\X$, ensuring that the labels are almost surely (a.s.) distinct. The association between feature allocations and labels can be encoded using binary random variables~$A_{i,\ell}$ for $i=1,\dots,n$ and $\ell = 1,\dots,K_n$, so that $A_{i,\ell}$ equals $1$ if the $i$th individual displays feature $X_\ell$ and $0$ otherwise. In other terms, the random set $B_{n,\ell}$ may be written as $B_{n,\ell} = \{ i \in [n]: A_{i,\ell} = 1\}$. A feature allocation can be represented through binary matrices, as depicted in Figure~\ref{fig:matrix_representation}, where there are $n=10$ individuals and $K_n = 18$ features. Each column of the binary matrix corresponds to a set $B_{n,\ell}$, with the $i$th element of the $\ell$th column representing the value~$A_{i,\ell}$. We assume that each individual may exhibit only a finite number of features; that is, an individual belongs to a finite number of $B_{n,\ell}$'s. This implies that the total number of observed features $K_n$ is a.s. finite for any $n \ge 1$.

A \textit{feature allocation model} is a probability distribution for a random ordered  feature allocation $F_n = (B_{n,1},\dots, B_{n,K_n})$. Alternatively, one may consider the probability distribution for an \emph{unordered feature allocation}~$\tilde{F}_n = \{(\tilde{B}_{n,1}, \tilde{K}_{n,1}),\ldots, (\tilde{B}_{n,H_n}, \tilde{K}_{n,H_n})\}$, where $\tilde{B}_{n,h} \subseteq [n]$ are the $H_n \le K_n$ distinct sets among the $B_{n,1},\ldots,B_{n,K_n}$, with $\tilde{K}_{n,h}$ being the number of sets equal to $\tilde{B}_{n,h}$. The unordered $\tilde{F}_n$ is sometimes used to define feature allocation models \citep[see][]{Broderick2013two}, but for our purposes, it is more convenient to deal with the ordered $F_n$. In any event, we assume that the probability of $\tilde{F}_n$ being equal to any unordered feature allocation $\tilde{f}_n$ is equally split between the $K_n! / \prod_{h=1}^{H_n} \tilde{K}_{n,h} !$ possible orders of the sets $B_{n,1},\ldots,B_{n,K_n}$. Hence, the following relationship holds
\begin{equation*}
\P(\tilde{F}_n = \tilde{f}_n) = \frac{{K_n}!}{\prod_{h=1}^{H_n} \tilde{K}_{n,h} !} \P(F_n = f_n),
\end{equation*}
where $f_n$ is one of the possible orderings of $\tilde{f}_n$.

In this paper, we focus on feature allocation models admitting an \textit{exchangeable feature probability function}  (\textsc{efpf}). This is a probabilistic object whose role is analogous to the exchangeable partition probability function (\textsc{eppf}) in the species sampling framework \citep{Pit96}, as carefully discussed in \citet{Broderick2013two}. Specifically, let $(M_{n,1}, \ldots , M_{n, K_n})$ be the random vector of feature frequencies, where $M_{n,\ell} = \# B_{n,\ell} = \sum_{i=1}^n{A_{i,\ell}}$, for $\ell =1, \ldots , K_n$. We assume that the probability distribution of $F_n$ depends on the sample solely through the vector $\bm{M}^{(n)} = (M_{n,1}, \ldots , M_{n, K_n})$, namely 
\begin{equation*}\label{eq:def_efpf}
 \P(F_n = f_n) = \pi_{n}(m_1,\ld,m_k),
\end{equation*}
for every $f_n$ and for every $n \ge 1$, where $\pi_n$ is a $[0,1]$-valued symmetric function defined on  $\bigcup_{k \geq 0} [n]^k$, and $(m_1,\ldots,m_k)$ are the feature     frequencies for $f_n$.
The function $\pi_n$ is termed exchangeable feature probability function, and it encapsulates all the relevant properties of the model.  By construction, feature allocation models admitting an \textsc{efpf} are exchangeable, meaning that the distribution of $F_n$ is invariant under any permutation of the indexes of the $n$ individuals. It is also natural to require feature models to be \emph{Kolmogorov consistent}, which means that the probability distribution of the feature allocation for $n$ individuals coincides with that for $n+1$ individuals once the last individual is integrated out. We refer to \citet{Broderick2013two} for a detailed discussion.

\subsection{Exchangeable Gibbs-type feature allocation models} \label{sec:product-formFAM}

Among the exchangeable and consistent models, the class of \textsc{efpf} in product form introduced by \citet{Bat18} represents a special subset that is still very rich and diversified. We refer to this class as \emph{exchangeable Gibbs-type feature allocation models}, or \emph{Gibbs-type feature models} for brevity, for the evident similarity with exchangeable Gibbs-type random partitions \citep{Gne05}. We consider \textsc{efpf}s of the following product form $\pi_n (m_1, \ldots , m_k) = V_{n,k} \prod_{\ell=1}^k  W_{m_\ell} U_{n-m_\ell}$, where  $V= (V_{n,k} : (n,k) \in \N \times \N_0)$ and  $W=(W_j : j \in\N)$, $U=(U_j : j \in \N_0)$ are two sequences of non-negative weights, with $\N$ denoting the set of natural numbers and $\N_0 = \N \cup \{0\}$. Apart from some limiting cases, an important result of \cite{Bat18} states that Gibbs-type feature models are necessarily of the form
\begin{equation} \label{eq:EFPF_product}
\pi_n (m_1, \ldots , m_k) = V_{n,k} \prod_{\ell=1}^k  (1-\alpha)_{m_\ell -1} (\theta+\alpha)_{n-m_\ell},
\end{equation}
for $- \infty < \alpha < 1$ and  $-\alpha < \theta < \infty$, where $(x)_m= \Ga (x+m)/\Gamma (x)$ is the Pochhammer symbol, and $\Gamma(x)$ is the gamma function. The array $V$ must satisfy the following recurrence relationship $V_{n,k} = \sum_{j=0}^\infty (k + j)!/\{j!k!\} \{(\theta + \alpha)_n \}^j (\theta + n)^k V_{n+1,k +j}$, which guarantees the consistency of the \textsc{efpf}. The limiting case $\alpha = 1$ corresponds to no feature sharing, i.e., $M_{n, \ell} = 1$ almost surely for $\ell = 1,\dots, K_n$, whereas $\theta = -\alpha$ corresponds to complete feature sharing, that is $M_{n, \ell} = n$ almost surely, for $\ell = 1,\dots, K_n$. These degenerate situations are uninteresting in practice.

The most popular and widely used feature allocation models are of Gibbs-type. A first noteworthy example is the three-parameter Indian buffet process (\textsc{ibp}), introduced in \cite{Teh09}, with parameters $(\gamma, \alpha, \theta)$ satisfying $\gamma > 0$, $0 \leq \alpha < 1$, and $\theta > -\alpha$. The \textsc{efpf} is in product form~\eqref{eq:EFPF_product} and the $V_{n,k}$'s are given by
\begin{equation}
\label{eq:EFPF_IBP}
V_{n,k}= \frac{1}{k!} \left\{ \frac{\gamma}{(\theta+1)_{n-1}} \right\}^k
\exp \left\{ -\gamma  g_n(\theta, \alpha) \right\}, \quad \text{with} \quad g_n(\theta, \alpha) := \sum_{i=1}^n \frac{(\theta+\alpha)_{i-1}}{(\theta+1)_{i-1}}.
\end{equation}
The choice $\alpha=0$ corresponds to the two-parameter \textsc{ibp}, while the one-parameter model is obtained by further considering $\theta=1$; see  \citet{Gri11}. We stress that the distribution of $K_n$, in the \textsc{ibp} case, has unbounded support. 
A second notable example is the beta Bernoulli (\textsc{bb}), with parameters $(N, \alpha, \theta)$ such that $N \in \N$, $\alpha <0$ and $\theta > -\alpha$ \citep{Gri11}. The \textsc{efpf} of a \textsc{bb} is also in product form~\eqref{eq:EFPF_product}, with the $V_{n,k}$'s given by
\begin{equation}
\label{eq:EFPF_BB}
V_{n,k} = \binom{N}{k} \left\{ \frac{-\alpha}{(\theta+\alpha)_n}   \right\}^k   \left\{ \frac{(\theta+\alpha)_n}{(\theta)_n}  \right\}^N
\indic_{\{ 0, 1 , \ldots,N  \}} (k),
\end{equation}
where $\indic_C$ denotes the indicator function of a set $C$. The \textsc{bb} model prescribes that the observed number of features $K_n$ is bounded by $N$.  
Remarkably, a characterization theorem due to \citet{Bat18} establishes that the \textsc{ibp} and the \textsc{bb} are the building blocks of any Gibbs-type feature model. More precisely, for fixed values of $(\theta, \alpha)$, the set of Gibbs coefficients $V_{n,k}$ satisfying the aforementioned consistency condition are necessarily mixtures of the $\gamma$ and $N$ parameters of the \textsc{ibp} and the \textsc{bb}, respectively. This is better clarified in the following result. 

\begin{proposition}[Theorem 1.1 of \cite{Bat18}]\label{thm:battiston_second}
For fixed values of $(\theta, \alpha)$ such that $\theta > - \alpha$, the set of solutions of the recursions for the $V_{n,k}$'s is:
\begin{itemize} 
\item[(i)] \label{case_3ibp} for $0 \le \alpha < 1$, mixtures over $\gamma \in \mathds{R}^{+}$ of the $V_{n,k}$'s of \textsc{ibp}s with respect to a distribution $P_\gamma$;
\item[(ii)] \label{case_bb} for $\alpha<0$, mixtures over $N \in \mathds{N}$ of the $V_{n,k}$'s of \textsc{bb}s with respect to a distribution $P_N$.
\end{itemize}
\end{proposition}

Hence, any Gibbs-type feature model is obtained by considering a prior distribution for the $\gamma$ and $N$ parameters of the \textsc{ibp} and \textsc{bb} models. This draws an elegant parallelism between Gibbs-type feature models and Gibbs-type partitions of \citet{Gne05} since, in both cases, product form distributions are obtained as mixtures of a set of simple models. These analogies will be strengthened in Section~\ref{sec:general_theory}, where we will show that the parameter $\alpha$ also controls the asymptotic growth rate for the number of distinct features $K_n$, as in species sampling models, leading to the analogous of the $\alpha$-diversity of \citet{Pitman2003} for feature allocations.

\subsection{Hierarchical representations and random measures}\label{sec:hier_formulation}

Gibbs-type feature models admit a hierarchical representation in terms of Bernoulli processes (\textsc{bp}s) and random measures. For the two-parameter \textsc{ibp}, such a stochastic representation was established in the illuminating contribution of \citet{Thi07}.
Let us consider the sequence of all possible feature labels $(\tilde{X}_j)_{j \geq 1 }$, where $\tilde{X}_j \iid P_0$ for any $j \geq 1$. The labels $X_1,\dots,X_{K_n}$ observed in a sample of $n$ individuals are a subset of the complete list of labels $(\tilde{X}_j)_{j \geq 1 }$. The $i$th individual is characterized by its expressed features, i.e., by the pairs $((\tilde{X}_j, \tilde{A}_{i,j}))_{j \geq 1 }$, where each $\tilde{A}_{i,j}=1$ if the $i$th individual exhibits feature $\tilde{X}_j$, and $\tilde{A}_{i,j}=0$ otherwise. Thus, the pairs $((\tilde{X}_j, \tilde{A}_{i,j}))_{j \geq 1 }$ may be organized through a counting measure $Z_i$ on the space $\X$, which is given by
\begin{equation}\label{eq:BP_observation}
    Z_i(\cdot) = \sum_{j\geq 1} \tilde{A}_{i,j} \delta_{\tilde{X}_j}(\cdot).
\end{equation}
The feature allocation $F_n = (B_{n,1},\ldots,B_{n,K_n})$ and the binary variables $A_{i,\ell}$ can be then expressed through the counting measures $Z_i$, since we have $B_{n,\ell} = \{ i \in [n]: Z_i(\{X_\ell\}) = 1\}$ and $A_{i,\ell} = Z_i(\{X_\ell\})$. 

In Gibbs-type feature models, the $\tilde{A}_{i,j}$'s are conditionally independent Bernoulli random variables given a sequence of random probabilities $(\tilde{q}_j)_{j \ge 1}$ such that $\sum_{j \ge 1} \tilde{q}_j < \infty$ almost surely. In other terms, we suppose that $\tilde{A}_{i,j} \mid \tilde{q}_j \ind \text{Bernoulli}(\tilde{q}_j)$. We organize these probabilities using another measure $\tilde{\mu}$ on $\X$, namely
\begin{equation}\label{eq:random_measure}
\tilde{\mu}(\cdot) := \sum_{j \geq 1} \tilde{q}_j \delta_{\tilde{X}_j}(\cdot).
\end{equation} 
We will say that, conditionally on $\tilde{\mu}$, the $Z_i$'s are i.i.d. Bernoulli processes with base measure $\tilde{\mu}$, written $Z_i \mid \mutilde \iid \BP (\mutilde)$ for any $i \ge 1$. {In other words, the infinite sequence of random measures $(Z_i)_{i \ge 1}$ is exchangeable.} Summarizing, the following hierarchical representation holds:
\begin{equation} \label{eq:model_generative}
\begin{split}
    Z_i \mid \mutilde & \simiid \BP (\mutilde), \qquad i \ge 1,\\
    \mutilde & \sim Q,
\end{split}
\end{equation}
where $Q$ is the \emph{prior} distribution of $\mutilde$, i.e., the \emph{de Finetti measure}. The hierarchical generative scheme outlined in equations \eqref{eq:BP_observation}-\eqref{eq:random_measure}-\eqref{eq:model_generative} is termed a \emph{feature frequency model}, and it leads to a consistent and exchangeable \textsc{efpf} \citep{Broderick2013two}. Gibbs-type feature models always admit such a hierarchical construction under specific prior distributions for the random measure $\tilde{\mu}$. This is a well-known fact for \textsc{ibp} and \textsc{bb} models, whose random measures are denoted by $\tilde{\mu} \mid \gamma$ and $\tilde{\mu} \mid N$. Hence, thanks to Proposition~\ref{thm:battiston_second}, the law of $\tilde{\mu}$ for any Gibbs-type feature model is a mixture over $\gamma$ or $N$ of the corresponding law for the \textsc{ibp} or \textsc{bb} model.

Let us firstly consider the \textsc{bb} model with parameters $(N, \alpha, \theta)$, in which there are $N$ possible features $\tilde{X}_1,\dots,\tilde{X}_N$. The hierarchical representation of the beta Bernoulli process is straightforward as we have $\tilde \mu \mid N = \sum_{j=1}^N \tilde{q}_j \delta_{\tilde{X}_j}$, with $\tilde{q}_j \iid \dbeta(-\alpha, \theta+\alpha)$ and $\tilde{X}_j \iid P_0$, for $j=1,\ldots,N$, recalling that $\alpha < 0$ and $\theta > -\alpha$. See Lemma~\ref{lemma:bb_process_representation} in the supplementary material for a precise statement of this simple fact. The construction of the \textsc{ibp}, on the other hand, is more elaborate, and it involves infinitely many labels $\tilde{X}_j$ and probabilities~$\tilde{q}_j$. Let us define the class of homogeneous completely random measures \citep[\textsc{crm}s, ][]{Kin67} without fixed atoms, which are characterized by a Laplace functional of the following type 
\begin{equation*}
\E [e^{-\int_\X  f (x) \mutilde (\D x )}] = \exp \left\{ -\int_\X \int_0^\infty \left[1-e^{-s f (x)}\right] \rho  (\D s ) P_0(\D x)\right\},
\end{equation*}
for any measurable function $f: \X \to \R_+$, where $\rho(\mathrm{d}s)$ is an intensity measure on $\R_+$, identifying the distribution of the probabilities $(\tilde{q}_j)_{j \ge 1}$, and $P_0$ is a diffuse distribution on $\X$ from which the labels $\tilde{X}_j$ are sampled. We will write $\mutilde \sim \textsc{crm} (\rho; P_0)$. We refer to \cite{Dal08} for a mathematical treatment of \textsc{crm}s and \cite{Lij10} for a presentation of \textsc{crm}s as a unifying concept in Bayesian nonparametrics. In model \eqref{eq:model_generative} the random measure $\mutilde$ must have jumps $\tilde{q}_j \in (0,1)$, hence we require the intensity measure $\rho(\D s)$ of the \textsc{crm} to be supported in $(0,1)$. As shown in \citet{Teh09}, in the \textsc{ibp} the measure $\tilde{\mu} \mid \gamma$ is distributed as a completely random measure and, more precisely, it follows a stable-beta process, whose intensity measure $\rho(\D s)$ is 
\begin{equation}\label{eq:stable_beta_process}
\rho(\mathrm{d}s) = \gamma \frac{\Ga(1+\theta)}{\Ga(1-\alpha) \Ga(\theta+ \alpha)} s^{-\alpha -1} (1-s)^{\theta+\alpha -1} \D s.
\end{equation}
Note that $\gamma$ is sometimes called the \emph{total mass} parameter because $\gamma = \mathds{E}\left[ {\tilde{\mu}}(\mathds{X}) \right] = \sum_{j \ge 1} \mathds{E} \left[\tilde{q}_j \right]$  is the expected sum of all the probabilities. The choice $\alpha = 0$ leads to the beta process of \citet{Hjo90}, as it was established by \citet{Thi07}. There exist several sampling strategies for the weights $\tilde{q}_j \in (0,1)$, for example based on size-biased constructions or the inverse of the L\'evy measure \citep{Teh09}. Alternative and more recent approaches include stick-breaking representations \citep{Broderick2012}, or independent finite approximations \citep{Lee2023, Nguyen2023}.

\section{Predictive structure of Gibbs-type feature models}  \label{sec:general_theory}

\subsection{A buffet metaphor for Gibbs-type feature models}\label{sec:predictive}

We begin our theoretical investigation of Gibbs-type feature models by presenting the predictive distribution for the $(n+1)$th individual, given a sample of $n$ data points. In the notation of Section~\ref{sec:hier_formulation}, we study the conditional distribution of $Z_{n+1}$, given a random sample $\bm{Z}^{(n)} = (Z_1,\dots,Z_n)$, where the latter entails $K_n = k$ observed features $X_1,\ldots,X_k$ whose presence is encoded by the binary variables $A_{i,\ell}$'s. The relevant aspects of the distribution of $Z_{n+1}$ are conveyed by the vector of random variables $(Y_{n+1}, A_{n+1,1}, \ldots , A_{n+1,k})$ such that: (i) $Y_{n+1}$ is the number of new features displayed by the $(n+1)$th individual, i.e., the features hitherto unobserved in the sample $\bm{Z}^{(n)}$; (ii) each $A_{n+1,\ell}$ is a binary random variable such that $A_{n+1,\ell}=1$ if the $(n+1)$th individual displays feature $X_\ell$ and $A_{n+1,\ell}=0$ otherwise. Our first key result provides the predictive law of Gibbs-type feature models, i.e., the probability distribution
\begin{equation*}
 p_{n+1}(y, a_1, \ldots , a_k) := \P((Y_{n+1}, A_{n+1,1}, \ldots, A_{n+1,k})= (y, a_1, \ldots, a_{k}) \mid \bm{Z}^{(n)}).
\end{equation*}
We will write $\Bcr(a ; p) = p^a(1-p)^{1-a}$ to denote the probability mass function of a Bernoulli random variable with parameter $p \in (0,1)$ evaluated at $a \in \{0,1\}$.
\begin{theorem}[Predictive law]\label{thm:predictive_general_V}
Suppose the \textsc{efpf} is in product form \eqref{eq:EFPF_product}, then the predictive law is
\begin{equation*}\label{eq:pred_law}
    p_{n+1}(y, a_1, \ldots , a_k) = \binom{k+y}{k} \frac{V_{n+1,k+y}}{V_{n,k}} \{(\theta+\alpha)_n\}^y (\theta+n)^k \prod_{\ell=1}^k \Bcr\left( a_\ell; \frac{m_\ell - \alpha}{\theta +n} \right).
\end{equation*}
\end{theorem}
An important remark is in order: given the sample $\bm{Z}^{(n)}$, the random variable $Y_{n+1}$ is independent on the binary random variables $A_{n+1,1},\ldots, A_{n+1,k}$, which are also mutually independent. This is a consequence of the product form representation \eqref{eq:EFPF_product}. In the second place, the count $Y_{n+1}$ of new features depends on the sample $\bm{Z}^{(n)}$ through the sample size $n$ and the number of observed features $K_n = k$, but not the frequencies $m_1,\dots,m_k$. It is also noteworthy that what distinguishes Gibbs-type feature models is only the distribution of the number of new features $Y_{n+1}$. In fact, the probability distribution of the variables referring to the previously observed features, $A_{n+1,1},\ldots, A_{n+1,k}$, is common to all Gibbs-type feature models and does not depend on the chosen set of $V_{n,k}$'s. We will provide more precise comments on the distribution of $Y_{n+1}$ when presenting specific examples in Sections \ref{sec:mixture_IBP}--\ref{sec:BB_mixtures}.

As an immediate consequence of Theorem~\ref{thm:predictive_general_V}, one can easily determine the probability that the $(n+1)$th individual does not exhibit new features. Such a probability may be interpreted as a sample-based version of the \emph{sample coverage} \citep{Goo53, Goo56}, that is, the probability of re-observing a feature among those in the sample. However, it is worth noting that other definitions of sample coverage have been proposed in the feature setting; see, for example, \cite{Chiu2023} and references therein.
\begin{corollary}[Sample coverage]\label{cor:notobs_new_general_V}
Suppose the \textsc{efpf} is in product form \eqref{eq:EFPF_product}, then the probability that $Z_{n+1}$ does not show any new features, given $\bm{Z}^{(n)}$, is
\begin{equation*}
    \P(``Z_{n+1} \text{ has no new features}" \mid \bm{Z}^{(n)})= \P(Y_{n+1} = 0 \mid \bm{Z}^{(n)}) = \frac{V_{n+1,k}}{V_{n,k}} (\theta+n)^k.
\end{equation*}
\end{corollary}

The predictive distribution presented in Theorem~\ref{thm:predictive_general_V} can be likened to the Indian buffet metaphor \citep{Gri11}. Our metaphor imagines a scenario where ``customers'', representing individuals, sequentially enter a restaurant and select a number of ``dishes'', corresponding to feature labels $(\tilde{X}_j)_{j\geq 1}$, as depicted in Figure~\ref{fig:metaphor_pic}. Each customer has the option to choose from previously ordered dishes or select new ones. For any Gibbs-type feature model, the generative process unfolds as follows: the first customer enters the restaurant and selects $Y_1$ dishes according to the distribution
\begin{equation*} \label{eq:genertive_1}
    \P(Y_1 = y) = V_{1,y},
\end{equation*}
and $K_1 = Y_1$. The $K_1$ selected dishes are associated with labels $X_\ell$, for $\ell = 1,\ldots, K_1$, provided that $K_1 > 0$. Then, the $(n+1)$th customer enters and selects dishes in two steps. Firstly, the customer picks $Y_{n+1}$ new dishes (not chosen by the previous $n$ customers) according to the distribution
\begin{equation*}
    \P(Y_{n+1} = y \mid K_n) = \binom{k+y}{k} \frac{V_{n+1,k+y}}{V_{n,k}} \{(\theta+\alpha)_n\}^y (\theta+n)^k,
\end{equation*}
where $K_n = k$ denotes the number of distinct dishes chosen by the first $n$ customers, so that $K_{n+1} = K_n + Y_{n+1}$. If $Y_{n + 1}>0$, then the new dishes are associated with labels $X_\ell$, for $\ell = K_n +1, \ldots, K_{n+1}$. Secondly, the $(n+1)$th customer may select some of the previously chosen dishes $X_1,\dots,X_k$ as encoded by the binary variables $A_{n+1,1},\ldots, A_{n+1,k}$, whose distribution is
\begin{equation} \label{eq:generative_2}
    A_{n + 1,\ell} \mid \bm{Z}^{(n)} \simind {\rm Bernoulli} \left( \frac{m_\ell - \alpha}{\theta + n} \right),
\end{equation}
where $m_\ell$ corresponds to the number of previous customers who selected dish $X_\ell$. Higher values of $m_\ell$ correspond to a higher probability of dish $X_\ell$ being selected again.

\begin{figure}[tbp]
\includegraphics[width=0.6\textwidth]{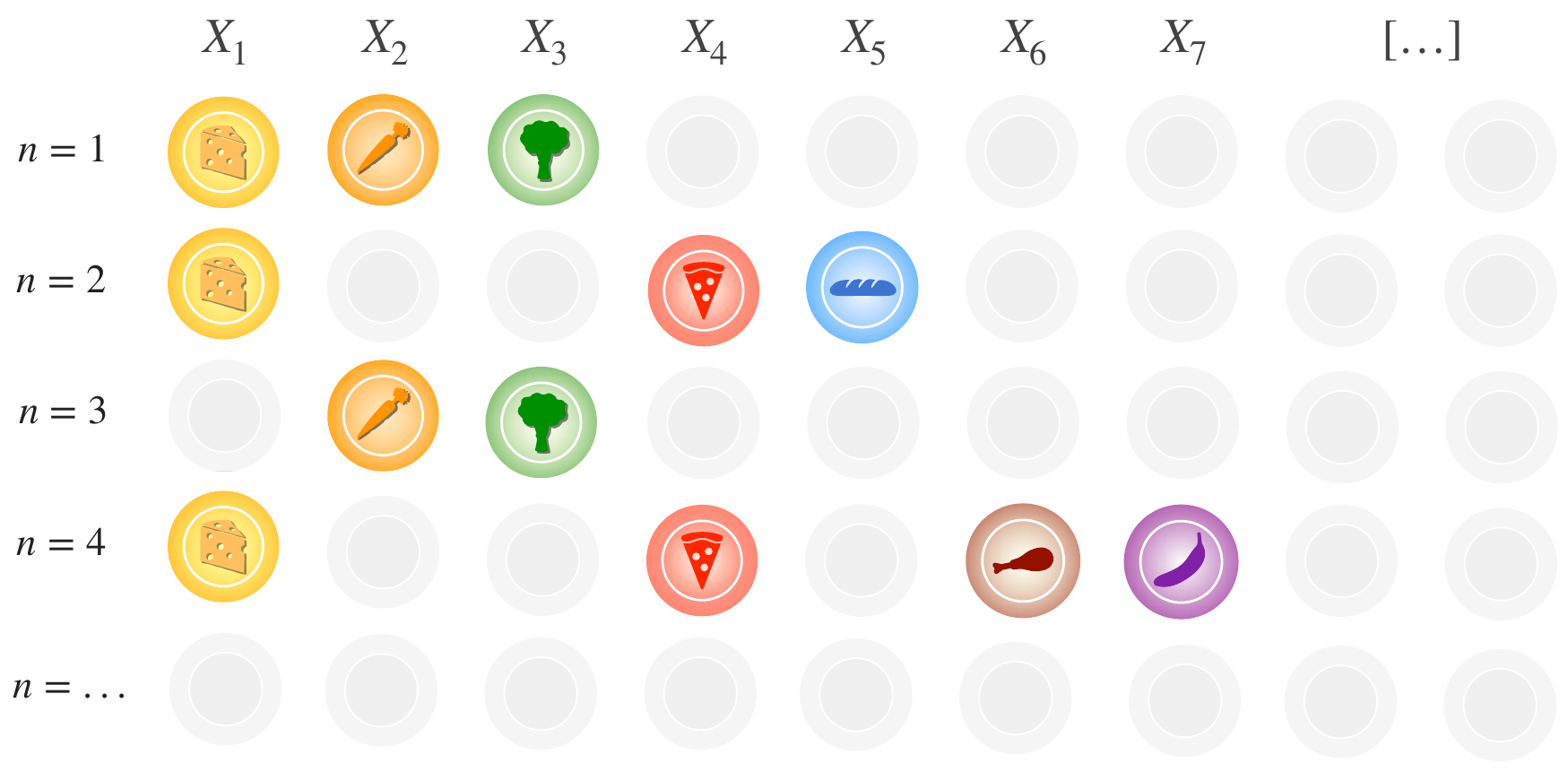}
\centering
\caption{The buffet metaphor for Gibbs-type feature models. In this example there are $n = 4$ customers and $K_n = 7$ dishes. The observed frequencies of the dishes are $(m_1,\dots, m_7) = (3,2,2,2,1,1,1)$. The numbers of new dishes picked by each customer (from top to bottom) are $Y_1 = 3$, $Y_2 = 2$, $Y_3 = 0$, and $Y_4 = 2$. \label{fig:metaphor_pic}}
\end{figure}

This general buffet metaphor provides a simple sampling strategy for any Gibbs-type feature allocation model; moreover, it offers a clear interpretation for the parameters $\theta$ and $\alpha$.
Essentially, the posterior probability of observing the $\ell$th feature can be viewed as the weighted combination of two factors: one representing the observed data and the other reflecting prior beliefs, akin to a typical Bayesian updating rule. Specifically, for $\ell = 1, \dots, k$:
\begin{equation*}
\P(A_{n+1, \ell} = 1 \mid \bm{Z}^{(n)}) = \frac{m_\ell - \alpha}{\theta +n}  =  \frac{n}{\theta +n} \hat{p}_\ell + \frac{\theta}{\theta + n} \left( - \frac{\alpha}{\theta}\right),
\end{equation*}
where $\hat{p}_\ell = m_\ell / n$ denotes the empirical frequency of the $\ell$th observed feature. When $\alpha < 0$, the quantity $-\alpha/\theta \in (0,1)$ can be conveniently interpreted as the prior frequency of a feature, while $\theta$ assumes the familiar role of prior sample size. Small values of $\theta >0$ suggest low confidence in the prior guess, and vice versa. In cases where $\alpha \in [0,1)$, $\theta$ continues to control the relevance of prior beliefs, while $-\alpha / \theta$ determines the extent of (potentially negative) shrinkage on the probability of re-observing an old feature.

Finally, we point out that the buffet metaphor for the \textsc{ibp}, also discussed in \cite{Teh09}, emerges as a special case.  In the \textsc{ibp}, the predictive process is characterized by $Y_1 \mid \gamma \sim \textup{Poisson}(\gamma)$ and $Y_{n +1} \mid K_n, \gamma \sim \textup{Poisson}\left(\gamma (\theta+\alpha)_n / (\theta+1)_n\right)$ for $n\ge 1$. Moreover,  as can be verified via Theorem \ref{thm:predictive_general_V} using $V_{n,k}$ from~\eqref{eq:EFPF_BB}, the \textsc{bb} model leads to $Y_1 \mid N \sim \textup{Binomial}\left(N, -\alpha / \theta \right)$ and $Y_{n+1} \mid K_n, N \sim \textup{Binomial}\left(N - k, -\alpha / (\theta + n) \right)$, where $\textup{Binomial}(n_0, p)$ denotes a binomial random variable with parameters $n_0 \in \mathds{N}$ and $p \in (0,1)$. Consequently, in a \textsc{bb} model, no new features can be displayed once $K_n = N$ features are observed.

\subsection{Distribution of the number of features}

We now investigate distributional properties of $K_n$, the total number of distinct features observed in a sample of size $n$. Note that we can also express $K_n$ as $Y_1 + Y_2 + \cdots +Y_n$, which is the summation of newly discovered features at each step of the buffet metaphor. In ecological applications, the expectations $\mathds{E}(K_1),\dots,\mathds{E}(K_n)$ should be regarded as a model-based \emph{rarefaction} curve \citep{Got2001, Zito23}, with the fundamental difference, compared to species sampling models, that our approach is appropriate for sample-based accumulation curves, and not individual-based. While these two frameworks are not comparable, as they refer to different sampling designs, there are several analogies between our work and \citet{Lij07}. For example, the following theorem describes the distribution of $K_n$ for any Gibbs-type feature model, which is the direct equivalent of a result of \citet{Gne05} for species sampling models.

\begin{theorem}[Number of features $K_n$] \label{thm:Kn_general}
Suppose the \textsc{efpf} is in product form \eqref{eq:EFPF_product} and let $\alpha<0$. Then, the number of distinct features observed in a sample of $n$ individuals, $K_n$, satisfies, for any $k \geq 0$, 
\begin{equation*}
    \P(K_n = k) = V_{n,k} \left\{\frac{(\theta+\alpha)_n - (\theta)_n}{\alpha} \right\}^k .
\end{equation*}
Alternatively, let $g_n(\theta, \alpha)$ be  defined as in \eqref{eq:EFPF_IBP}, if $\alpha \in [0,1)$ then for any $k \geq 0$,
\begin{equation*}
    \P(K_n = k) = V_{n,k} \left\{g_n(\theta,\alpha) \: (\theta+1)_{n-1} \right\}^k.
\end{equation*}
\end{theorem}
A preliminary version of this result is also presented in \cite{Bat18}, up to some further simplifications. In Sections \ref{sec:mixture_IBP}--\ref{sec:BB_mixtures}, we will specialize the formulas of Theorem~\ref{thm:Kn_general} for specific choices of $V_{n,k}$'s, recovering well-known distributions. In Proposition \ref{prop:msharedgeneralV} of the supplementary material, we also provide distributional insights for the statistic $K_{n,r}$ for $r \in \{1,\ldots,n\}$, denoting the number of features appearing exactly $r$ times among $n$ individuals, thereby defining $K_n = K_{n,1} + \cdots + K_{n,n}$. Note that this is helpful to determine the law of the number of rare features, i.e., features expressed only in a single individual, corresponding to $K_{n,1}$.

Theorem \ref{thm:Kn_general} and Proposition \ref{prop:msharedgeneralV} pertain to \emph{a priori} properties of $K_n$ and $K_{n,r}$. We now consider the much more compelling problem of predicting the number of hitherto unseen features $K_m^{(n)}$ that would be observed in an additional sample of size $m$, conditioned on the available data $\bm{Z}^{(n)}$, i.e., \emph{a posteriori}. In ecological contexts, this leads to the \emph{extrapolation} of the accumulation curve \citep{Zito23}, namely the expectations $\mathds{E}(K_{n +m} \mid \bm{Z}^{(n)}) = k + \mathds{E}(K_m^{(n)} \mid \bm{Z}^{(n)})$ for any $m \ge 1$. What follows is a key result of this paper since we provide the distribution of $K_m^{(n)}$ for any Gibbs-type feature model, analogous to the main finding of \citet[Proposition 1,][]{Lij07} for Gibbs-type species sampling models.

\begin{theorem}[Number of hitherto unobserved features]\label{thm:Kmn_general_V}
Suppose the \textsc{efpf} is in product form \eqref{eq:EFPF_product}. Then, the distribution of $K_m^{(n)}\mid \bm Z^{(n)}$ is such that, for any $y \geq 0$,
\begin{equation*}
    \P(K_m^{(n)} = y \mid \bm Z^{(n)}) = \binom{k+y}{k} \frac{V_{n+m,k+y}}{V_{n,k}} \left\{(\theta+n)_m\right\}^{k+y}  \left\{c_{n,m}(\theta,\alpha)\right\}^y ,
\end{equation*}
where, if $\alpha<0$,
\begin{equation*}
c_{n,m}(\theta, \alpha) =  \frac{(\theta+\alpha)_n}{\alpha} \left(  \frac{(\theta+\alpha+n)_m}{(\theta+n)_m} -1 \right),
\end{equation*}
and if $\alpha \in [0,1)$, 
\begin{equation*}
c_{n,m}(\theta, \alpha) = (\theta+1)_{n-1} \{g_{n+m}(\theta,\alpha) - g_n(\theta,\alpha)\}. 
\end{equation*}
\end{theorem}
The predictive distribution for $Y_{n+1}$ corresponds to the special case $m = 1$ in the above theorem, as it can be easily checked. Moreover, from Theorem \ref{thm:Kmn_general_V}, it is evident that the distribution of $K_{m}^{(n)}$ depends on the data $\bm Z^{(n)}$ only through the sample size $n$ and the number of distinct features $K_n$, but not on the feature counts $m_\ell$. In other words, the number of observed features $K_n = k$ is sufficient for predicting  $K_{m}^{(n)}$. This remarkable property is also a characteristic of Gibbs-type priors \citep{Lij07}.  We refer once again to Sections \ref{sec:mixture_IBP}--\ref{sec:BB_mixtures} for tractable special cases of Theorem~\ref{thm:Kmn_general_V}.

\subsection{Asymptotic behavior and $\alpha$-diversity}\label{sec:alpha_diversity}

The $\alpha$ parameter of a Gibbs-type feature model plays a key role, as hinted by Proposition~\ref{thm:battiston_second}. We show that $\alpha$ identifies the asymptotic behavior of $K_n$, precisely as in Gibbs-type species sampling models \citep{DeB15}. In particular, as $n \rightarrow \infty$, the number $K_n$ converges to a finite random variable whenever $\alpha < 0$, and it diverges when $\alpha \in [0,1)$. Moreover, $K_n$ grows at a logarithmic rate when $\alpha = 0$ and at a polynomial rate when $\alpha \in (0,1)$. This behavior is well known for  \textsc{bb} and \textsc{ibp} models \citep[e.g.][]{Gri11, Teh09}, but it is in fact a general property of Gibbs-type feature models, as the following proposition illustrates.

\begin{proposition}[$\alpha$-diversity]
\label{diversity}
 Suppose the \textsc{efpf} is in product form \eqref{eq:EFPF_product} and let $K_n$ be the number of features observed in a sample of $n$ individuals, as $n \rightarrow \infty$
\begin{itemize}
\item[(i)] if $\alpha < 0$, then $K_n \dconv N$,
\item[(ii)] if $\alpha = 0$, then $K_n / \log(n) \dconv \gamma \:\theta$,
\item[(iii)] if $0 < \alpha < 1$, then $K_n / n^\alpha \dconv \gamma \: \Gamma(\theta+1) / \{\alpha \Gamma(\theta + \alpha)\}$,
\end{itemize}
where the random variables $N$ and~$\gamma$ have distributions $P_\gamma$ and $P_N$ as in the mixture representation of the weights $V_{n,k}$'s described in Proposition~\ref{thm:battiston_second}. 
\end{proposition}

Summarizing, let $c_\alpha(n)$ be a function such that $c_\alpha(n) = 1$ if $\alpha <0$, $c_\alpha(n) =\log(n)$ if $\alpha = 0$, and $c_\alpha(n) = n^{\alpha}$ if $\alpha\in\left(0,1\right)$, then, in general, $K_{n}/c_\alpha(n) \dconv S_{\alpha}$, as $n \rightarrow \infty$.
We call random variable $S_{\alpha}$ the \emph{$\alpha$-diversity} of a feature allocation model, analogous to the $\alpha$-diversity introduced by \citet{Pitman2003}. The random variable $S_\alpha$ is often of direct interest in ecological problems, as it represents a synthetic biodiversity measurement. Naturally, comparing $\alpha$-diversities across different datasets makes sense only if they are based on the same growth rate. Note that, for fixed values of $\alpha$ and $\theta$, in the \textsc{bb} and \textsc{ibp} models the $\alpha$-diversity is deterministic. In practice, the $\alpha$-diversity is unknown, and it may be estimated employing a prior distribution for $N$ or $\gamma$, which results in a Gibbs-type feature model thanks to Proposition~\ref{thm:battiston_second}.  Moreover, the posterior law of $\gamma$ and $N$ may be obtained by combining the prior density $p_\gamma(\mathrm{d}\gamma)$ associated to $P_\gamma$, or the prior probability distribution $p_N(y)$ associated to $P_N$, with the \textsc{efpf} of equations~\eqref{eq:EFPF_IBP}-\eqref{eq:EFPF_BB}, giving respectively 
\begin{equation}\label{eq:gamma_N_posterior}
p_\gamma(\mathrm{d}\gamma \mid \bm{Z}^{(n)}) \propto p_\gamma(\mathrm{d}\gamma) \gamma^k
\exp \left\{ -\gamma  g_n(\theta, \alpha) \right\}, \quad p_N(y \mid \bm{Z}^{(n)}) \propto p_N(y) \frac{y!}{(y-k)!} \left(\frac{(\theta+\alpha)_n}{(\theta)_n} \right)^y,
\end{equation}
for $y = k,k+1,\dots$. We note that under suitable choices for $p_\gamma(\mathrm{d}\gamma)$ and $p_N(y)$, the posterior law corresponds to well-known distributions. Such a posterior distribution of $S_\alpha$ also has an elegant connection with accumulation curves, as shown in the following proposition.

\begin{proposition}[Posterior law of the $\alpha$-diversity]\label{prop_asy_anyprior}
Suppose the \textsc{efpf} is in product form \eqref{eq:EFPF_product}. Let $K_m^{(n)}\mid\bm{Z}^{(n)}$ be the number of hitherto unseen features and let $S_\alpha$ be the $\alpha$-diversity. Then, as $m\rightarrow \infty$
\begin{equation*}
\frac{K_m^{(n)} + k}{c_\alpha(m)} \mid \bm{Z}^{(n)} \dconv S'_\alpha, \qquad S'_\alpha \dequal S_{\alpha} \mid \bm{Z}^{(n)}.
\end{equation*}
\end{proposition}

Thus, the posterior law of $S_\alpha$ coincides with the $\alpha$-diversity associated with the extrapolation of the accumulation curve $K_m^{(n)} + k \mid \bm{Z}^{(n)}$, providing an insightful alternative perspective. 

\subsection{Posterior characterization}

We conclude our theoretical investigation with another pivotal result: the determination of the posterior distribution arising from model~\eqref{eq:model_generative} of $\tilde{\mu} = \sum_{j \ge 1}\tilde{q}_j\delta_{\tilde{X}_j}$, given $\bm{Z}^{(n)}$. This posterior characterization of $\tilde{\mu}$ not only elucidates the learning mechanism underpinning Gibbs-type feature models but also enables the simulation of arbitrary functionals of interest associated with $\tilde{\mu}$. Its availability also proves advantageous for Markov Chain Monte Carlo algorithms when Gibbs-type feature models are employed as latent building blocks of more complex models. 

Posterior characterizations have already been studied for specific models. For the two-parameter \textsc{ibp}, namely the beta process, \citet{Thi07} used the conjugacy result of \citet{Hjo90} to obtain the posterior distribution. For the \textsc{ibp} with $\alpha \in (0, 1)$, namely the stable-beta process, the posterior can be deduced by carefully reading \citet{Teh09}, which, in turn, is based on \citet{Kim1999}. Finally, for the \textsc{bb} model, the posterior derivation is straightforward thanks to the independence among the $\tilde{q}_j$'s and the beta-binomial conjugacy. A systematic discussion for the broad class of \textsc{crm}s priors for $\mutilde$ is provided in \citet{Jam17}. Refer also to \citet{Cam24} for related findings. The following theorem applies to any Gibbs-type feature model, albeit with notable simplifications forthcoming in Sections~\ref{sec:mixture_IBP}-\ref{sec:BB_mixtures}, where we discuss specific choices for the priors of $\gamma$ and~$N$. 

\begin{theorem}[Posterior distribution of $\mutilde$]\label{thm:post_proc}
Suppose $\bm{Z}^{(n)} = (Z_1,\dots,Z_n)$ follows model~\eqref{eq:model_generative}, then the posterior distribution of $\mutilde$, given $\bm{Z}^{(n)}$, satisfies the following decomposition
\begin{equation} \label{eq:post_proc_dist_eq}
    \mutilde\mid\bm{Z}^{(n)} \dequal \mutilde' + \mutilde^*,
\end{equation}
where $\mutilde^* \dequal \suml q_\ell \delta_{X_\ell}$ is a random measure such that $q_\ell \ind \dbeta(m_\ell-\alpha, \alpha + \theta + n - m_\ell)$, for $\ell = 1,\dots,k$, and $X_1,\dots,X_k$ denote the observed features. Moreover, the random measure $\mutilde'$ in \eqref{eq:post_proc_dist_eq} is independent of $\mutilde^*$, and its distribution depends on the value of $\alpha$, as specified below.
\begin{itemize}
    \item[(i)] If $\alpha < 0$, then the random measure $\mutilde' \mid N' \dequal \sum_{j=1}^{N'} q'_j \delta_{\tilde{X}_j}$, where $q'_j \iid \dbeta(-\alpha, \theta +\alpha + n )$ and $\tilde{X}_j \iid P_0$, for $j=1,\ldots, N'$. Moreover, $N' + k \dequal N \mid \bm{Z}^{(n)}$ as in \eqref{eq:gamma_N_posterior}. 
    \item[(ii)] If $\alpha \in [0, 1)$, then the random measure $\mutilde' \mid \gamma' \sim \textsc{crm}(\rho'; P_0)$ with updated intensity $\rho'(\mathrm{d}s) = \gamma' \Ga(1+\theta)/\{\Ga(1-\alpha) \Ga(\theta+ \alpha)\} s^{-\alpha -1} (1-s)^{\theta+\alpha + n -1} \D s$. Moreover, $\gamma'\dequal  \gamma \mid \bm{Z}^{(n)}$ as in \eqref{eq:gamma_N_posterior}.
\end{itemize}
\end{theorem}

The previous distributional equality \eqref{eq:post_proc_dist_eq} decomposes the posterior distribution of $\mutilde$ into two parts: $\mutilde'$ accounts for the newly observed features, while $\mutilde^*$ deals with those previously observed. Regarding the latter, the $(n+1)$th individual exhibits an existing feature $X_\ell$ if $A_{n+1,\ell} = 1$, where each $A_{n+1,\ell}\mid q_\ell  \ind {\rm Bernoulli}(q_\ell)$ and $q_\ell  \ind \dbeta(m_\ell-\alpha, \alpha + \theta + n - m_\ell)$. By integrating out the $q_\ell$'s, we obtain $A_{n+1,\ell} \mid \bm{Z}^{(n)} \ind {\rm Bernoulli}\left((m_\ell - \alpha)/ (\theta + n) \right)$, consistent with the predictive structure~\eqref{eq:generative_2}. It is worth highlighting that the distribution of $\tilde{\mu}^*$ remains the same across all Gibbs-type feature models. However, $\mutilde'$ exhibits structural differences depending on the specific choices for the $V_{n,k}$'s. In the case of $\alpha \in [0,1)$, we observe that $\mutilde' \mid \gamma' \sim \textsc{crm}(\rho'; P_0)$ benefits from a conjugacy property, as the intensity measure $\rho'(\mathrm{d}s)$ characterizes a stable-beta process with updated parameters $(\gamma' (\alpha + \theta)_n / (\theta + 1)_n, \alpha, \theta + n)$, a point noted by \citet{Teh09} in the \textsc{ibp} case.

Simulating posterior samples for $\tilde{\mu}$ is straightforward. Initially, one needs to draw values from $N \mid \bm{Z}^{(n)}$ or $\gamma \mid \bm{Z}^{(n)}$, which typically follow highly tractable distributions. Then, conditioned on $N$ or $\gamma$, both random measures $\mutilde'$ and $\mutilde^*$ are simple to sample from: $\tilde{\mu}^*$ involves a finite number of beta random variables, as does $\tilde{\mu}'$ when $\alpha < 0$. Even simulating $\tilde{\mu}'$ when $\alpha \in [0,1)$ is straightforward, as due to conjugacy, $\tilde{\mu}'$ follows a stable-beta process, for which efficient sampling algorithms exist; see, for example, \citet{Teh09}.

\section{Gamma mixture of Indian buffet processes}  \label{sec:mixture_IBP}

\subsection{Predictive structure, number of features, and $\alpha$-diversity}

In this section, we discuss relevant special cases of Gibbs-type feature models within the $\alpha = 0$ and $\alpha \in (0,1)$ regimes, where there are infinitely many possible features $\tilde{X}_j$. We define a new Gibbs-type feature model termed \emph{gamma mixture of \textsc{ibp}s} by employing a gamma prior for $\gamma$. Upon examining the posterior distribution in~\eqref{eq:gamma_N_posterior}, it becomes evident that a gamma prior is \emph{conjugate}, as its posterior remains gamma with updated parameters. If $\gamma \sim \dgamma(a,b)$, then the associated \textsc{efpf}, obtained by integrating \eqref{eq:EFPF_IBP} with respect to the prior density, follows a product form~\eqref{eq:EFPF_product} and the weights are
\begin{equation}\label{eq:efpf_gamma_mix}
    V_{n,k}  = \frac{1}{k!}\frac{b^a \,(a)_k}{\{(\theta+1)_{n-1}\}^k \{b + g_n(\theta, \alpha)\}^{a + k}}.
\end{equation}
This Gibbs-type feature model has connections with the stable-beta scaled process of \cite{Cam24}, which is, in fact, a special case of~\eqref{eq:efpf_gamma_mix}. In particular, a stable-beta scaled process with parameters $(\alpha, c, d)$ can be represented as a gamma mixture of \textsc{ibp}s under the constraint $\theta = 1 - \alpha$ and prior distribution $\gamma \sim \dgamma(c+1, d (1-\alpha)/\alpha)$. Such a hierarchical representation is not discussed in \citet{Cam24}, but it can be proved by directly inspecting the \textsc{efpf}s. 

We now compare the \textsc{ibp} of \cite{Teh09} with the feature model~\eqref{eq:efpf_gamma_mix}, utilizing the general findings from Section~\ref{sec:general_theory}. The predictive laws of the \textsc{ibp} and the gamma mixture of \textsc{ibp}s follow by specializing Theorem~\ref{thm:predictive_general_V}, substituting the $V_{n,k}$'s of equations~\eqref{eq:EFPF_IBP} and~\eqref{eq:efpf_gamma_mix}, respectively, into the general formulas. As discussed in Section~\ref{sec:predictive}, the predictive distributions of Gibbs-type feature models differ only in the law governing the number of new features, whereas the law of the binary variables $A_{n+1,1},\ldots, A_{n+1,k}$, already described in~\eqref{eq:generative_2}, is the same. Thus, for the sake of brevity, here we concentrate on the distribution of the number of new features $Y_{n+1}$. Let $Y \sim \text{NegBinomial} (n_0, \mu_0)$ denote a negative binomial random variable with mean parameter $\mu_0 > 0$ and dispersion parameter $n_0 > 0$, where its probability mass function $\Ncr (y ; n_0, \mu_0) \propto p^{n_0}  (1-p)^y$ is defined for any $y \in \mathds{N}_0$, with $p = n_0 / (\mu_0 + n_0)$, so that $\mathds{E}(Y) = \mu_0$ and $\Var(Y) = \mu_0 + \mu_0^2 / n_0$. Moreover, let  $Y_{n+1} \mid K_n, \gamma$ be the number of new features for the $(n+1)$th individual in the \textsc{ibp} case, and let $Y_{n+1} \mid K_n$ be the same quantity for the gamma mixture model. Then, simple calculus yields  
\begin{equation*}
Y_{n+1} \mid K_n, \gamma \sim \text{Poisson}\left(\gamma \frac{(\theta+\alpha)_{n}}{(\theta+1)_{n}}\right), \quad Y_{n+1} \mid K_n \sim \text{NegBinomial}\left(a + k, \frac{a + k}{b + g_n(\theta, \alpha)}\frac{(\theta+\alpha)_{n}}{(\theta+1)_{n}}\right).
\end{equation*}
Additional distributional properties can be derived by specializing the results of Section \ref{sec:general_theory} for the \textsc{ibp} and gamma mixture of \textsc{ibp}s. We summarize some of these properties in the following proposition and refer to the supplementary material for additional findings and simplifications, such as the distribution of the number of shared features $K_{n,r}$ or the sample coverage.

\begin{proposition}\label{prop_k_gamma_ibp} Suppose the \textsc{efpf} is in product form \eqref{eq:EFPF_product} with $\alpha \in [0, 1)$. Let $K_n \mid \gamma$ and $K_n$ be the number of features observed in a sample of $n$ individuals for the \textsc{ibp} and the gamma mixture in~\eqref{eq:efpf_gamma_mix}. Then, we have
\begin{equation} \label{eq:Kn_P_NB}
K_{n} \mid \gamma \sim \textup{Poisson} \left( \gamma g_n(\theta, \alpha)\right), \qquad K_{n} \sim \textup{NegBinomial} \left(a, (a/b)g_n(\theta, \alpha)\right).
\end{equation}
Moreover, let $K_m^{(n)}\mid\bm{Z}^{(n)}$ be the number of hitherto unseen features in an additional sample of size $m \ge 1$, then for the \textsc{ibp} 
\begin{equation} \label{eq:Knm_Poiss}
K_m^{(n)}\mid \bm{Z}^{(n)}, \gamma \sim \textup{Poisson}\left( \gamma \left( g_{n+m}(\theta, \alpha) - g_n(\theta,\alpha) \right)\right),
\end{equation}
whereas for the gamma mixture
\begin{equation} \label{eq:Kmn_NB}
K_m^{(n)}\mid \bm{Z}^{(n)} \sim \textup{NegBinomial}\left(a + k,  \frac{a + k}{b + g_n(\theta, \alpha)} \left( g_{n+m}(\theta, \alpha) - g_n(\theta,\alpha) \right) \right).
\end{equation}
\end{proposition}

The results regarding the \textsc{ibp} have been deduced from the general theory outlined in Theorem~\ref{thm:Kn_general} and Theorem~\ref{thm:Kmn_general_V}, albeit these results were already known. Indeed, the distribution of $K_{n} \mid \gamma$  has been determined by \citet{Teh09}, while the distribution of $K_m^{(n)}\mid \bm{Z}^{(n)}, \gamma$ has been unveiled by \citet{Mas22}. Conversely, the results concerning the gamma mixture are novel. The expected values of the number of features, also known as rarefaction, are $\mathds{E}(K_n \mid \gamma) = \gamma g_n(\theta, \alpha)$ and $\mathds{E}(K_n) = (a/b) g_n(\theta, \alpha)$.  The function $g_n(\theta, \alpha)$ has an interesting interpretation, being the expected value of the number of clusters in a sample of size $n$ from a Pitman--Yor process \citep{Pit97} with parameters $(\alpha, \theta)$. The $\alpha = 0$ case corresponds to the Dirichlet process \citep{Fer73}, reducing to $g_n(\theta, 0) = \sum_{i=1}^n \theta / (\theta + i - 1)$. This fact underscores once more the close relationship between Gibbs-type feature models and Gibbs-type species sampling models. 

One notable advantage of the negative binomial distribution derived from the gamma mixture model is its ability to introduce overdispersion in $K_n$. Furthermore, the posterior distribution of $K_m^{(n)}\mid \bm{Z}^{(n)}, \gamma$, corresponding to the \textsc{ibp} case, remains independent of the number of observed features $K_n = k$, a characteristic that some may find unappealing. In contrast, the negative binomial posterior for $K_m^{(n)}\mid \bm{Z}^{(n)}$ considers $k$, influencing both the mean and variance of the distribution. Higher values of $k$ result in overdispersion, which is often desirable. The Bayesian estimators for the number of previously unseen features, crucial for extrapolating the accumulation curve, are as follows:
\begin{equation*}
\mathds{E}(K_m^{(n)}\mid \bm{Z}^{(n)}, \gamma)  = \gamma\left( g_{n+m}(\theta, \alpha) - g_n(\theta,\alpha) \right), \quad \mathds{E}(K_m^{(n)}\mid \bm{Z}^{(n)})  = \frac{a + k}{b + g_n(\theta, \alpha)}\left( g_{n+m}(\theta, \alpha) - g_n(\theta,\alpha) \right),
\end{equation*}
for the \textsc{ibp} and the gamma mixture of \textsc{ibp}s, respectively. 

Recall that, as stated in Proposition~\ref{diversity}, the parameter $\alpha$ controls the growth rate of $K_n$ so that when $\alpha = 0$, then $K_n / \log(n) \overset{\textup{d}}{\longrightarrow} \gamma \:\theta$, whereas when $0 < \alpha < 1$, then $K_n / n^\alpha \overset{\textup{d}}{\longrightarrow} \gamma \: \Gamma(\theta+1) / \{\alpha \Gamma(\theta + \alpha)\}$. In the \textsc{ibp}, the parameter $\gamma$, and hence the $\alpha$-diversity, is deterministic \citep{Mas22}. Conversely, by definition, $\gamma$ follows a priori a $\dgamma(a,b)$ in the gamma mixture model~\eqref{eq:efpf_gamma_mix}, which allows the Bayesian learning of the $\alpha$-diversity through its posterior.  

\begin{proposition} \label{prop_asy_km_gamma_3ibp}
Suppose the \textsc{efpf} is in product form \eqref{eq:EFPF_product} with $\alpha \in [0, 1)$ and the $V_{n,k}$'s defined as in~\eqref{eq:efpf_gamma_mix}, so that a priori $\gamma \sim \dgamma(a,b)$. Then, the posterior is $\gamma \mid \bm{Z}^{(n)} \sim \dgamma(a + k, b + g_n(\theta, \alpha))$ and therefore the $\alpha$-diversity, for $\alpha = 0$, is given by
    \begin{equation*}
        \frac{K_m^{(n)}}{\log(m)} \mid \bm{Z}^{(n)} \dconv S'_\alpha, \qquad S'_\alpha \dequal S_\alpha \mid \bm{Z}^{(n)} \sim \dgamma\left(a + k, \frac{b + g_n(\theta,0)}{\theta}\right),
    \end{equation*}
  as $m \to \infty$, whereas for $\alpha \in (0,1)$   
    \begin{equation*}
        \frac{K_m^{(n)}}{m^\alpha} \mid \bm{Z}^{(n)} \dconv S'_\alpha, \qquad S'_\alpha \dequal S_\alpha \mid \bm{Z}^{(n)}  \sim \dgamma\left(a + k, \{b + g_n(\theta,\alpha)\} \frac{\Ga(\theta+\alpha)\alpha}{\Ga(\theta+1)}\right) .
    \end{equation*}   
\end{proposition}

A consequence of the deterministic $\alpha$-diversity in the \textsc{ibp} is that the width of the credible intervals for $K_m^{(n)}$ grows at a rate slower than $m^{\alpha}$.  In contrast, the mixtures of \textsc{ibp}s yield larger credible intervals, whose widths grow proportionally to $m^\alpha$, as highlighted by the previous proposition. This difference can be observed in simulation study B of the supplementary material: Figure \ref{fig:extr_simulation_B} suggests that the \textsc{ibp} underestimates the uncertainty in predicting the number of unseen features. Proposition \ref{prop_asy_km_gamma_3ibp} presents some of the first results concerning the posterior distribution of the $\alpha$-diversity for feature allocation models, with an early contribution for the stable-beta scaled process being available in~\citet{Cam24}. Analogous findings for the Pitman--Yor species sampling model are given in \citet{Favaro2009} when $\alpha \in (0,1)$, while for the Dirichlet process ($\alpha = 0$) an interpretable and tractable prior is proposed in \citet{Zito2023}.

\subsection{Posterior characterizations and negative binomial processes}  \label{sec:mixture_IBP_process}

We specialize here the posterior characterization of Theorem~\ref{thm:post_proc} for the gamma mixture of \textsc{ibp}s, which can be conveniently described in terms of \emph{negative binomial processes} \citep{Greg84}, whose use in Bayesian nonparametrics is still much unexplored. Building upon \cite{Greg84}, we say that $\mutilde$ is a \textit{negative binomial random measure} with parameter $a > 0$, intensity measure $\rho(\D s)$ on $\R_+$ and diffuse base measure $P_0$ on $\X$, if $\mutilde$ has Laplace functional
\begin{equation} \label{eq:NB_Laplace_functional}
    \E [e^{-\int_\X  f (x) \mutilde (\D x )}] = \left\{  1+ \int_\X  \int_0^\infty  \left[1-e^{-s f(x)}\right] \rho (\D s) P_0(\D x )  \right\}^{-a},
\end{equation}
for any measurable function $f: \X \to \R_+$. We will write $\mutilde \sim \textsc{nb} (a, \rho; P_0)$ and we assume the intensity measure $\rho(\D s)$ is supported in $(0,1)$ as before. A negative binomial random measure may arise by considering a \textsc{crm} with random intensity measure $\tilde{c} \: \rho(\mathrm{d}s)$, where $\tilde{c}$ is distributed as a gamma random variable with parameters $(a, 1)$. Hence, the hierarchical formulation for the gamma mixture of \textsc{ibp}s becomes
\begin{equation}\label{model:ibp_gamma_process}
\begin{aligned}
    Z_i \mid \mutilde &\iid \BP(\mutilde), \qquad \qquad i \ge 1,\\
    \tilde \mu & \sim \NB(a,\rho; P_0),
\end{aligned}
\end{equation}
where the intensity measure is $\rho(\D s) = (1/b)\Gamma(1+\theta) / \{\Gamma(1-\alpha) \Gamma(\theta+\alpha)\}s^{-\alpha-1}(1-s)^{\theta+\alpha-1} \D s$. The proof is shown in Section \ref{app:proofs_mixtureIBPs}, but it is merely a consequence of mixing the intensity measure of a completely random measure with respect to a gamma distribution. The following corollary of Theorem~\ref{thm:post_proc} characterizes the posterior distribution of the process $\mutilde$, given $\bm{Z}^{(n)}$, in terms of negative binomial random measures. 
\begin{corollary}\label{thm:post_proc_ibp_gamma}
Suppose $\bm{Z}^{(n)} = (Z_1,\dots,Z_n)$ follows model~\eqref{model:ibp_gamma_process}, then the posterior distribution of $\mutilde$, given $\bm{Z}^{(n)}$, satisfies the  decomposition $\mutilde\mid\bm{Z}^{(n)} \dequal \mutilde' + \mutilde^*$ in \eqref{eq:post_proc_dist_eq}, where $\mutilde'$ and $\mutilde^* $ are independent random measures such that $\mutilde^* $ is distributed as in Theorem~\ref{thm:post_proc}, whereas $\tilde \mu'  \sim \NB (a + k, \rho' ; P_0)$, with updated intensity $\rho'(\D s) = 1 / \{b + g_n(\theta, \alpha)\} \Gamma(1+\theta)/ \{\Gamma(1-\alpha) \Gamma(\theta+\alpha)\}s^{-\alpha-1}(1-s)^{n+\theta+\alpha-1} \D s$.
\end{corollary}

\section{Gibbs-type feature models with finitely many features} \label{sec:BB_mixtures}

\subsection{Predictive structure, number of features, and richness estimation}

In species sampling models, mixtures of Dirichlet multinomial processes are an important subclass of Gibbs-type priors \citep[see, e.g.][]{DeB15}, which include, for instance, the models of \citet{Gne10,DeBlasi2013}. In feature allocation models, a similar role is played by mixtures of \textsc{bb} models with $N$ features, which corresponds to the $\alpha < 0$ case. These feature allocation models assume a finite number of features $N$, representing the \emph{richness} in ecological problems.The standard \textsc{bb} model assumes that $N$ is known in advance, although this is a  critical parameter and object of inference. 
In the following, we concentrate on two novel and tractable specifications: (i) $N$ is a Poisson random variable with parameter $\lambda >0$, referred to as \emph{Poisson mixture of \textsc{bb}s}; (ii) $N$ is a negative binomial random variable with parameters $(n_0, \mu_0)$, referred to as \emph{negative binomial mixture of \textsc{bb}s}. These random variables serve as the prior distribution for the richness, enabling its Bayesian learning. We begin by providing the expressions for the corresponding \textsc{efpf}s.

\begin{proposition}\label{prop_efpf_poiss_nb}
If $N \sim \textup{Poisson} (\lambda)$ in the mixture representation of Proposition~\ref{thm:battiston_second}, then the model has \textsc{efpf} in product form~\eqref{eq:EFPF_product} and the $V_{n,k}$'s are given by
\begin{equation}
\label{eq:EFPF_BB_Poiss}
V_{n,k}  = 
\frac{1}{k!} \exp  \left\{ -\lambda \left(1-\frac{(\theta+\alpha)_n}{(\theta)_n}\right) \right\}  \left\{ \frac{-\lambda \alpha}{(\theta)_n} \right\}^k .
\end{equation}
If instead $N \sim \textup{NegBinomial}(n_0, \mu_0)$, then the $V_{n,k}$'s are given by
\begin{equation}
\label{eq:EFPF_BB_NB}
V_{n,k}  = 
\binom{k+n_0-1}{k}  \left\{ \frac{ -\alpha}{(\theta)_n} \frac{\mu_0}{\mu_0 + n_0}\right\}^k  \left( 1- \frac{\mu_0}{\mu_0 + n_0} \frac{(\theta+\alpha)_n}{(\theta)_n}  \right)^{-n_0-k}   \left(\frac{n_0}{\mu_0 + n_0}\right)^{n_0}.
\end{equation}
\end{proposition}
Clearly, the negative binomial mixture of \textsc{bb}s allows for a higher degree of prior uncertainty regarding the total number of features $N$ compared to the Poisson case, as the negative binomial induces overdispersion. It is worth noting that the negative binomial mixture of \textsc{bb}s may be obtained by choosing a gamma prior for the $\lambda$ parameter of the Poisson mixture of \textsc{bb}s. Specifically, assuming $N\mid \lambda \sim \textup{Poisson} (\lambda)$, and $\lambda \sim \dgamma(a,b)$ is equivalent to having $N \sim \textup{NegBinomial} (n_0,\mu_0)$, with $n_0 = a$ and $\mu_0 = a / b$.  This provides a hierarchical justification for the negative binomial mixture of \textsc{bb}s: one may initially consider the Poisson mixture model, but if there is uncertainty about $\lambda$, then one could learn it by employing a gamma prior, resulting in a negative binomial mixture of \textsc{bb}s. 

We now apply the general results of Section~\ref{sec:general_theory} to the aforementioned mixtures of \textsc{bb} models. For brevity, we focus solely on the number of features $K_n$, representing the rarefaction, and the number of hitherto unseen features $K_m^{(n)} \mid \bm{Z}^{(n)}$, leading to the extrapolation of the accumulation curves. It is worth reiterating that the predictive distribution $Y_{n+1} \mid \bm{Z}^{(n)}$ involved in the buffet metaphor can be derived as a special case, setting $m = 1$ in the distribution of $K_m^{(n)} \mid \bm{Z}^{(n)}$, which is a trivial task in the following formulas.

\begin{proposition} \label{prop_pred_poiss_nb}
Suppose the \textsc{efpf} is in product form~\eqref{eq:EFPF_product} with $\alpha < 0$ and let $p_n(\theta, \alpha) = 1-(\theta+\alpha)_n / (\theta)_n$. If $N$ is fixed, corresponding to the \textsc{bb} model~\eqref{eq:EFPF_BB}, then
\begin{equation*}
K_n \mid N \sim \textup{Binomial} \left(N, p_n(\theta,\alpha) \right), \qquad K_{m}^{(n)} \mid  \bm{Z}^{(n)}, N \sim \textup{Binomial} \left(N-k, p_m(\theta + n, \alpha) \right).
\end{equation*}
If instead $N \sim \textup{Poisson} (\lambda)$, corresponding to model~\eqref{eq:EFPF_BB_Poiss}, then
\begin{equation*}
K_n \sim \textup{Poisson}\left(\lambda \: p_n(\theta, \alpha)\right), \qquad K_{m}^{(n)}\mid  \bm{Z}^{(n)} \sim \textup{Poisson} \left( \lambda p_m(\theta + n, \alpha)\{1 - p_n(\theta, \alpha)\}\right).
\end{equation*}
Finally, if $N \sim \textup{NegBinomial} (n_0, \mu_0)$, corresponding to model~\eqref{eq:EFPF_BB_NB}, then
\begin{equation*}
\begin{aligned}
 K_n &\sim \textup{NegBinomial}\left( n_0, \mu_0 \: p_n(\theta, \alpha)\right), \\
K_{m}^{(n)} \mid  \bm{Z}^{(n)} &\sim \textup{NegBinomial}\left(n_0 + k,  \frac{n_0 + k}{n_0 / \mu_0 + p_n(\theta, \alpha)}p_m(\theta + n, \alpha)\{1 - p_n(\theta, \alpha)\} \right).
\end{aligned}
\end{equation*}
\end{proposition}

 Proposition~\ref{prop_pred_poiss_nb} is the first result of this kind for feature allocation models with finitely many features. Moreover, it underscores the high degree of interpretability and transparency in Gibbs-type feature allocation models. Specifically, when $N$ is deterministic the prior expectation for  $\mathds{E}(K_n \mid N) = N p_n(\theta,\alpha)$ depends on $p_n(\theta,\alpha) = 1-(\theta+\alpha)_n / (\theta)_n$. This probability may be understood as the expected fraction of features observed in a sample of size $n$, out of a pool of $N$ possible features. This elegant interpretation holds true also for the Poisson and negative binomial cases, as $\mathds{E}(K_n) = \lambda p_n(\theta,\alpha)$ and $\mathds{E}(K_n) = \mu_0 p_n(\theta,\alpha)$, respectively, so that $p_n(\theta,\alpha)$ can be read as the expected fraction of observed features out of the expected total number of features $\lambda$ and $\mu_0$. The conditional distributions for $K_m^{(n)}$ may also be expressed in terms of the probability $1 - p_n(\theta,\alpha)$, accounting for the old features, and the ``updated'' probability $p_m(\theta + n, \alpha)$, representing the future sample. If $N$ is deterministic, the Bayesian estimator for the number of unseen features is $\mathds{E}(K_{m}^{(n)} \mid  \bm{Z}^{(n)}, N) = (N - k)p_m(\theta + n, \alpha)$, in which $p_m(\theta + n, \alpha)$ is the expected fraction of features in a future sample of size $m$, out of the remaining $N - k$  features. In the Poisson and negative binomial cases, where the total number of features $N$ is learned from the data, the predictive mechanism is more sophisticated. The Bayesian estimators for the number of hitherto unseen features, under a Poisson and negative binomial prior for $N$, are
\begin{equation*}
\mathds{E}(K_{m}^{(n)} \mid  \bm{Z}^{(n)}) = 
\begin{cases}
\lambda\, p_m(\theta + n, \alpha)\{1 - p_n(\theta, \alpha)\}, & \text{(Poisson mixture)}, \\
\frac{n_0 + k}{n_0 / \mu_0 + p_n(\theta, \alpha)}\, p_m(\theta + n, \alpha)\{1 - p_n(\theta, \alpha)\}, & \text{(negative binomial mixture)}.
\end{cases}
\end{equation*}
Note that $\mathds{E}(\lambda \mid \bm{Z}^{(n)}) = (n_0 + k)/\{n_0 / \mu_0 + p_n(\theta, \alpha)\}$ is the posterior expectation of $\lambda$ in the Poisson-gamma representation of the negative binomial, where $\mu_0 / n_0$ denotes the overdispersion. It is important to emphasize how the sampling information $\bm{Z}^{(n)}$ affects the distribution of the statistic $K_{m}^{(n)}$. In the Poisson mixture, the distribution of $K_{m}^{(n)}$ depends on the initial sample $\bm{Z}^{(n)}$ only through the sample size $n$. Conversely, under the negative binomial mixture, $K_{m}^{(n)}$ also depends on the number of distinct features $K_n = k$ observed in the initial sample. Lastly, Proposition \ref{prop_pred_poiss_nb} offers a key motivation for adopting the proposed mixtures of \textsc{bb}s over the standard \textsc{bb} model. In the latter, the uncertainty around $K_m^{(n)}$ monotonically decreases for large values of $m$, and in the limit $K_m^{(n)}$ degenerates to a point mass at $N - k$. This eventual shrinkage of uncertainty is a very undesirable behavior. In contrast, under both the Poisson and negative binomial mixtures of \textsc{bb}s, the variance of $K_m^{(n)}$ monotonically increases with $m$, yielding a more realistic representation of uncertainty. 

Finally, we study the total number of features $N$, i.e., the richness, which coincides with the $\alpha$-diversity, since $K_n \dconv N$ as $n\rightarrow \infty$. The posterior distribution of the richness is one of the main quantities of interest in ecology and, as outlined in Proposition~\ref{prop_asy_anyprior}, it may be equivalently obtained by extrapolating the accumulation curve, specifically by considering $\lim_{m\rightarrow \infty} K_m^{(n)} + k \mid \bm{Z}^{(n)}$.  For the \textsc{bb} model, the posterior distribution of $N$ is deterministic, since $N$ is known a priori. This yields critical issues as highlighted in simulation study A of the supplementary material (Figures \ref{fig:richness_points_simulation_A} and \ref{fig:richness_distr_simulation_A}).  For the proposed mixtures of \textsc{bb}s, the following result can be easily proved using Proposition \ref{prop_pred_poiss_nb} and noting that $\lim_{m \rightarrow \infty} p_m(\theta + n, \alpha) = 1$.

\begin{proposition}\label{prop:asymp} Suppose the \textsc{efpf} is in product form~\eqref{eq:EFPF_product} with $\alpha < 0$. Then  $K_m^{(n)}\mid \bm{Z}^{(n)} \dconv N'$ with $N' + k \dequal N \mid \bm{Z}^{(n)}$, as $m \rightarrow \infty$. Let $p_n(\theta, \alpha) = 1-(\theta+\alpha)_n / (\theta)_n$, if $N \sim \textup{Poisson} (\lambda)$, then  
\begin{equation*}
N' \sim \textup{Poisson}\left(\lambda (1 - p_n(\theta, \alpha) \right),
 \end{equation*}
 whereas if $N \sim \textup{NegBinomial} (n_0, \mu_0)$, then
    \begin{equation} \label{eq:posterior_Nprime_NB}
       N' \sim \textup{NegBinomial}\left(n_0 + k,  \frac{n_0 + k}{n_0 / \mu_0 + p_n(\theta, \alpha)}\{1 - p_n(\theta, \alpha)\} \right).
    \end{equation}
\end{proposition}

Note that, as before, the results have an appealing interpretation in terms of the probability $1 - p_n(\theta,\alpha)$. The Bayesian estimators for the richness, under a Poisson and negative binomial prior for $N$, are the posterior expectations
\begin{equation*}
\mathds{E}(N \mid  \bm{Z}^{(n)}) =
\begin{cases}
k + \lambda \{1 - p_n(\theta, \alpha)\},  & \text{(Poisson mixture),}\\
k + \frac{n_0 + k}{n_0 / \mu_0 + p_n(\theta, \alpha)}\{1 - p_n(\theta, \alpha)\}, &\text{(negative binomial mixture).}    
\end{cases}
\end{equation*}
To make these formulas operative, one needs to specify values for $\lambda$, $\theta$, and $\alpha$ and possibly for $n_0$ and $\mu_0$. We remark that all these quantities have a very transparent interpretation. Therefore, their elicitation may be based on prior information in many applied contexts. In our numerical studies, we will propose an empirical Bayes approach to set the hyperparameters. Moreover, in the supplementary material we investigate a fully Bayesian procedure in which we specify suitable priors for the hyperparameters.

\subsection{Hierarchical formulation for mixtures of beta Bernoulli models}\label{sec:hierarchical_BB}

We now specialize the general posterior characterization in Theorem~\ref{thm:post_proc} for the Poisson and negative binomial mixtures of \textsc{bb} models. Recall that when $\alpha < 0$ the underlying random measure $\tilde{\mu}$ for any Gibbs-type feature model in the hierarchical representation \eqref{eq:model_generative} can be described as $\mutilde \mid N = \sumj \tilde{q}_j \delta_{\tilde{X}_j}$, with $\tilde{q}_j \iid \dbeta(-\alpha,\theta+\alpha)$ and $\tilde{X}_j \iid P_0$, for some prior distribution $N \sim P_N$. When $N$ follows a Poisson distribution, this can be compactly expressed in terms of completely random measures. Specifically, in the the Poisson mixture of \textsc{bb}s, the statistical model which induces the \textsc{efpf} in equation \eqref{eq:EFPF_BB_Poiss} may be written as 
\begin{equation}\label{model:bb_poisson_process}
\begin{aligned}
    Z_i \mid \mutilde &\iid \BP(\mutilde), \qquad i \ge 1\\
    \mutilde & \sim \textsc{crm} (\rho ; P_0),
\end{aligned}
\end{equation}
where the intensity measure $\rho(\D s)$ is finite since $\alpha < 0$ and is proportional to the density of a $\text{Beta}(-\alpha, \theta + \alpha)$ distribution $\rho (\D s) = \lambda \Gamma(\theta)/\{\Gamma(-\alpha)\Gamma(\theta+\alpha)\} s^{-\alpha-1} (1-s)^{\theta+\alpha-1} \D s$. This result follows by a construction of completely random measures with finitely many jumps $\tilde{q}_j$, whose number has a Poisson distribution \citep{Dal08}. As for the negative binomial mixture, the \textsc{efpf} \eqref{eq:EFPF_BB_NB} is associated with the following statistical model involving negative binomial processes:
\begin{equation}\label{model:bb_nb_process}
\begin{aligned}
    Z_i \mid \mutilde &\iid \BP(\mutilde), \qquad i \ge 1,\\
    \mutilde & \sim \NB (n_0, \rho ; P_0),
\end{aligned}
\end{equation}
where $\rho(\D s) = (\mu_0 / n_0) \Gamma(\theta)/\{\Gamma(-\alpha)\Gamma(\theta+\alpha)\} s^{-\alpha-1} (1-s)^{\theta+\alpha-1} \D s$. The proof of these results is straightforward, but it is provided in Section \ref{app:proofs_mixtures_BB} for completeness. The following corollary is a consequence of Theorem~\ref{thm:post_proc}, and it characterizes the posterior distribution of $\mutilde$ in both cases.
\begin{corollary}\label{thm:post_process_bb_mixture}
Suppose $\bm{Z}^{(n)} = (Z_1,\dots,Z_n)$ follows models~\eqref{model:bb_poisson_process} or \eqref{model:bb_nb_process}, then the posterior distribution of $\mutilde$, given $\bm{Z}^{(n)}$, satisfies the decomposition $\mutilde\mid\bm{Z}^{(n)} \dequal \mutilde' + \mutilde^*$ in \eqref{eq:post_proc_dist_eq}, where $\mutilde'$ and $\mutilde^* $ are independent random measures such that $\mutilde^* $ is distributed as in Theorem~\ref{thm:post_proc}. Under model \eqref{model:bb_poisson_process} $\tilde \mu'  \sim \textsc{crm}(\rho' ; P_0)$, with updated intensity $\rho'(\D s) = \lambda \Gamma(\theta)/\{\Gamma(-\alpha)\Gamma(\theta+\alpha)\} s^{-\alpha-1} (1-s)^{\theta+ n +\alpha-1} \D s$, whereas under model \eqref{model:bb_nb_process} $\tilde \mu'  \sim \textsc{nb} (n_0 + k, \rho' ; P_0)$, with updated intensity $\rho'(\D s) = 1/\{n_0 / \mu_0 + p_n(\theta, \alpha)\} \Gamma(\theta)/\{\Gamma(-\alpha)\Gamma(\theta+\alpha)\} s^{-\alpha-1} (1-s)^{\theta+ n +\alpha-1} \D s$.
\end{corollary}

The relevance of this corollary is mostly theoretical. In practice, to simulate a posterior value for $\tilde{\mu}$ one relies on the hierarchical representation of Theorem~\ref{thm:post_proc}, given the availability of the posterior for $N$ given in Proposition~\ref{prop:asymp}. However, this posterior result unveils the central role of abstract constructions, such as completely random measures and negative binomial processes, even for allocation models with finitely many features. For example, it reveals that the lack of dependence on $K_n = k$ in the predictive structure of $K_m^{(n)}$, in the Poisson mixture model, follows because $\tilde{\mu}$ is distributed as \textsc{crm}s, as explained in \citet{Jam17} and \citet{Cam24}.

\section{Model fitting and simulation studies} \label{sec:comparison_models}

\subsection{Elicitation of the hyperparameters}\label{sec:parameter_elic}

We describe here an empirical Bayes approach to selecting the hyperparameters, albeit other Bayesian strategies could be considered. As for the two mixtures of \textsc{bb}s, we suggest to maximize the \textsc{efpf} \eqref{eq:EFPF_product} of a  \textsc{bb} model, obtained with $V_{n,k}$'s as in \eqref{eq:EFPF_BB}. This procedure provides us with an estimate $\hat{\alpha}$ and $\hat{\theta}$ of the values of $\alpha$ and $\theta$, together with an estimate $\hat{N}$ for the total number of features $N$. Then, in the Poisson mixture we set $\lambda = \mathds{E}(N) = \hat{N}$, whereas we set $\mu_0 = \mathds{E}(N) = \hat{N}$ in the negative binomial mixture of \textsc{bb}s. We remark that, differently from the Poisson mixture, the negative binomial mixture also allows us to specify the variance of the prior distribution on $N$. We argue that such variance should just reflect the practitioner's degree of uncertainty around the prior guess $\mathds{E}(N)$, possibly performing a sensitivity analysis for it. In our simulated scenarios, available in the supplementary material, we will consider additional choices of the prior expectation $\mathds{E}(N)$ for the mixtures of \textsc{bb}s, instead of specifying it through the data-driven approach just described. In such cases, we estimate the parameters $\alpha$ and $\theta$ by maximizing the model-specific \textsc{efpf}, that is \eqref{eq:EFPF_BB_Poiss} for the Poisson mixture, with $\lambda = \mathds{E}(N)$, and \eqref{eq:EFPF_BB_NB} for the negative binomial mixture, with $\mu_0 = \mathds{E}(N)$ and $n_0$ such that the desired prior variance is obtained. 

Along a similar argument, for the mixtures of \textsc{ibp}s we firstly propose to maximize the \textsc{efpf} \eqref{eq:EFPF_product}  with $V_{n,k}$'s as in \eqref{eq:EFPF_IBP} to find an estimate $\hat{\alpha}$ and $\hat{\theta}$ for the parameters $\alpha$, $\theta$, together with an estimate $\hat{\gamma}$ for the total mass $\gamma$.  Secondly, we choose the parameters of the prior for $\gamma$ by enforcing the condition $\E (\gamma)= \hat{\gamma}$. In particular, for the gamma mixture of \textsc{ibp}s, we assume $\gamma \sim \dgamma(a,b)$ and we set $\hat{\gamma} = \E(\gamma) = a/b$. Similarly to the negative binomial mixture of \textsc{bb}s, the prior variance of $\gamma$ can be specified according to the user's preferences, possibly exploring different values for robustness checks.

Finally, we remark that a fully Bayesian approach might be adopted, instead of the proposed empirical Bayes one. This consists of assuming prior distributions for parameters $\alpha$ and $\theta$ for all the mixtures. We exploit such a fully Bayesian approach in the real data analysis of  Section \ref{app:application_prior}, where we also show that posterior inferences obtained with the two procedures are coherent.

\subsection{Model-checking} \label{sec:model_check}

A preliminary step of our simulation studies and applications consists in the choice of the best model: either mixtures of \textsc{ibp}s or mixtures of \textsc{bb}s. 
The decision between the two classes pertains to the analyst. We propose two approaches to guide the selection of the best class of models: (i) a pair of visual procedures; (ii) a quantitative criterion for establishing which class of mixtures best fits the data. As for (i),
the first check relies on comparing the observed values $K_1, \ldots , K_n$ with the expected values $\E (K_1), \ldots , \E (K_n)$ under different models. 
Since the observed values  $K_1, \ldots , K_n$ refer to a particular ordering of the observations, in place of $K_1, \ldots , K_n$, we will consider the in-sample accumulation curve 
$K_1^\prime , \ldots , K_n^\prime $, obtained by averaging the number of distinct features over all possible orderings of the data. 
The second informal model check is based on the statistic $K_{n,r}$, i.e., the number of features observed with prevalence $r \geq 1$ in a sample of size $n$. To assess the  model performance, we compare the observed values $K_{n,1}, \ldots , K_{n,\bar r}$ and the expected  values $\E (K_{n,1}), \ldots ,\E ( K_{n,\bar r})$, until a certain $\bar r \leq n$, under different models' choices. While the empirical curves are always obtained from the data, the expected values depend on $\alpha$, $\theta$ and the prior mean of $N$ (resp. $\gamma$) if a mixture of \textsc{bb}s (resp. \textsc{ibp}s) is selected. As a consequence, if we adopt the  empirical Bayes approach described in Section \ref{sec:parameter_elic}  for parameters elicitation, then all the mixtures of \textsc{bb}s (resp. \textsc{ibp}s) have the same rarefaction curve and the same curve $\E(K_{n,r})$, for $r=1,\ldots,n$. By visual inspections, the previous model checks provide an indication of whether the mixtures of \textsc{bb}s or the mixtures of \textsc{ibp}s may be appropriate for the problem at the hand, which is the ultimate goal of our model selection. 

For a quantitative comparison of the goodness-of-fit between the two classes of mixtures, we rely on the (minimum) \emph{deviances} of the \textsc{bb} model and the \textsc{ibp} model, as representatives within the two classes. Given an observed dataset $\bm Z^{(n)}$ and a model described by parameters $\bm \theta$ (referred to as hyperparameters in our Bayesian setting), the (minimum) deviance is defined as $D(\hat{\bm \theta}) = -2 \log \mathcal{L}(\bm Z^{(n)} \mid \hat{\bm \theta}) = -2 \log \pi_n(m_1,\ldots, m_k \mid \hat{\bm \theta})$, where $\mathcal{L}$ denotes the likelihood of the model, i.e., the \textsc{efpf} in our case, and $\hat{\bm \theta}$ is the maximum likelihood estimate of $\bm \theta$. 
%As discussed in Section~\ref{sec:parameter_elic}, the empirical Bayes procedures select hyperparameters to maximize the corresponding likelihoods, thereby aligning them with $\hat{\bm \theta}$.
We suggest to compute the (minimum) deviances of the \textsc{bb} and the \textsc{ibp} model, and the one  yielding the smaller deviance is selected. Notably, since both models involve the same number of hyperparameters, comparing deviances yields the same conclusions as comparisons based on standard model selection criteria such as the Akaike information criterion (\textsc{aic}) and the Bayesian information criterion (\textsc{bic}).  

In general, we expect the conclusions drawn from the visual inspections of the curves and the quantitative information criteria to be consistent. We will show that the two criteria consistently support the same model selection decisions across all our simulation studies and real data analyses.

\subsection{Overview of simulation studies} \label{sec:simulation_study}

The goal of the simulation studies  is to showcase the prediction abilities of our models under different experiments. We stress that the class of Gibbs-type feature models with finitely many features also allows to perform inference on the total number of features $N$, corresponding to the $\alpha$-diversity. Specifically, in ecological applications, such quantity is referred to as taxon richness, which is a natural measure of biodiversity, as we will highlight in applications.

In Section \ref{app:simulation}, we extensively discuss three main simulation studies (A, B and C) to test the performance of our models. 
For each simulated dataset, we fix the parameters of the models via the empirical Bayes approach described in Section \ref{sec:parameter_elic}. As a second step, we apply the model-checking approaches we presented in Section \ref{sec:model_check} to conclude that either the mixtures of \textsc{bb}s or the mixtures of \textsc{ibp}s may be assumed to be correctly specified for the data at the hand. Experiments A and C correspond to situations where our model-checking clearly indicates that mixtures of \textsc{bb}s can be assumed to be correctly specified. Consequently, the assumption of finite species richness is plausible. Conversely, in experiment B, mixtures of \textsc{ibp}s best fit the data according to our model-checking procedures.
In all the cases, we focus on the prediction of the number of unseen features in an additional sample of size $m$. We present the predictions obtained via the selected models and compare them with a variation of the Good--Toulmin estimators (\textsc{gt}) from \cite{Chak19}. Additionally, in {experiments} A and C, we also compare with the well-known frequentist estimator from \cite{Cha14}, that is specifically designed under the assumption of finite species richness.  Our simulation studies indicate that the estimator of  \cite{Chak19} exhibits poor predictive performance and stability issues as $m$ grows. In general, our models show good predictive abilities, often outperforming the competitors.
In {experiments} A and C, we also address the estimation of the species richness $N$ and its uncertainty quantification. In this regard, our experiments highlight that the negative binomial mixture of \textsc{bb}s is usually more robust than the Poisson mixture under bad prior guesses on $N$.

\section{Assessing diversity in ecological applications}\label{sec:real_data}

The quantification of biological diversity is a central aspect of many ecological studies and an active research focus of ecology, due to its importance in many conservation strategies, in monitoring and management projects \citep{Cha14}. The most commonly employed and basic metric for biodiversity in a community is undoubtedly the species richness, namely the total number of species in the assemblage. Besides, another insightful characterization of the biological diversity of the assemblage may be provided by the asymptotic growth rate of the extrapolation curve $K_m^{(n)} + k \mid \bm{Z}^{(n)}$, for $m =1,2,\ldots$, and the associated $\alpha$-diversity (refer to Proposition \ref{prop_asy_anyprior} of Section \ref{sec:alpha_diversity}). In addition, these kinds of extrapolation problems are commonly faced in order to assess whether it is worth investing additional resources in looking for possibly new species. Specifically, ecologists may be interested in how many new species they are going to observe if they sample a number $m$ of additional plots. Based on such estimation, they might decide not to further analyze additional plots in the region if they expect to record a number of new species that are not worth the additional resources they are required to invest. Such information is naturally and straightforwardly available within our Bayesian framework, described by the posterior distribution of the statistic $K_m^{(n)}$, for $m \geq 1$.

Here, we illustrate how we address the aforementioned ecological research questions for two real-world datasets, which present different structural characteristics. We discuss the adequacy of the Poisson and negative binomial mixtures of \textsc{bb}s and the gamma mixture of \textsc{ibp}s, where the parameters are estimated via the empirical Bayes approach described in Section \ref{sec:parameter_elic}. Prediction and inference are then faced using the most appropriate model, selected through the model-checking described in Section \ref{sec:model_check}. {For a more exhaustive assessment of the models' predictive ability, we perform a data-holdout experiment in Section \ref{app:additional_plots}, where all the models are trained on half of the observed data, and predictions on the withheld data are compared. These analyses further support the decisions on model selection obtained through the two proposed procedures.}
In Section \ref{app:application_prior}, we also report posterior inferences obtained when 
we adopt a fully Bayesian approach for parameters' elicitation, showing that it leads to similar results obtained via the empirical Bayes procedure described in Section \ref{sec:parameter_elic}.

\subsection{Vascular plants in Danish forest}\label{sec:vascular_plants}

We consider the data collected in \cite{Mazz2016} concerning the forest of Lille Vildmose nature reserve in Denmark. Here, for each of the 102 forest plots object of the 2013 monitoring campaign, the species incidence (presence-absence) for four organism groups, i.e., epiphytic bryophytes, epiphytic lichens, vascular plants, and wood-inhabiting fungi, are measured. 
For the purpose of illustration, we focus on vascular plants, also analyzed in \cite{Mazz_paper2016}, where $k = 215$ distinct species are recorded on the $n=102$ plots. In Figure \ref{fig:accumulation_plants} of the supplementary material, we also report the taxon accumulation curve, which has clearly not yet reached convergence, thus the richness is certainly expected to be larger than the $215$ observed species.

In order to assess whether mixtures of \textsc{bb}s (finite species richness) or mixtures of \textsc{ibp}s (infinite species richness) are more appropriate in this context, we rely on {both the} visual inspections of the model-checking tools we introduced in Section \ref{sec:model_check} {and the quantitative assessment through the comparison of deviances}. From the plots of Figure \ref{fig:rare_knr_plants}, we argue that the mixtures of \textsc{ibp}s, which assume infinitely many features, are plausible models for such data, and are definitely more suitable than mixtures of \textsc{bb}s. {This claim is further supported by the comparison of deviances, with the \textsc{bb} model yielding $D(\hat{\bm \theta}) = 10320.1$ compared to  $D(\hat{\bm \theta}) = 10312.4$ for the \textsc{ibp} model. }
Therefore we focus on the gamma mixture of \textsc{ibp}s, and we consider two possible prior variances for $\gamma$, i.e., {$\Var(\gamma) \in \{1, 100\}$}. From Proposition \ref{prop_asy_km_gamma_3ibp}, the asymptotic growth rate of the curve $K_m^{(n)} + k \mid \bm{Z}^{(n)}$, for $m = 1,2,\ldots$, is of order $m^\alpha$, with an estimated rate of $\hat{\alpha} = 0.17$, and the posterior $\alpha$-diversity $S_\alpha'$ is gamma distributed; the expected value of $S_\alpha'$ equals $186.48$ for both the choices of the prior variance of $\gamma$, but we get {$\Var(S_\alpha') = 58.3$ if $\Var(\gamma) = 1$}, and $\Var(S_\alpha') = 158.9$ if $\Var(\gamma) = 100$.

\begin{figure}[tbp]
    \centering    
    \includegraphics[width = 0.49\linewidth]{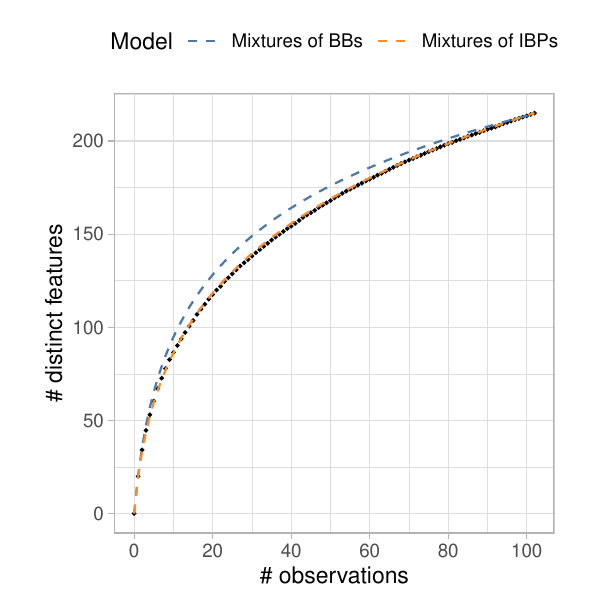}    
    \includegraphics[width = 0.49\linewidth]{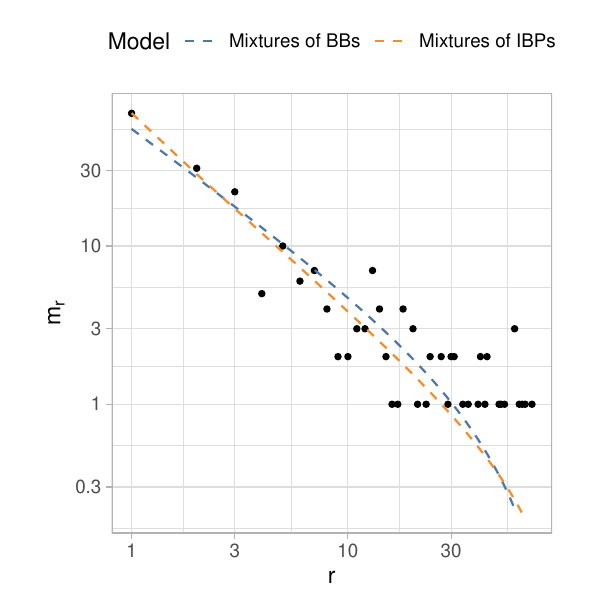}
    \caption{Left panel: the empirical accumulation curve (black dots) and the rarefaction curve of the models (blue and orange dashed lines). Right panel: the observed values of $K_{n,r}$ (black dots) compared with the {expected curve $\E(K_{n,r})$ of the models (blue and orange dashed lines). The right plot is in log-log scale; the orange (resp. blue) curves are identical for all the mixtures of \textsc{ibp}s (resp. \textsc{bb}s).}}
    \label{fig:rare_knr_plants}
\end{figure}

The extrapolation problem is addressed in Figure \ref{fig:extr_plants}, where we report the expected values and the $95\%$ credible intervals for the total number of species that might be observed in $m$ additional plots, given the observed collection of $n$ plots, i.e., $K_m^{(n)} + k \mid \bm{Z}^{(n)}$, with $k = 215$. The posterior point estimates are similar for the two selected prior variances, while the variability increases as the prior variance of  $\gamma$ increases.

\begin{figure}[!htb]
    \centering    
    \includegraphics[width = 0.49\linewidth]{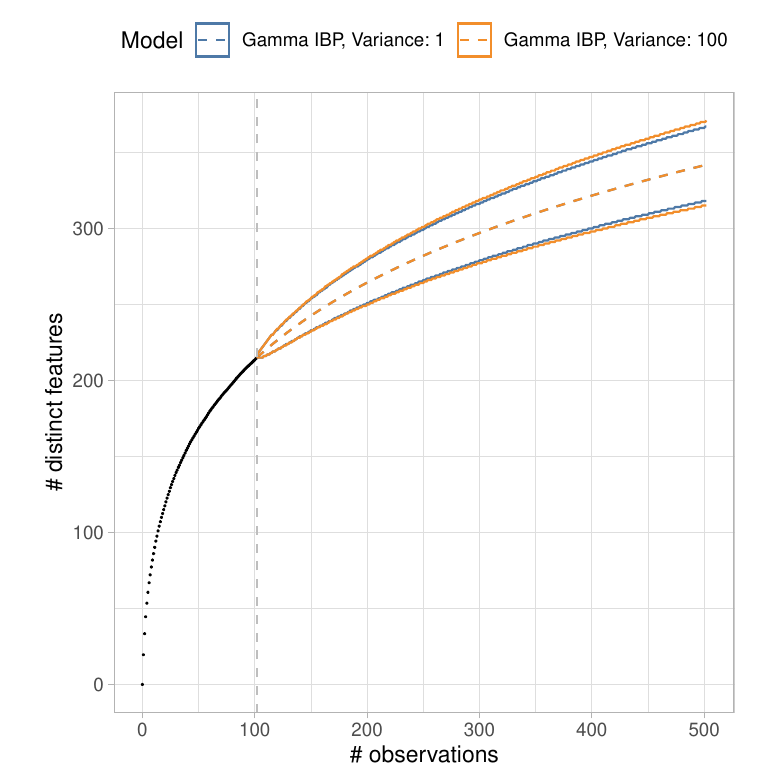}  
    
    \caption{Expected values and the $95\%$ credible intervals for $K_m^{(n)} + k \mid \bm{Z}^{(n)}$, with $k = 215$, for the gamma mixture of \textsc{ibp}s. The extrapolation horizon is $m=1,\ldots, 400$.}
    \label{fig:extr_plants}
\end{figure}

To provide a more quantitative answer to the extrapolation problem ecologists might be interested in, we report in Table \ref{table:K_mn_plants} the expected values and the credible intervals for the number of new species that are going to be observed if a number $m$ of additional plots is examined, for some values of $m$. Both the choices for the prior variance of $\gamma$ in the gamma mixture of \textsc{ibp}s lead to the same point-wise estimates: if ecologists are considering whether to analyze additional plots in the region, they should expect to find $6.50$ new species if they sample additional $m=10$ plots. Differences between the two gamma mixtures are visible for $m \in \{100, 1000\}$ in terms of their credible intervals.

\begin{table}[tbp]
\centering
\caption{Expected values $\mathds{E}(K_m^{(n)} \mid \bm{Z}^{(n)})$ and $95\%$ credible intervals (in brackets) for the statistic $K_m^{(n)} \mid \bm{Z}^{(n)}$, for $m \in \{1, 10, 100, 1000\}$, for the vascular plants in \cite{Mazz2016}.}
\begin{small}
\begin{tabular}{lrrrr}
\toprule
Gamma mixture of \textsc{ibp}s $(n = 102, k = 215)$ & $m = 1$ & $m = 10$ & $m = 100$  & $m 
= 1000$ \\ \midrule
Prior variance {$\text{Var}(\gamma) = 1$} & $0.673$ $[0,3]$  & $6.50$ $[2,12]$ & $50.1$ {$[36,65]$}  & $204$ {$[172, 237]$}    \\ 
Prior variance $\text{Var}(\gamma) = 100$  & $0.673$ $[0,3]$  & $6.50$ $[2,12]$ & $50.1$ $[35,66]$ & $204$ $[166, 244]$ \\ 
\bottomrule
\end{tabular}
\label{table:K_mn_plants}
\end{small}
\end{table}

\subsection{Trees in Barro Colorado Island
}\label{sec:bci_data}

As a second illustration, we analyze the presence-absence dataset of tree species in $n=50$ plots of one hectare in Barro Colorado Island, for a total of $k = 225$ observed species. The data are publicly available in the \textsc{vegan} package in \textsc{R}. In terms of richness estimation, exploring the taxon accumulation curve, reported in Figure \ref{fig:accumulation_bci} of the supplementary material, we may argue that it has not reached convergence yet, though the growth is rather slow. Overall, the species richness is expected to be larger than the $225$ observed species.

In order to select which model best fits the observed data, we perform the usual model-checking we introduced in Section \ref{sec:model_check}. {The visual inspection of Figure \ref{fig:rare_knr_bci} suggests that} the mixtures of \textsc{bb}s can be considered correctly specified for such data, while the mixtures of \textsc{ibp}s are not. {This preference for the mixtures of \textsc{bb}s is further supported by the comparison of deviances: $D(\hat{\bm \theta}) = 10245.6$ for the \textsc{bb} model, compared to $D(\hat{\bm \theta}) = 10266.9$ for the \textsc{ibp} model. }
Differently from the vascular plant data analyzed in the previous section, we thus claim that it is reasonable to assume that the species richness is finite. Hence we focus on the Poisson and negative binomial mixtures of \textsc{bb}s,  as for the latter, we analyze two choices for the prior variance of $N$, i.e., $\Var(N) = \mu_0 \times c$, for {$c \in \{10, 100\}$}.
\begin{figure}[tbp]
    \centering    
    \includegraphics[width = 0.49\linewidth]{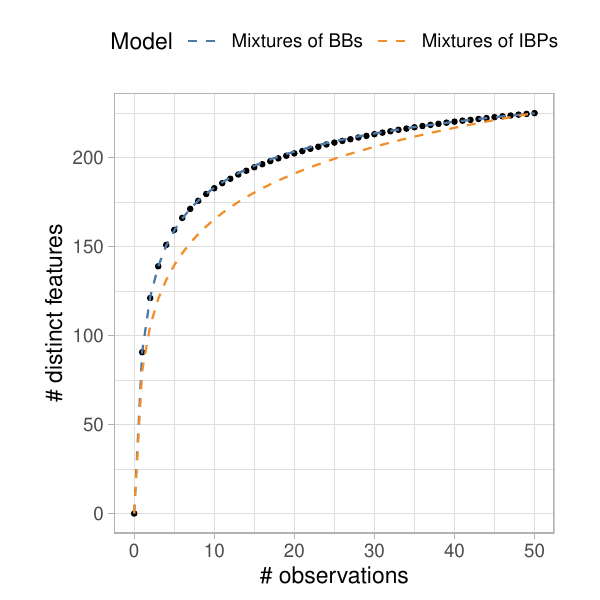}    
    \includegraphics[width = 0.49\linewidth]{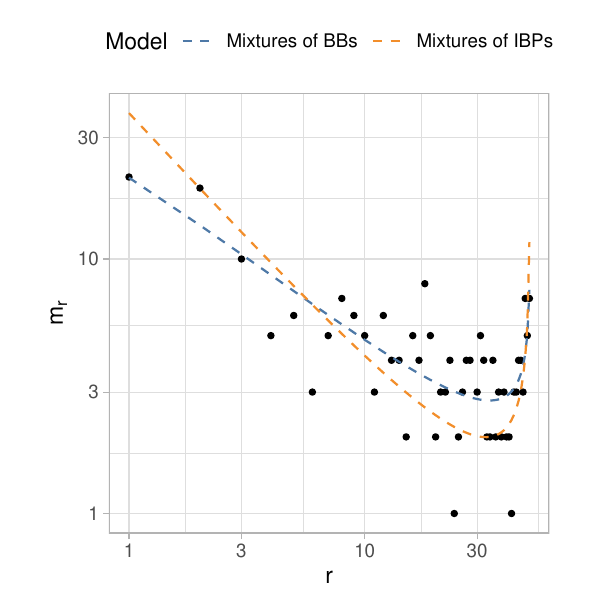}    
    \caption{Left panel: the empirical accumulation curve (black dots) and the rarefaction curve of the models (blue and orange dashed lines). Right panel: the observed values of $K_{n,r}$ (black dots) compared with the {expected curve $\E(K_{n,r})$ of the models (blue and orange dashed lines). The right plot is in log-log scale; the orange (resp. blue) curves are identical for all the mixtures of \textsc{ibp}s (resp. \textsc{bb}s).} }
    \label{fig:rare_knr_bci}
\end{figure}
In such contexts, the species richness represents the most natural measure of biodiversity of the assemblage, therefore there is interest in estimating it. In the left panel of Figure \ref{fig:bci_posterior_extr}, we report the posterior distribution of the species richness $N$, for the different mixtures of \textsc{bb}s. Specifically, the expected species richness is equal to $296.13$ for the Poisson mixture of \textsc{bb}s, with a credible interval equal to $[280, 313]$. For both the negative binomial mixtures, we get an expected species richness of $296.17$, with credible intervals~$[278, 316]$.

\begin{figure}[!htb]
    \centering    
    \includegraphics[width = 0.49\linewidth]{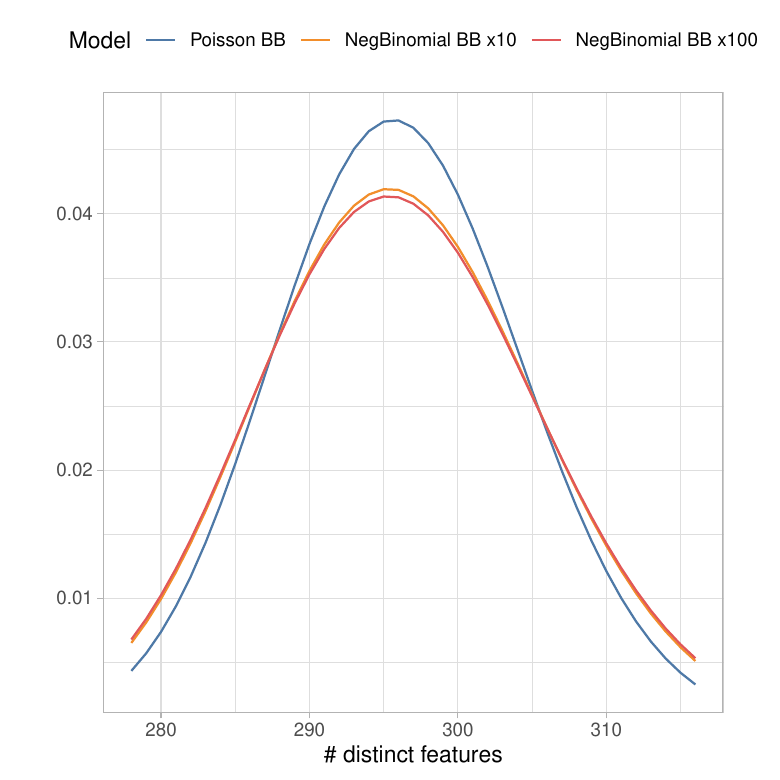}
    \includegraphics[width = 0.49\linewidth]{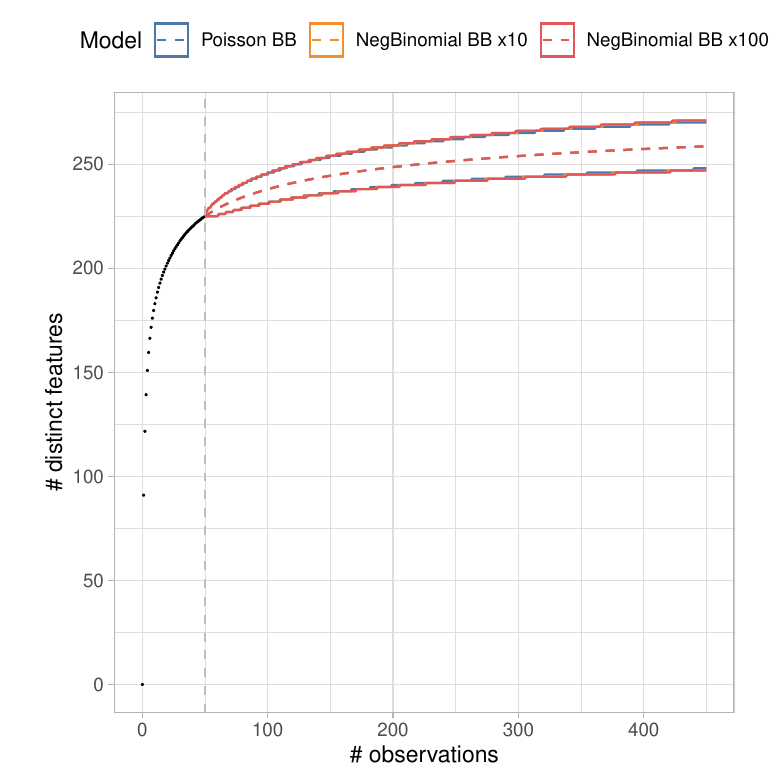} 
    \caption{Left panel: posterior distributions of the species richness $N$ for different mixtures of \textsc{bb}s. Right panel: expected values and  $95\%$ credible intervals for $K_m^{(n)} + k \mid \bm{Z}^{(n)}$, for different mixtures of \textsc{bb}s.}
    \label{fig:bci_posterior_extr}
\end{figure}

As far as the extrapolation problem is concerned, Figure \ref{fig:bci_posterior_extr} (right panel) reports the expected values and the $95\%$ credible intervals for the posteriors $K_m^{(n)} + k \mid \bm{Z}^{(n)}$, for $m=1,\ldots,400$, where $k = 225$. It can be noted that the expected number of new features grows rather slowly with the size of the additional sample $m$; moreover, from Proposition \ref{prop_asy_anyprior}, we remark that such a sequence $K_m^{(n)} + k \mid \bm{Z}^{(n)}$, $m = 1,2,\ldots$, converges to the posterior distribution of the species richness $N$. As we have just discussed, all the mixtures of \textsc{bb}s that we have fitted provide an expected richness of around  $296$.

We finally report the numerical values of expected values and the credible intervals for the number of new species that are going to be observed if a number $m$ of additional plots are sampled, for some values of $m$. All the mixtures of \textsc{bb}s fitted here provide the same point-wise estimates as well as the same credible intervals for $m \in \{1,10,100\}$: if ecologists are considering whether to analyze additional plots in the region, they should expect to find $0.41$ new species if they sample additional $m=1$ plots, $3.66$ new species for $m=10$ and $19.4$ new species for $m=100$. Moreover, the credible intervals are $[0,2]$ for $m=1$, $[0,8]$ for $m=10$ and $[11,29]$ for $m=100$. We can spot a difference among the models for $m = 1000$: the expected number of new species is $41.8$, with credible intervals equal to $[30,55]$ for the Poisson mixture and $[29,56]$ for the negative binomial mixtures.

\section{Discussion}

In the present paper, we analyzed Gibbs-type feature models, namely those models exhibiting an \textsc{efpf} in product form \eqref{eq:EFPF_product}. We argue that this class stands out among feature allocation models, similar to Gibbs-type priors, which play a fundamental role within species sampling models. We provided a comprehensive distribution theory for this class of models and a plethora of results in closed form. Additionally, we discussed two noteworthy examples: mixtures of \textsc{ibp}s and mixtures of \textsc{bb}s. While the first class assumes an infinite number of features in the population, the latter can be adopted when the total number of features is supposed to be finite. We also proposed coherent methods for parameters' elicitation and model selection.  Finally, we have emphasized the importance of our findings in addressing ecological problems, such as estimating biodiversity and quantifying species richness. The code is available online at the link: \url{https://github.com/LGhilotti/ProductFormFA}.

Our contribution could pave the way for several future research directions. First, recall that we introduced a class of Gibbs-type feature models for exchangeable observations. However, in some applied problems, data are divided into different, though related, groups and the assumption of partial exchangeability would be more appropriate. Hence, an interesting direction for future research involves defining and investigating feature allocation models in the presence of multi-sample data. Although numerous models are available for partially exchangeable data within the species setting \citep{Qui22}, work is still ongoing in the framework of feature allocation models; see, e.g., \citet{Teh10} for a few early examples. The availability of tractable classes of priors for grouped incidence data would enable the prediction for the number of shared species as well as the quantification of the so-called $\beta$-biodiversity, namely the heterogeneity among different ecological communities. Along these lines, the recent paper by \cite{stolf24} extends the \textsc{ibp} in defining a multivariate probit Indian buffet process.

Secondly, Gibbs-type are natural tools for modeling biodiversity when focusing on a single level of the Linnean taxonomy, such as \emph{family, genus} or \emph{species}. However, in modern sampling designs, each statistical unit often comprises a collection of $L$ different taxa, which are organized in a nested fashion; see, e.g., \citet{Zito23taxonomic}. As a crude exemplification, one might consider using a separate Gibbs-type feature model for each layer of the Linnean taxonomy. However, this approach would overlook the rich and informative nested structure of the data. Bayesian nonparametric models for such data are underdeveloped even for species sampling models, and there is ample room for new ideas and theoretical developments. 

Thirdly, a potentially impactful ramification of our results pertains to the usage of Gibbs-type feature models as a building block of more complex hierarchical models, i.e., when employed as a latent component. We refer to \cite{Gri11} for a general discussion. Among the potential applications of Gibbs-type feature models, it is worth mentioning their role in Bayesian factor analysis, in which $N$, in our notation, would represent the number of factors. The \textsc{ibp} has been successfully used in this context to incorporate sparsity for instance to model gene expression data \citep{Kno11}. A further example is given by \cite{Aye21}, who explored suitable extensions of the \textsc{ibp} to discover latent communities in network data. Our paper provides several alternatives to the \textsc{ibp} for Bayesian factor models and related applications. It is also worth mentioning that the mixtures of \textsc{bb}s, unlike the traditional \textsc{ibp} or the \textsc{bb}, enable the incorporation of prior opinions on the number of latent factors in Bayesian modeling. Work on these problems is deemed to be future research.

%\vspace{-1cm}
\section*{Acknowledgments}
{
The authors are extremely grateful to an Associate Editor and three anonymous referees for their constructive comments and valuable suggestions, which
led to a substantial improvement of the paper.
The authors  acknowledge the support from the Italian Ministry of University and Research (MUR), ``Dipartimenti di Eccellenza'' grant 2023-2027. 
FC and LG were supported by the European Union – Next Generation EU funds, component M4C2, investment 1.1., PRIN-PNRR 2022 (P2022H5WZ9). TR gratefully acknowledges the support from the European Research Council under the European Union’s Horizon 2020 research and innovation programme (grant agreement No 856506).
FC is a member of the Gruppo Nazionale per l’Analisi Matematica, la Probabilità e le loro Applicazioni (GNAMPA) of the Istituto Nazionale di Alta Matematica (INdAM).}

\bibliography{references}

\clearpage

\setcounter{equation}{0}
\renewcommand\theequation{S\arabic{equation}}
\renewcommand\theHequation{S\arabic{equation}}

\setcounter{theorem}{0}
\renewcommand\thetheorem{S\arabic{theorem}}
\renewcommand\theHtheorem{S\arabic{theorem}}

\setcounter{corollary}{0}
\renewcommand\thecorollary{S\arabic{corollary}}
\renewcommand\theHcorollary{S\arabic{corollary}}

\setcounter{proposition}{0}
\renewcommand\theproposition{S\arabic{proposition}}
\renewcommand\theHproposition{S\arabic{proposition}}

\setcounter{lemma}{0}
\renewcommand\thelemma{S\arabic{lemma}}
\renewcommand\theHlemma{S\arabic{lemma}}

\setcounter{figure}{0}
\renewcommand\thefigure{S\arabic{figure}}
\renewcommand\theHfigure{S\arabic{figure}}

\setcounter{table}{0}
\renewcommand\thetable{S\arabic{table}}
\renewcommand\theHtable{S\arabic{table}}

\setcounter{section}{0}
\renewcommand{\thesection}{S\arabic{section}}
\renewcommand{\theHsection}{S\arabic{section}}

\begin{center}
   \LARGE Supplementary material for:\\
    ``Bayesian analysis of product feature allocation models''
\end{center}

\section*{Organization of the supplementary material}

The supplementary material is organized as follows. In Section \ref{app:proof_general_results}, we prove all the general results of Section \ref{sec:general_theory}, valid for any Gibbs-type feature model. In addition, in Section \ref{app:msharedgeneralV} we provide general distributional results for the statistic $K_{n,r}$, denoting the number of features appearing exactly $r$ times across $n$ individuals. Section \ref{app:knr_results_spec} specializes the general distribution theory for $K_{n,r}$ under the Poisson and negative binomial mixtures of \textsc{bb}s and the gamma mixture of \textsc{ibp}s. Section \ref{app:proofs_mixtureIBPs} (resp. Section \ref{app:proofs_mixtures_BB})
contains all the proofs and details of Section \ref{sec:mixture_IBP} (resp. Section \ref{sec:BB_mixtures}). Our simulation studies are reported in Section \ref{app:simulation}. Finally, Section \ref{app:applications} includes additional plots referring to the ecological applications presented in the main paper; moreover, in Section \ref{app:application_prior}, we also propose a fully Bayesian approach for parameters' elicitation, alternative to the empirical Bayes one described in Section \ref{sec:parameter_elic}. We apply such a fully Bayesian approach to the two real data scenarios of Section \ref{sec:real_data} in comparison with the inference obtained via the empirical Bayes procedure.

\section{Proofs of Section \ref{sec:general_theory} and additional results}\label{app:proof_general_results}

%% Predictive general V
\subsection{Proof of Theorem \ref{thm:predictive_general_V}}

In order to determine the predictive distribution, we have to pay attention to the ordering of the new features. More precisely, if $K_n=k$, there are $\binom{k+y}{k}$ possible ways to order the new $Y_{n+1}=y$ observed features with respect to the first $K_n$ ordered features. By taking into account this combinatorial coefficient, the predictive law equals
\begin{equation*}
p_n(y, a_1, \ldots , a_k) =  \frac{\binom{k+y}{k}  \pi_{n+1} (m_1 +a_1, \ldots , m_k+a_k, 1, \ldots , 1)}{\pi_n (m_1, \ldots , m_k)}
\end{equation*}
where the number $1$ at the numerator is repeated $y$ times. The numerator corresponds to the probability of the feature allocation induced by the observed sample $\bm{Z}^{(n)}$ and $Z_{n+1}$, by considering all the possible orders of the $y$ newly observed features. Besides, the denominator is the probability distribution of the feature allocation induced by the sample $\bm{Z}^{(n)}$. By substituting the specific expression of the \textsc{efpf} in product form \eqref{eq:EFPF_product}, we obtain
\begin{equation*}
    p_n(y, a_1, \ldots , a_k) =  \binom{k+y}{k} \frac{V_{n+1,k+y}}{V_{n,k}} (\theta+\alpha)_n^y \cdot \prod_{\ell=1}^k \frac{(1-\alpha)_{m_\ell + a_\ell -1}}{(1-\alpha)_{m_\ell-1}} \frac{(\theta+\alpha)_{n+1-m_\ell -a_\ell}}{(\theta+\alpha)_{n-m_\ell}} .
\end{equation*}
By observing that
\[
\frac{(1-\alpha)_{m_\ell + a_\ell -1}}{(1-\alpha)_{m_\ell-1}} = (m_\ell -\alpha)^{a_\ell},\quad \frac{(\theta+\alpha)_{n+1-m_\ell -a_\ell}}{(\theta+\alpha)_{n-m_\ell}} = (\theta+\alpha+n-m_\ell)^{1-a_\ell},
\]
we get the following expression:
\begin{equation*}
    \begin{split}
     p_n(y, a_1, \ldots , a_k) & =  \binom{k+y}{k} \frac{V_{n+1,k+y}}{V_{n,k}} \{(\theta+\alpha)_n\}^y (\theta+n)^k \prod_{\ell=1}^k \left(\frac{m_\ell-\alpha}{\theta+n}\right)^{a_\ell} \cdot \left( 1- \frac{m_\ell-\alpha}{\theta+n}\right)^{1-a_\ell}\\
     & =  \binom{k+y}{k} \frac{V_{n+1,k+y}}{V_{n,k}} \{(\theta+\alpha)_n\}^y (\theta+n)^k \prod_{\ell=1}^k \Bcr\left( a_\ell; \frac{m_\ell - \alpha}{\theta +n} \right)
     \end{split}
\end{equation*}
and the thesis follows.

%% Prob not observing new features general V
\subsection{Proof of Corollary \ref{cor:notobs_new_general_V}}

The corollary easily follows from Theorem \ref{thm:predictive_general_V}. Indeed, the sample coverage equals
\begin{equation*}
    \begin{split}
       \P (Y_{n+1}=0 \mid \bm{Z}^{(n)}) &= \sum_{(a_1, \ldots , a_k ) \in \{ 0,1\}^k}  p_n(0, a_1, \ldots , a_k) 
       =  \frac{V_{n+1,k}}{V_{n,k}} (\theta+n)^k \prod_{\ell=1}^k \sum_{a_\ell =0}^1\Bcr\left( a_\ell; \frac{m_\ell - \alpha}{\theta +n} \right)\\
       & = \frac{V_{n+1,k}}{V_{n,k}} (\theta+n)^k
    \end{split}
\end{equation*}
where we have marginalized out the probability of observing old features from the predictive distribution.

%% Kn for general Vnk
\subsection{Proof of Theorem \ref{thm:Kn_general}}

We start by proving the following lemma, which provides the probability distribution of $K_n$ for any exchangeable feature allocation model admitting an \textsc{efpf}.
\begin{lemma}\label{lemma:general_Kn__on_EFPF}
For any feature allocation model admitting an \textsc{efpf}, the probability distribution of $K_n$ is
\begin{equation*}
    \P(K_n=k) = \sum_{\substack{m_\ell \in \{ 1,\ldots,n\}\\ \ell=1,\ldots,k} }\pi_n(m_1,\ldots,m_k) \binom{n}{m_1} \cdots \binom{n}{m_k} .
\end{equation*}
\end{lemma}
\begin{proof}
The event $\{K_n=k\}$ corresponds to the union of all feature allocations of the type $F_n = (B_{n,1},\ldots,B_{n,k})$, where $m_\ell = |B_{n,\ell}|$ takes any value in the set $ \{1,\ldots,n\}$, for any $ \ell =1,\ldots,k$. 
For a specific configuration of $(m_1, \ldots , m_k)$, there are several feature allocations $F_n$ such that 
$m_\ell = |B_{n,\ell}|$, as $\ell=1, \ldots , k$. Indeed, there are $\binom{n}{m_1}$ different ways to choose the indexes in the set $[n]$ to define $B_{n,1}$, similarly there are $\binom{n}{m_2}$ different ways to choose the set $B_{n,2}$, etc.. Hence, the probability of observing all feature allocations having predetermined block sizes $(m_1, \ldots , m_k)$ equals $\pi_n(m_1,\ldots,m_k) \binom{n}{m_1} \cdots \binom{n}{m_k}$. As a consequence, the probability $\P(K_n=k)$ can be obtained by summing over all the possible configurations of the vector  $(m_1,\ldots,m_k) \in \{1,\ldots,n\}^k$.
\end{proof}

We now exploit Lemma \ref{lemma:general_Kn__on_EFPF} to find the probability distribution of $K_n$ for any Gibbs-type feature allocation model:
\begin{equation} \label{eq:PKn_proof1}
\begin{aligned}
    \P(K_n=k) &= \sum_{\substack{m_\ell \in \{ 1,\ldots,n \}\\ \ell=1,\ldots,k} } V_{n,k} \prodl (1-\alpha)_{m_\ell -1} (\theta+\alpha)_{n-m_\ell}\binom{n}{m_\ell}\\
    &= V_{n,k} \prodl \sum_{m_\ell=1}^n \binom{n}{m_\ell} (1-\alpha)_{m_\ell-1} (\theta+\alpha)_{n-m_\ell}\\
    &= V_{n,k} \left[ \sum_{i=1}^n \binom{n}{i} (1-\alpha)_{i-1} (\theta+\alpha)_{n-i} \right]^k .
\end{aligned}
\end{equation}
We now consider the case $\alpha <0$ and $\alpha \in [0,1)$ separately.

\noindent
\textbf{Case $\bm{\alpha<0}$}. For these values of $\alpha$, the last sum in \eqref{eq:PKn_proof1} can be expressed as
\begin{equation*}
    \sum_{i=1}^n \binom{n}{i} (1-\alpha)_{i-1} (\theta+\alpha)_{n-i} = \frac{\Gamma(-\alpha)}{\Gamma(1-\alpha)}\sum_{i=0}^n \binom{n}{i} (-\alpha)_{i} (\theta+\alpha)_{n-i} + \frac{(\theta+\alpha)_n}{\alpha}.
\end{equation*}
Thanks to the Chu-Vandermonde identity, we obtain
\begin{equation*}
    \frac{\Gamma(-\alpha)}{\Gamma(1-\alpha)} \sum_{i=0}^n \binom{n}{i} (-\alpha)_{i} (\theta+\alpha)_{n-i} + \frac{(\theta+\alpha)_n}{\alpha} = -\frac{(\theta)_n}{\alpha}  + \frac{(\theta+\alpha)_n}{\alpha} 
\end{equation*}
and the thesis follows for $\alpha <0$.\\

\noindent
\textbf{Case $\bm{\alpha \in [0,1)}$.}
For positive values of $\alpha$, the sum in \eqref{eq:PKn_proof1} can be expressed as
\begin{equation*}
\begin{split}
    \sum_{i=1}^n \binom{n}{i} (1-\alpha)_{i-1} &(\theta+\alpha)_{n-i} = 
     \sum_{i=1}^n \binom{n}{i} \frac{\Gamma(i-\alpha)}{\Ga(1-\alpha)} \frac{\Ga(\theta+\alpha+n-i)}{\Ga(\theta+\alpha)} \cdot \frac{\Ga(\theta+n)}{\Ga(\theta+n)} \\
     &= \frac{\Ga(\theta+n)}{\Ga(1-\alpha) \Ga(\theta+\alpha) } \sum_{i=1}^n \binom{n}{i} B(i-\alpha, \theta+\alpha+n-i),
     \end{split}
\end{equation*}
where $B(a,b)$ denotes the Euler's beta function evaluated at $a,b>0$. 
Thanks to the integral representation of the beta function, one obtains:
\begin{equation*}
\begin{split}
  \sum_{i=1}^n \binom{n}{i} (1-\alpha)_{i-1} &(\theta+\alpha)_{n-i} = \frac{\Ga(\theta+n)}{\Ga(1-\alpha) \Ga(\theta+\alpha) } \sum_{i=1}^n \binom{n}{i} \int_0^1 x^{i-\alpha-1} (1-x)^{\theta+\alpha+n-i-1} \D x\\
     &= \frac{\Ga(\theta+n)}{\Ga(1-\alpha) \Ga(\theta+\alpha) }  \int_0^1 x^{-\alpha-1} (1-x)^{\theta+\alpha+n-1} \sum_{i=1}^n \binom{n}{i} \left(\frac{x}{1-x}\right)^i  \D x\\
     &= \frac{\Ga(\theta+n)}{\Ga(1-\alpha) \Ga(\theta+\alpha) }  \int_0^1 x^{-\alpha-1} (1-x)^{\theta+\alpha+n-1} \left[ \left(\frac{x}{1-x} + 1\right)^n -1\right]  \D x\\
     &= \frac{\Ga(\theta+n)}{\Ga(1-\alpha) \Ga(\theta+\alpha) }  \int_0^1 x^{-\alpha-1} (1-x)^{\theta+\alpha+n-1} \left[ \left(\frac{1}{1-x}\right)^n -1\right]  \D x\\
     &= \frac{\Ga(\theta+n)}{\Ga(1-\alpha) \Ga(\theta+\alpha) }  \int_0^1 x^{-\alpha-1} (1-x)^{\theta+\alpha-1} \left[ 1-(1-x)^n \right]  \D x     .
\end{split}
\end{equation*}
By virtue of the following equality 
\[
1-(1-x)^n = x\sum_{i=0}^{n-1} (1-x)^i,
\]
we obtain
\begin{equation*}
\begin{aligned}
    \frac{\Ga(\theta+n)}{\Ga(1-\alpha) \Ga(\theta+\alpha) }  &\int_0^1 x^{-\alpha-1} (1-x)^{\theta+\alpha-1} \left[ 1-(1-x)^n \right]  \D x  \\
    &= \frac{\Ga(\theta+n)}{\Ga(1-\alpha) \Ga(\theta+\alpha) }  \int_0^1 x^{-\alpha+1 -1} (1-x)^{\theta+\alpha-1} \sum_{i=0}^{n-1} (1-x)^i  \D x  \\
    &= \frac{\Ga(\theta+n)}{\Ga(1-\alpha) \Ga(\theta+\alpha) }  \sum_{i=0}^{n-1} \int_0^1 x^{-\alpha+1 -1} (1-x)^{\theta+\alpha+ i-1}  \D x  \\
    &= \frac{\Ga(\theta+n)}{\Ga(1-\alpha) \Ga(\theta+\alpha) }  \sum_{i=0}^{n-1} \frac{\Ga(i+\theta+\alpha) \Ga(1-\alpha)}{\Ga(i+\theta+1)}\\
    &= \frac{\Ga(\theta+n)}{\Ga(1-\alpha) \Ga(\theta+\alpha) }  \sum_{i=1}^{n} \frac{\Ga(i+\theta+\alpha-1) \Ga(1-\alpha)}{\Ga(i+\theta)}\\
    &= \frac{\Ga(\theta+n)}{ \Ga(\theta+1) }  \sum_{i=1}^{n} \frac{(\theta+\alpha)_{i-1}}{(\theta+1)_{i-1}} = (\theta+1)_{n-1} \sum_{i=1}^{n} \frac{(\theta+\alpha)_{i-1}}{(\theta+1)_{i-1}}    
\end{aligned}
\end{equation*}
and the thesis follows by substituting the previous expression in \eqref{eq:PKn_proof1}.

%% Number of m*-shared features general V
\subsection{Distributional results for the number of $r$-shared features $K_{n,r}$}
\label{app:msharedgeneralV}

Here we provide distributional results for the statistic $K_{n,r}$, where $r \in \{1,\ldots,n\}$, denoting the number of features appearing exactly $r$ times among $n$ individuals. 
We start by proving the following lemma, which provides the probability distribution of $K_{n,r}$ for any feature allocation model admitting an \textsc{efpf}. 
\begin{lemma}\label{lemma:sample_cov_on_EFPF}
For any feature allocation model admitting an \textsc{efpf}, the distribution of $K_{n,r}$ equals 
\begin{equation} \label{eq:Mnr_general}
\begin{split}
  &  \P(K_{n,r} = y )\\
  & \qquad = \binom{n}{r}^y \Bigg\{ \pi_n(r,\ldots,r) + \sum_{k=y+1}^\infty  \binom{k}{y} \sum_{\substack{m_\ell \in \{ 1,\ldots,n\} \\ m_\ell \neq r\\
    \ell=1,\ldots,k-y}} \pi_n(r,\ldots,r,m_1,\ldots,m_{k-y}) \prod_{\ell=1}^{k-y}\binom{n}{m_\ell} \Bigg\}
    \end{split}
\end{equation}
where $(r,\ldots,r)$ in the previous expression is a vector of length $y$.
\end{lemma}
\begin{proof}
The event $\{K_{n,r} = y \}$ corresponds to all random ordered feature allocations of the type $F_n = (B_{n,1},\ldots,B_{n,k})$,  where $k\geq y$ and there are exactly $y$ sets out of the $B_{n, \ell}$'s with cardinality $r$.\\
Consider $k > y$, and suppose that  the first $y$ sets are those with cardinality $r$, i.e., $|B_{n, \ell}|=r$, for all  ${\ell}=1,\ldots,y$, and that the remaining sets $B_{n,y+\ell}$, as $ \ell =1,\ldots,k-y$, have cardinalities $m_\ell \in \{1,\ldots,n\}$,  where $  m_\ell \neq r$. There are $\binom{n}{r}$ different ways to choose the indexes in $[n]$ to construct each of the first $y$ blocks $B_{n, \ell}$, as $\ell =1, \ldots , y$. Analogously, there are $\binom{n}{m_\ell}$ different ways to define $B_{n, y+\ell}$, as $\ell=1, \ldots , k-y$. Thus, the probability of observing all feature allocations having the first $y (< k)$ blocks with cardinality $r$ equals
\[
  \pi_n(r,\ldots,r,m_1,\ldots,m_{k-y}) \binom{n}{r}^y \binom{n}{m_1}\cdots \binom{n}{m_{k-y}} .  
\]
The previous probability refers to all the feature allocations where the first $y$ blocks have cardinality $r$, it is apparent that we have to multiply by the factor $\binom{k}{y}$ to take into account all the possible rearrangements of the $y$ blocks with cardinality $r$. When $k=y$,   the probability that all the blocks of the feature allocation have cardinality $r$ equals $\pi_n(r,\ldots,r)\binom{n}{r}^y$, where $(r,\ldots,r)$ is a vector of length $y$.
The expression of $\P(K_{n,r} = y ) $ can be obtained by summing the previous probabilities when $k \geq y$:
\begin{equation*}
\begin{aligned}
     \P(K_{n,r} = y ) &=\pi_n(r,\ldots,r)\binom{n}{r}^y + \\
     &\quad \sum_{k=y+1}^\infty   \sum_{\substack{m_\ell \in \{ 1,\ldots,n\}\\ m_\ell \neq r\\
    \ell=1,\ldots,k-y}} \binom{k}{y} \pi_n(r,\ldots,r,m_1,\ldots,m_{k-y}) \binom{n}{r}^y \prod_{\ell=1}^{k-y}\binom{n}{m_\ell}
\end{aligned}
\end{equation*}
and the thesis now follows.
\end{proof}

We now state our main result of the section by specializing the previous lemma to the class of Gibbs-type feature allocation models.
\begin{proposition}[Number of $r$-shared features $K_{n,r}$]\label{prop:msharedgeneralV}
Suppose the \textsc{efpf} is in product form \eqref{eq:EFPF_product}. Then, for any $k \geq 0$,
\begin{equation*}
  \P(K_{n,r} = k) = \left\{ \binom{n}{r} (1-\alpha)_{r-1} (\theta+\alpha)_{n-r}\right\}^k \sum_{j=0}^\infty \binom{j + k}{k} V_{n,j + k} \{t_{n,r}(\theta, \alpha)\}^k,
\end{equation*}
where, if $\alpha<0$,
\begin{equation*}
    t_{n,r}(\theta, \alpha) =         -\frac{(\theta)_n}{\alpha} + \frac{(\theta+\alpha)_n}{\alpha} - \binom{n}{r} (1-\alpha)_{r-1} (\theta+\alpha)_{n-r},
\end{equation*}
whereas, if $ \alpha \in [0,1)$,
\begin{equation*}
   t_{n,r}(\theta, \alpha) =  (\theta+1)_{n-1} \: g_n(\theta, \alpha) - \binom{n}{r} (1-\alpha)_{r-1} (\theta+\alpha)_{n-r}.
\end{equation*}
\end{proposition}
\begin{proof}
We exploit Lemma \ref{lemma:sample_cov_on_EFPF} to evaluate $\P(K_{n,r} = y )$ for the Gibbs-type feature allocation models. A plain application of Equation \eqref{eq:Mnr_general} leads to
\begin{equation*}
\begin{aligned}
    \P(K_{n,r} = y ) &= \left[ \binom{n}{r} (1-\alpha)_{r-1} (\theta+\alpha)_{n-r} \right]^y  \\
    &\qquad \times\Bigg\{ V_{n,y} +  \sum_{z=1}^\infty  \binom{z+y}{y} V_{n,z+y}  \sum_{\substack{m_\ell    \in  \{ 1,\ldots,n\}\\ m_\ell \neq r\\
    \ell=1,\ldots,k-y}}  \prod_{\ell=1}^{z}\binom{n}{m_\ell}(1-\alpha)_{m_\ell-1} (\theta+\alpha)_{n-m_\ell} \Bigg\} \\
    &= \left[ \binom{n}{r} (1-\alpha)_{r-1} (\theta+\alpha)_{n-r} \right]^y \times\\
    &\qquad \times \sum_{z=0}^\infty  \binom{z+y}{y} V_{n,z+y}  \Bigg[\sum_{\substack{i \in \{ 1,\ldots,n \}\\ i \neq r}}  \binom{n}{i}(1-\alpha)_{i-1} (\theta+\alpha)_{n-i}\Bigg]^z  \\
        &= \left[ \binom{n}{r} (1-\alpha)_{r-1} (\theta+\alpha)_{n-r} \right]^y \times\\
    &\qquad \times \sum_{z=0}^\infty  \binom{z+y}{y} V_{n,z+y}  \Bigg[\sum_{i=1}^n  \binom{n}{i}(1-\alpha)_{i-1} (\theta+\alpha)_{n-i}   - \binom{n}{r}(1-\alpha)_{r-1} (\theta+\alpha)_{n-r}   \Bigg]^z  .
\end{aligned}
\end{equation*}
We can now apply a similar strategy as in the proof of Theorem \ref{thm:Kn_general} to get the expressions for the case  $\alpha<0$ and $\alpha\in [0,1)$.
\end{proof}

%% Specialized results

\subsection{Additional specialized results for novel models}\label{app:knr_results_spec}

The general expression for the number of $r$-shared features $K_{n,r}$ presented in Proposition \ref{prop:msharedgeneralV} can be specialized in the case of gamma mixture of \textsc{ibp}s and Poisson and negative binomial mixtures of \textsc{bb}s, leading to standard and well-known probability distributions. In the following proposition, we report the specialized results for the mixtures of \textsc{bb}s.

\begin{proposition} 
Suppose the \textsc{efpf} is in product form~\eqref{eq:EFPF_product} with $\alpha < 0$ and let $b_{n,r}(\theta,\alpha) := - \alpha \binom{n}{r}  (1-\alpha)_{r-1} (\theta +\alpha)_{n-r} / (\theta)_n $. If $N$ is fixed, corresponding to the \textsc{bb} model~\eqref{eq:EFPF_BB}, then
\begin{equation*}
K_{n,r} \mid N \sim \textup{Binomial} \left(N, b_{n,r} \right).
\end{equation*}
If instead $N \sim \textup{Poisson} (\lambda)$, corresponding to model~\eqref{eq:EFPF_BB_Poiss}, then
\begin{equation*}
K_{n,r} \sim \textup{Poisson}\left(  \lambda \cdot b_{n,r} \right).
\end{equation*}
Finally, if $N \sim \textup{NegBinomial} (n_0, \mu_0)$, corresponding to model~\eqref{eq:EFPF_BB_NB}, then
\begin{equation*}
 K_{n,r} \sim \textup{NegBinomial}\left( n_0, \mu_0 \cdot b_{n,r}  \right).
\end{equation*}
\end{proposition}

For the gamma mixture of \textsc{ibp}s, the prior distribution of $K_{n,r}$ is described in the next proposition.

\begin{proposition} Suppose the \textsc{efpf} is in product form \eqref{eq:EFPF_product} with $\alpha \in [0, 1)$ and $d_{n,r} := \binom{n}{r} (1-\alpha)_{r-1} (\theta+\alpha)_{n-r} \Gamma(\theta+1)/\Gamma(\theta+n)$. If $\gamma$ is fixed, corresponding to the \textsc{ibp} model \eqref{eq:EFPF_IBP}, then 
\begin{equation*}
 K_{n,r} \mid \gamma \sim \textup{Poisson}\left( \gamma  \cdot d_{n,r} \right).
\end{equation*}
If instead $\gamma \sim \textup{Gamma}(a,b)$, then
\begin{equation*}
 K_{n,r} \sim \textup{NegBinomial}\left( a, a/b \cdot d_{n,r} \right).
\end{equation*}
\end{proposition}

%% Kmn general V
\subsection{Proof of Theorem \ref{thm:Kmn_general_V}}
We start by proving the following lemma, which provides the probability distribution of $K_m^{(n)}\mid \bm Z^{(n)}$ for any feature allocation model admitting an \textsc{efpf}. 
\begin{lemma}\label{lemma:Kmn_on_EFPF}
For any  feature allocation model admitting an \textsc{efpf}, the distribution of $K_m^{(n)}\mid \bm Z^{(n)}$ equals
\begin{equation*}
\begin{aligned}
   &  \P(K_m^{(n)} = y\mid \bm Z^{(n)}) = \binom{k+y}{k} \frac{1}{\pi_n(m_1,\ldots,m_k)} \\
     &\qquad\qquad \times \sum_{\substack{r_t \in \{1, \ldots , m\}\\ t=1,\ldots,y}} \sum_{\substack{c_\ell \in \{0, \ldots , m\}\\ \ell=1,\ldots,k}} \pi_{n+m}(m_1+c_1,\ldots,m_k+c_k,r_1,\ldots,r_y) \prod_{r=1}^y \binom{m}{r_t} \prod_{\ell=1}^k \binom{m}{c_\ell} .
\end{aligned}
\end{equation*}
\end{lemma}
\begin{proof}
We denote by $F_n = (B_{n,1}, \ldots , B_{n,K_n})$ the random ordered feature allocation corresponding to the random sample
$\bm{Z}^{(n)}$, and by $f_n$ the observed value of $F_n$.
Thanks to the  correspondence between $\bm{Z}^{(n)}$ and $F_n$, as explained by \cite{Broderick2013two}, we have
\begin{equation}
    \label{eq:Kmn_ratio}
    \P(K_m^{(n)} = y\mid \bm Z^{(n)}) = \frac{\P(K_m^{(n)} = y, F_n =f_n)}{\P(F_n=f_n )},
\end{equation}
where we observe that $\P(F_n=f_n ) = \pi_n(m_1,\ldots,m_k)$. Denote by $(B_{m+n,1}, \ldots , B_{m+n,k})$  the random ordered feature allocation which refers to the old features, and by $(B_{m+n,k+1}, \ldots , B_{m+n, k+y})$ the random ordered feature allocation corresponding to the $y$ new features. 
In order to evaluate the probability at the numerator in \eqref{eq:Kmn_ratio}, we need to evaluate the probability that the aforementioned random ordered feature allocations satisfy:
\begin{itemize}
    \item[(i)]  $(B_{m+n,1}, \ldots , B_{m+n,k})$ is consistent with the initial $F_n$, in the sense that each set 
    $B_{m+n,\ell}$ contains the $m_\ell$ indexes of the corresponding $B_{n,\ell}$, as $\ell=1, \ldots , k$;
    \item[(ii)] each set in $(B_{m+n,1}, \ldots , B_{m+n,k})$  contains $c_\ell$ indexes from $\{n+1,\ldots,n+m\}$, for all the possible choices of $c_\ell \in \{0,\ldots,m\}$, as $\ell =1, \ldots , k$;
    \item[(iii)] the sets $B_{m+n,k+1}, \ldots , B_{m+n, k+y}$, corresponding to the \textit{new} features, contain  $r_t$ indexes out of the set $\{n+1,\ldots,n+m\}$, for all the possible choices of 
    $r_t \in \{1,\ldots,m\}$, as $t=1, \ldots , y$.
\end{itemize}
We note that $B_{m+n,\ell}$ can be chosen in $\binom{m}{c_\ell}$ different ways, as $\ell =1, \ldots, k$, because this is the number of possible ways to select $c_\ell$ indexes among  $\{n+1,\ldots,n+m\}$ of the subjects displaying feature $\ell$. Besides, $B_{m+n, k+t}$ can be chosen in $\binom{m}{r_t}$ different ways, as $t=1, \ldots , y$. The probability of observing random ordered feature allocations satisfying (i)--(iii) equals
\[
\pi_{n+m}(m_1+c_1,\ldots,m_k+c_k, r_1,\ldots,r_y)\prod_{t=1}^y \binom{m}{r_t} \prod_{\ell=1}^k \binom{m}{c_\ell}.
\]
We finally observe that there are  $\binom{k+y}{y}$ possible ways for ordering the $y$ new features  among the $k$ old features, as a consequence one has:
\begin{equation*}
\begin{aligned}
    & \P(K_m^{(n)} = y\mid \bm Z^{(n)}) = \frac{1}{\pi_n(m_1,\ldots,m_k)}  \\
     & \qquad \qquad\times \sum_{\substack{r_t \in \{  1, \ldots , m\}\\ t=1,\ldots,y}} \sum_{\substack{c_\ell \in \{ 0, \ldots , m\}\\ \ell=1,\ldots,k}} \pi_{n+m}(m_1+c_1,\ldots,m_k+c_k,r_1,\ldots,r_y) \prod_{r=1}^y \binom{m}{r_t} \prod_{\ell=1}^k \binom{m}{c_\ell}\binom{k+y}{k}
\end{aligned}
\end{equation*}
and the thesis follows.
\end{proof}

We now concentrate on the posterior probability distribution of $K_m^{(n)}$ in Gibbs-type feature allocation models. By virtue of   Lemma \ref{lemma:Kmn_on_EFPF}, one has:
\begin{equation} \label{eq:prob_Kmn_proof1}
\begin{aligned}
    \P(K_m^{(n)} = y\mid \bm Z^{(n)}) & = \binom{k+y}{k} \frac{V_{n+m,k+y}}{\pi_n(m_1,\ldots,m_k)}  \sum_{\substack{r_t \in \{1, \ldots , m\}\\ t=1,\ldots,y}}  \prod_{t=1}^y \binom{m}{r_t} (1-\alpha)_{r_t-1} (\theta+\alpha)_{n+m-r_t}  \\ &\qquad\qquad \times\sum_{\substack{c_\ell \in \{0, \ldots , m\}\\ \ell=1,\ldots,k}}
 \prod_{\ell=1}^k  \binom{m}{c_\ell} (1-\alpha)_{m_\ell + c_\ell -1} (\theta+\alpha)_{n+m-m_\ell -c_\ell} .
\end{aligned}
\end{equation}
First, consider the following sum
\begin{equation*}
   \sum_{\substack{c_\ell \in \{0, \ldots , m\}\\ \ell=1,\ldots,k}} 
 \prod_{\ell=1}^k  \binom{m}{c_\ell} (1-\alpha)_{m_\ell + c_\ell -1} (\theta+\alpha)_{n+m-m_\ell -c_\ell} = \prod_{\ell=1}^k \sum_{c_\ell =0}^m 
 \binom{m}{c_\ell} (1-\alpha)_{m_\ell + c_\ell -1} (\theta+\alpha)_{n+m-m_\ell -c_\ell} ,
\end{equation*}
by observing that
\[
(1-\alpha)_{m_\ell + c_\ell -1} = (1-\alpha)_{m_\ell -1} (m_\ell - \alpha)_{c_\ell}
\]
and 
\[
(\theta+\alpha)_{n+m-m_\ell -c_\ell} = (\theta+\alpha)_{n-m_\ell} (\theta+\alpha+n-m_\ell)_{m-c_\ell},
\]
the Chu-Vandermonde identity implies 
\begin{equation*}
   \prod_{\ell=1}^k   \sum_{c_\ell =0}^m 
 \binom{m}{c_\ell} (1-\alpha)_{m_\ell + c_\ell -1} (\theta+\alpha)_{n+m-m_\ell -c_\ell} =   \prod_{\ell=1}^k \Bigg[   (1-\alpha)_{m_\ell -1} (\theta+\alpha)_{n-m_\ell} (\theta+n)_m  \Bigg] .
\end{equation*}
By exploiting the previous identity and the expression of the \textsc{efpf} \eqref{eq:EFPF_product}, the probability of interest \eqref{eq:prob_Kmn_proof1} boils down to
\begin{equation}
\label{eq:post_Kmn_proof2}
    \P(K_m^{(n)} = y\mid \bm Z^{(n)}) = \binom{k+y}{k} \frac{V_{n+m,k+y}}{V_{n,k}} (\theta+n)_m^{k}  \left[ \sum_{i =1}^m   \binom{m}{i} (1-\alpha)_{i-1} (\theta+\alpha)_{n+m-i}    \right]^y .  
\end{equation}
We now consider separately the two cases $\alpha<0$ and $\alpha \in [0,1)$.\\
\textbf{Case $\bm{\alpha<0}$.} For such values of $\alpha$, we observe that
\[
(1-\alpha)_{i-1} = -\frac{1}{\alpha} (-\alpha)_{i} \quad \text{and}  \quad
(\theta+\alpha)_{n+m-i} = (\theta+\alpha)_n (\theta+\alpha+n)_{m-i},
\]
and, thanks to the Chu-Vandermonde identity, one has
\begin{equation*}
    \sum_{i =1}^m   \binom{m}{i} (1-\alpha)_{i-1} (\theta+\alpha)_{n+m-i} = -\frac{1}{\alpha} (\theta+\alpha)_n \left[(\theta+n)_m - (\theta+\alpha+n)_m \right] .
\end{equation*}
The result for $\alpha<0$ follows  by substituting the previous expression in \eqref{eq:post_Kmn_proof2}. \\

\noindent
\textbf{Case $\bm{\alpha \in [0,1)}$.}
As for $\alpha \in [0,1)$, the sum appearing in \eqref{eq:post_Kmn_proof2} can be rewritten as follows:
\begin{equation*}
\begin{aligned}
    \sum_{i=1}^m \binom{m}{i} (1-\alpha)_{i-1} &(\theta+\alpha)_{n+m-i} = 
     \sum_{i=1}^m \binom{m}{i} \frac{\Gamma(i-\alpha)}{\Ga(1-\alpha)} \frac{\Ga(\theta+\alpha+m+n-i)}{\Ga(\theta+\alpha)} \cdot \frac{\Ga(\theta+m+n)}{\Ga(\theta+m+n)} \\
     &= \frac{\Ga(\theta+m+n)}{\Ga(1-\alpha) \Ga(\theta+\alpha) } \sum_{i=1}^m \binom{m}{i} B(i-\alpha, \theta+\alpha+m+n-i).
\end{aligned}
\end{equation*}
We can now use the integral representation of the beta function to get
\begin{equation*}
\begin{aligned}
    \sum_{i=1}^m \binom{m}{i} (1-\alpha)_{i-1} &(\theta+\alpha)_{n+m-i}  = \frac{\Ga(\theta+m+n)}{\Ga(1-\alpha) \Ga(\theta+\alpha) } \sum_{i=1}^m \binom{m}{i} \int_0^1 x^{i-\alpha-1} (1-x)^{\theta+\alpha+m+n-i-1} \D x\\
     &= \frac{\Ga(\theta+m+n)}{\Ga(1-\alpha) \Ga(\theta+\alpha) }  \int_0^1 x^{-\alpha-1} (1-x)^{\theta+\alpha+m+n-1} \sum_{i=1}^m \binom{m}{i} \left(\frac{x}{1-x}\right)^i  \D x\\
     &= \frac{\Ga(\theta+m+n)}{\Ga(1-\alpha) \Ga(\theta+\alpha) }  \int_0^1 x^{-\alpha-1} (1-x)^{\theta+\alpha+m+n-1} \left[ \left(\frac{x}{1-x} + 1\right)^m -1\right]  \D x\\
     &= \frac{\Ga(\theta+m+n)}{\Ga(1-\alpha) \Ga(\theta+\alpha) }  \int_0^1 x^{-\alpha-1} (1-x)^{\theta+\alpha+n-1} \left[ 1-(1-x)^m \right]  \D x  .   
\end{aligned}
\end{equation*}
By virtue of the simple identity
\[ 
1-(1-x)^m = x \sum_{i=0}^{m-1} (1-x)^i,
\]
the previous expression becomes
\begin{equation*}
\begin{aligned}
     \sum_{i=1}^m \binom{m}{i} (1-\alpha)_{i-1} &(\theta+\alpha)_{n+m-i}  \\
    &= \frac{\Ga(\theta+m+n)}{\Ga(1-\alpha) \Ga(\theta+\alpha) }  \int_0^1 x^{-\alpha+1 -1} (1-x)^{\theta+\alpha+n-1} \sum_{i=0}^{m-1} (1-x)^i  \D x  \\
    &= \frac{\Ga(\theta+m+n)}{\Ga(1-\alpha) \Ga(\theta+\alpha) }  \sum_{i=0}^{m-1} \int_0^1 x^{-\alpha+1 -1} (1-x)^{\theta+\alpha+n+ i-1}  \D x  \\
    &= \frac{\Ga(\theta+m+n)}{\Ga(1-\alpha) \Ga(\theta+\alpha) }  \sum_{i=0}^{m-1} \frac{\Ga(i+\theta+\alpha+n) \Ga(1-\alpha)}{\Ga(i+\theta+n+1)}\\
    &= \frac{\Ga(\theta+m+n)}{\Ga(1-\alpha) \Ga(\theta+\alpha) }  \sum_{i=1}^{m} \frac{\Ga(i+\theta+\alpha+n-1) \Ga(1-\alpha)}{\Ga(i+\theta+n)}\\
    &= \frac{\Ga(\theta+m+n)}{ \Ga(\theta+1) }  \sum_{i=1}^{m} \frac{(\theta+\alpha)_{n+i-1}}{(\theta+1)_{n+i-1}} = (\theta+1)_{m+n-1} \sum_{i=1}^{m} \frac{(\theta+\alpha)_{n+i-1}}{(\theta+1)_{n+i-1}}\\
    &= (\theta+1)_{n-1} (\theta+n)_{m} \sum_{i=1}^{m} \frac{(\theta+\alpha)_{n+i-1}}{(\theta+1)_{n+i-1}}
\end{aligned}
\end{equation*}
and the thesis for $\alpha \in [0,1)$ follows by substituting this expression in~\eqref{eq:post_Kmn_proof2}.

%% Diversity propostion
\subsection{Proof of Proposition~\ref{diversity}}

The asymptotic behavior of $K_n$  follows by specializing  the
asymptotic distribution of $K_{m}^{(n)}$, determined in Proposition~\ref{prop_asy_anyprior}, when $n=0$ and $m$ is replaced with $n$. One may also prove this result by working along the same lines as in the proof of Proposition~\ref{prop_asy_anyprior}.

%% Posterior diveristy propostion
\subsection{Proof of Proposition~\ref{prop_asy_anyprior}}

We prove the theorem for $\alpha <0$, $\alpha=0$ and $\alpha \in (0,1)$ separately.\\
\textbf{Case \bm{$\alpha<0$}.} This choice corresponds to a mixture over $N$ of a \textsc{bb} model. We first observe that, in the case of a \textsc{bb} model with parameter $N$, the posterior distribution of the number of hitherto unseen features $ K_m^{(n)}\mid \bm{Z}^{(n)}, N$ has a binomial distribution given by
\begin{equation*}
K_{m}^{(n)} \mid  \bm{Z}^{(n)}, N \sim \textup{Binomial} \Big(N-k, 1-\frac{(\theta+\alpha+n)_m}{(\theta+n)_m} \Big).
\end{equation*}
This result is a simple consequence of Theorem \ref{thm:Kmn_general_V} that can be obtained  by substituting the expression of the $V_{n,k}$'s of a 
\textsc{bb} model with parameter $N$, displayed in Equation~\eqref{eq:EFPF_BB}.
Thus, by letting $p_{m} (\theta+n, \alpha) =  1-(\theta+\alpha+n)_m /(\theta+n)_m$, the characteristic function of $K_{m}^{(n)} \mid \bm{Z}^{(n)}, N$, denoted here as $\Psi_{m}^{(n)}$, equals
\begin{equation} \label{eq:BB_characteristic_Knm}
    \Psi_{m}^{(n)}(t) = \left[ 1-p_{m} (\theta+n, \alpha)  + p_{m} (\theta+n, \alpha) e^{it} \right]^{N-k},
\end{equation}
for any $t \in \R$, and $i$ is the imaginary unit. 
Note that $\lim\limits_{m\to \infty} p_{m} (\theta+n, \alpha)  = 1$, indeed we have:
\begin{equation*}
\begin{split}
    p_{m} (\theta+n, \alpha)  & = 1- \frac{(\theta+\alpha+n)_m}{(\theta+n)_m} =1-  \frac{\Gamma(\theta+\alpha+n+m)}{\Gamma(\theta+\alpha+n) } \cdot\frac{\Gamma(\theta+n)}{\Gamma(\theta+n+m)}\\
    & = 1-  \frac{\Gamma(\theta+n)}{\Gamma(\theta+\alpha+n) } \cdot m^\alpha+ o (m^\alpha) , 
    \end{split}
\end{equation*}
where we have used the asymptotic expansion of ratio of gamma functions \citep{expansion_gamma}, and the little-$o$ notation. Since $\alpha <0$, the limit of $p_{m} (\theta+n, \alpha) $ goes to $1$, as $m \to + \infty$. Thus, for all $ t \in \R$, the limit of the characteristic function in \eqref{eq:BB_characteristic_Knm} equals 
\[
\lim\limits_{m\to \infty} \Psi_{m}^{(n)}(t) = e^{it(N-k)},
\]
in other words, it holds $K_m^{(n)}\mid \bm{Z}^{(n)}, N \dconv N-k$.\\
Now, consider a prior distribution for the parameter $N$ having probability mass function $p_N$, and denote by $\Phi_m^{(n)}$ the characteristic function of $K_m^{(n)}\mid \bm{Z}^{(n)}$ for this statistical model. We can now compute the limit of the characteristic function for any $t \in \R$:
\begin{equation*}
\begin{aligned}
    \lim_{m\to \infty}\Phi_m^{(n)}(t) = \lim_{m\to \infty} \E\left[ e^{itK_m^{(n)}}\mid \bm{Z}^{(n)}\right] = \lim_{m\to \infty} \E\left[\E \left[ e^{itK_m^{(n)}}\mid \bm{Z}^{(n)}, N\right]\mid \bm{Z}^{(n)}\right]
\end{aligned}
\end{equation*}
and note that $\Psi_m^{(n)}(t) = \E \left[ e^{itK_m^{(n)}}\mid \bm{Z}^{(n)}, N\right]$ is the characteristic function under a \textsc{bb} model. Since $\mid~\Psi_m^{(n)}(t)~\mid \leq 1$ and $\lim\limits_{m\to \infty} \Psi_m^{(n)}(t) = e^{it(N-k)}$, the  dominated convergence theorem implies
\[
 \lim_{m\to \infty}\Phi_m^{(n)}(t) = \E\left[ e^{it(N-k)}\mid \bm{Z}^{(n)}\right].
\]
Therefore, $ K_m^{(n)}\mid \bm{Z}^{(n)}$ does converge in distribution to the random variable $ N-k\mid \bm{Z}^{(n)}  \dequal N'$, as $m \to + \infty$. Note that $N'$ is almost surely finite since $N$ is a priori almost surely finite, and the posterior distribution of $N$ equals %, and $\P(N' = 0\mid \bm{Z}^{(n)}) = \P(N = k\mid \bm{Z}^{(n)})$, where
\[
p_N (y \mid \bm{Z}^{(n)}) \propto \frac{y!}{(y-k)!} \left(\frac{(\theta+\alpha)_n}{(\theta)_n} \right)^y \cdot \indic_{\{k,k+1,\ldots\}}(y) p_N (y).
\]

\noindent
\textbf{Case $\bm{\alpha=0}$.} This case corresponds to a mixture over $\gamma$ of a two-parameter \textsc{ibp} model. We first assume that $\gamma$ is a fixed parameter, and we focus on the asymptotic behaviour of the sequence of random variables $K_m^{(n)}/\log(m)\mid \bm{Z}^{(n)}, \gamma$, as $m \to + \infty$.
As a consequence of Theorem \ref{thm:Kmn_general_V}, by substituting the expression of the $V_{n,k}$'s of a two-parameter \textsc{ibp} displayed in Equation \eqref{eq:EFPF_IBP}, it follows that 
\begin{equation*}
K_{m}^{(n)} \mid  \bm{Z}^{(n)}, \gamma \sim \textup{Poisson} \left( \gamma \theta \sum_{i=1}^m {(n + \theta + i -1)^{-1}} \right).
\end{equation*}
Denoting with $\Psi_{m}^{(n)}$ the characteristic function of $K_{m}^{(n)}/\log(m) \mid \bm{Z}^{(n)}, \gamma$, it holds that
\begin{equation*}
    \begin{aligned}
        \Psi_{m}^{(n)}(t) &= \E\left[ \exp\left\{ i t K_m^{(n)} /\log(m) \right\} \mid \bm{Z}^{(n)}, \gamma \right]  \\
        &= \exp\left\{ \gamma \theta \sum_{i=1}^m {(n + \theta + i -1)^{-1}} \cdot \left( e^{it/\log(m)} -1 \right) \right\},
    \end{aligned}
\end{equation*}
for any $t \in \R$, and $i$ stands for the imaginary unit. We observe that, as $m \to + \infty$
\begin{equation*}
        \Psi_m^{(n)}(t) = \exp\left\{ \gamma \theta \sum_{i=1}^m {(n + \theta + i -1)^{-1}} \cdot \left[ it/\log(m) + \mathcal{O}(\log(m)^{-2}) \right] \right\},
\end{equation*}
where we have used the big-$\mathcal{O}$ notation. Moreover, as $m \to + \infty$, it is easy to see that
\[
 \sum_{i=1}^m {(n + \theta + i -1)^{-1}} / \log(m) = 1 + \mathcal{O}(\log(m)^{-1}),
\]
as a consequence, for any $t \in \R$, the characteristic function satisfies 
\[
\lim_{m \to \infty} \Psi_m^{(n)}(t) = e^{it\gamma \theta}.
\]
In other words, we have shown that  $K_m^{(n)}/\log(m)\mid\bm{Z}^{(n)}, \gamma \dconv \gamma\theta$.

Now, we consider a prior distribution for the parameter $\gamma$ having density function $p_\gamma$, and we focus on the point-wise convergence of the characteristic function of the random variable~$K_m^{(n)}/\log(m)\mid \bm{Z}^{(n)}$, denoted here as $\Phi_m^{(n)}$. To this end, let $t \in \R$ and evaluate the following limit
\begin{equation*}
\begin{aligned} 
\lim_{m\to\infty}\Phi_m^{(n)}(t) &= \lim_{m\to\infty} \E\left[ \exp\left\{it K_m^{(n)}/\log(m)\right\} \mid \bm{Z}^{(n)} \right] \\
& = \lim_{m\to\infty} \E\left[ \E\left[ \exp\left\{it K_m^{(n)}/\log(m)\right\} \mid \bm{Z}^{(n)},\gamma\right] \mid \bm{Z}^{(n)} \right]. 
\end{aligned}
\end{equation*}
Note that $\Psi_m^{(n)}(t) = \E\left[ \exp\left\{it K_m^{(n)}/\log(m)\right\} \mid \bm{Z}^{(n)},\gamma\right] $ is the characteristic function under the two-parameter \textsc{ibp} model. Since $\mid \Psi_m^{(n)}(t)\mid \leq 1$ and $\lim\limits_{m\to \infty} \Psi_m^{(n)}(t) = e^{it \gamma \theta}$, the  dominated convergence theorem implies
\begin{equation*}
     \lim_{m\to\infty}\Phi_m^{(n)}(t) = \E\left[e^{it\gamma\theta}\mid\bm{Z}^{(n)}\right],
\end{equation*}
therefore $K_m^{(n)}/\log(m)\mid \bm{Z}^{(n)}$ does converge in distribution to the random variable  $\gamma\theta\mid\bm{Z}^{(n)}$, and the thesis follows.

\noindent
\textbf{Case $\bm{\alpha \in (0,1)}$.} We consider the limit of the characteristic function of $K_m^{(n)}/m^\alpha\mid \bm{Z}^{(n)}$, denoted with $\Phi_m^{(n)}$, for a mixture over $\gamma$ of a three-parameter \textsc{ibp} model with parameters $(\gamma, \alpha, \theta)$:
\begin{equation*}
\begin{aligned}
    \lim_{m\to\infty}\Phi_m^{(n)}(t) &= \lim_{m\to\infty} \E\left[ \exp\left\{itK_m^{(n)}/m^\alpha\right\} \mid \bm{Z}^{(n)} \right]\\
    &= \lim_{m\to\infty} \E\left[ \E\left[ \exp\left\{it K_m^{(n)}/m^\alpha\right\} \mid \bm{Z}^{(n)},\gamma\right] \mid \bm{Z}^{(n)} \right].
\end{aligned}
\end{equation*}
Note that $\E\left[ \exp\left\{it K_m^{(n)}/m^\alpha\right\} \mid \bm{Z}^{(n)},\gamma\right]$ corresponds to the characteristic function of the random variable $K_m^{(n)}/m^\alpha\mid\bm{Z}^{(n)},\gamma$ for the three-parameter \textsc{ibp}. \cite[Proposition 2]{Mas22} show that $K_m^{(n)}/m^\alpha\mid\bm{Z}^{(n)},\gamma \overset{a.s.}{\longrightarrow} \xi$, with $\xi = \gamma\Ga(\theta+1)/(\alpha\Ga(\theta+\alpha))$. Thus, as $m\to \infty$, we get
\[
\E\left[ \exp\left\{it K_m^{(n)}/m^\alpha\right\} \mid \bm{Z}^{(n)},\gamma\right] \longrightarrow e^{it\xi} .
\]
By an application of the dominated convergence theorem, we obtain 
\begin{equation*}
     \lim_{m\to\infty}\Phi_m^{(n)}(t) = \E\left[e^{it\xi}\mid\bm{Z}^{(n)} \right] = \E\left[\exp\left\{it\frac{\gamma\Ga(\theta+1)}{\alpha\Ga(\theta+\alpha)}\right\}\mid\bm{Z}^{(n)} \right] = \Phi_{\gamma\mid\bm{Z}^{(n)}}\left(t\frac{\Ga(\theta+1)}{\alpha\Ga(\theta+\alpha)} \right)
\end{equation*}
where $\Phi_{\gamma\mid\bm{Z}^{(n)}}$ denotes the characteristic function of $\gamma\mid\bm{Z}^{(n)}$, and the thesis follows for $\alpha \in (0,1)$.

%% Asymptotic for anyprior 3IBP
\subsection{Proof of Theorem~\ref{thm:post_proc}}

We prove the theorem for $\alpha <0$ and $\alpha \in [0,1)$ separately.\\
\textbf{Case \bm{$\alpha<0$}.}

We preface a Lemma to  provide the hierarchical representation of the \textsc{bb} process. We remark that the following is an alternative construction of the \textsc{bb} model in \cite{Bat18} and \cite{Gri11}.
\begin{lemma}\label{lemma:bb_process_representation}
The \textsc{bb} model with parameters $(N,\alpha,\theta)$ may be equivalently described in the following hierarchical form
\begin{equation}\label{eq:bb_process}
\begin{aligned}
    Z_i \mid \mutilde &\iid \BP(\mutilde), \qquad i \geq 1 \\
    \tilde \mu & = \sum_{j=1}^N \tilde q_j \delta_{\tilde X_j} ,
\end{aligned}
\end{equation}
where $\tilde q_j$ are i.i.d. beta random variables with parameters $(-\alpha, \theta+\alpha)$ and $\tilde X_j \iid P_0$, for $j=1,\ldots,N$.
\end{lemma}
\begin{proof}
We show that the \textsc{efpf} induced by the model \eqref{eq:bb_process} is the \textsc{efpf} in \eqref{eq:EFPF_product} with $V_{n,k}$'s in \eqref{eq:EFPF_BB} characterizing the \textsc{bb} model. Assuming model \eqref{eq:bb_process}, introduce the $n \times N$ binary matrix $\bm \Ztilde$, whose generic element $\Ztilde_{i,j}$ is defined as $\Ztilde_{i,j} = Z_i(\tilde X_j)$. In particular, $\Ztilde_{i,j} = 1$ if subject $i$ possesses feature $j$, $\Ztilde_{i,j} = 0$ otherwise.  It holds that the entries of the matrix $\bm \Ztilde$ are distributed as
\begin{equation*}
\begin{aligned}
\Ztilde_{i,j} \mid \tilde{q}_j & \simind {\rm Bernoulli} (\tilde{q}_j),  \quad i \geq 1,\, j= 1, \ldots , N,\\
\tilde{q}_j  &  \simiid  {\rm Beta} (-\alpha, \theta + \alpha),  \quad j=1, \ldots, N.
\end{aligned}
\end{equation*}
Indeed, let $\bm{\ztilde} \in \{ 0,1 \}^{n\times N}$ a generic realization of $\bm \Ztilde$, then
\begin{equation} \label{eq:model_BB_matrix}
\P (\bm{\Ztilde}= \bm{\ztilde}\mid \tilde{q}_1, \ldots ,  \tilde{q}_N) = \prod_{j=1}^N  \prod_{i=1}^n  \tilde{q}_j^{\ztilde_{i,j}} (1-\tilde{q}_j)^{1-\ztilde_{i,j}} = \prod_{j=1}^N  \tilde{q}_j^{m_j} (1-\tilde{q}_j)^{n-m_j},
\end{equation}
where $m_j = \sum_{i=1}^n \ztilde_{i,j}$. The marginal distribution of $\bm \Ztilde$ results in 
\begin{equation*}
    \begin{aligned}
        \P (\bm{\Ztilde} = \bm \ztilde) &= \int \P (\bm{\Ztilde}= \bm{\ztilde}\mid \tilde{q}_1, \ldots ,  \tilde{q}_N) \cdot \prod_{j=1}^N \tilde{q}_j^{-\alpha - 1} (1 - \tilde{q}_j)^{\alpha + \theta  -1 } \D \tilde{q}_1 \cdots \D \tilde{q}_N\\
        &=  \left(  \frac{\Gamma(\theta)}{\Gamma (- \alpha) \Gamma (\theta + \alpha) \Gamma (n+\theta)}\right)^N \prod_{j=1}^N \Gamma (m_j - \alpha) \Gamma (n-m_j +\alpha + \theta).
    \end{aligned}
\end{equation*}
The \textsc{efpf}  corresponding to $\bm \ztilde$, can be evaluated by taking into account the $\binom{N}{k}$ possible ways  to place the null columns of $ \bm{ \tilde z}$, thus obtaining:
\begin{equation*}
\begin{aligned}
\pi_n(m_1,\ldots,m_k) & =  \binom{N}{k} \P(\bm \Ztilde = \bm{\tilde z}) \\ %&=  \P((B_1,\ldots,B_k)) = \binom{N}{k} \P(\bm \Ztilde) \\
&= \binom{N}{k} 
\left(  \frac{\Gamma(\theta)}{\Gamma (- \alpha) \Gamma (\theta + \alpha) \Gamma (n+\theta)}\right)^N \prod_{j=1}^N \Gamma (m_j - \alpha) \Gamma (n-m_j +\alpha + \theta),
 \end{aligned}
\end{equation*}
which corresponds to the \textsc{efpf} in \eqref{eq:EFPF_product} with $V_{n,k}$'s as in \eqref{eq:EFPF_BB}.
\end{proof}

We now need to provide a posterior representation of the latent measure $\tilde{\mu}$ in Lemma \ref{lemma:bb_process_representation}. In particular, assume that $k$ features have been observed in $\bm Z^{(n)}$, with labels $X_1,\ldots, X_k$. Moreover, assume that each observed feature $X_\ell$ has been displayed in $m_\ell$ subjects, for $\ell = 1,\ldots, k$. Among the $N$ atoms of $\tilde{\mu}$ in model \eqref{eq:bb_process}, $k$ atoms of the posterior random measure $\tilde{\mu} \mid \bm Z^{(n)}$ are necessarily $X_1,\ldots,X_k$. The remaining $N-k$ atoms of $\tilde{\mu} \mid \bm Z^{(n)}$ are drawn independently from the prior. Specifically, the posterior distribution of $\tilde{\mu} \mid \bm Z^{(n)}$ may be characterized as
\[
\mutilde\mid\bm{Z}^{(n)} \dequal \mutilde' + \mutilde^*,
\]
where $\mutilde^* = \suml q_\ell \delta_{X_\ell}$ accounts for the $k$ observed features and $\mutilde' = \sum_{j=1}^{N-k} q'_j \delta_{\tilde{X}_j}$ describes the remaining $N-k$ possible features, with $\tilde{X}_j \iid P_0$, for $j=1,\ldots,N-k$. The joint distribution of  $q'_j$'s and $q_\ell$'s can be obtained by multiplying  the likelihood function \eqref{eq:model_BB_matrix} with the prior distribution of the $\tilde{q}_j$'s, which are independent beta random variables with parameters $(-\alpha, \theta + \alpha)$. Thus, an application of the Bayes theorem leads to $q_\ell \ind \dbeta(m_\ell-\alpha, \alpha + \theta + n - m_\ell)$, as $\ell=1, \ldots , k$, and  $q'_j \iid \dbeta(-\alpha, \alpha + \theta + n)$, for $j=1,\ldots,N-k$.

%We are left with describing the distribution of the jumps $(q_\ell)_{\ell = 1,\ldots,k}$ and $(q'_j)_{j=1,\ldots,N-k}$. From model \eqref{eq:bb_process}, the posterior distribution of the $\tilde{q}_j$'s is equivalent to the posterior distribution of $\tilde{q}_j$'s in the following model: $Y_j\mid \tilde{q}_j \ind {\rm Binomial}(n, \tilde{q}_j), \tilde{q}_j \iid \dbeta(-\alpha, \theta + \alpha)$, where $Y_j$ is the number of subjects exhibiting feature $j$, for $j = 1,\ldots,N$. Therefore, we first note that, a posteriori, the $\tilde{q}_j$'s are independent. Then, for the $\ell$-th observed feature $X_\ell$, we have $q_\ell \dequal \tilde{q}_1\mid Y_1 = m_\ell$, which corresponds to $q_\ell \ind \dbeta(m_\ell-\alpha, \alpha + \theta + n - m_\ell)$. Finally, for the remaining jumps $(q'_j)_{j=1,\ldots,N-k}$, we have $q'_j \dequal \tilde{q}_1\mid Y_1 = 0$, which leads to $q'_j \iid \dbeta(-\alpha, \alpha + \theta + n)$, for $j=1,\ldots,N-k$. 

Point (i) of Theorem \ref{thm:post_proc}, valid for any mixture of \textsc{bb}s, follows by conditioning on $N$ and applying the just discussed distributional equality for the posterior distribution of $\tilde{\mu}$.

\noindent
\textbf{Case $\bm{\alpha \in [0,1)}$.}
This corresponds to the \textsc{ibp} case, where $\tilde\mu \mid \gamma$ is a stable-beta process with intensity measure given by \eqref{eq:stable_beta_process}, namely:
\begin{equation} \label{eq:rho_stable_beta_app}
    \rho(\mathrm{d}s) = \gamma \frac{\Ga(1+\theta)}{\Ga(1-\alpha) \Ga(\theta+ \alpha)} s^{-\alpha -1} (1-s)^{\theta+\alpha -1} \D s.
\end{equation}
Thus, conditionally on $\gamma$, the posterior representation in (ii) follows from \cite[Theorem 3.1]{Jam17}, which provides a characterization for general \textsc{crm}s. More precisely, one has
\begin{equation*}
    \mutilde\mid\bm{Z}^{(n)} \dequal \mutilde^* + \mutilde'
\end{equation*}
where  $\mutilde^*$ and $\mutilde'$ are independent random measures such that:
\begin{itemize}
    \item[(1)] $\tilde \mu '\sim \textsc{crm} (\rho' ; P_0)$, with $\rho'(\D s) = (1-s)^n \rho( \D s)$;
    \item[(2)] the random measure $\mutilde^*= \sum_{\ell=1}^k q_\ell \delta_{X_\ell}$ is almost surely discrete,
    where the $X_\ell$'s are the distinct feature labels out of  the observed sample $\bm{Z}^{(n)}$, and $ q_\ell$'s are independent random variables with density 
    \begin{equation*} 
        f_{q_\ell} (\D s ) \propto (1-s)^{n-m_\ell} s^{m_\ell} \rho(\D s)
    \end{equation*}
    as $\ell = 1, \ldots , k$. 
\end{itemize}
By substituting the specific expression of $\rho$ \eqref{eq:rho_stable_beta_app} in the previous characterization, part (ii) of the theorem easily follows.

\section{Proofs of Section \ref{sec:mixture_IBP}}\label{app:proofs_mixtureIBPs}

%% EFPF Gamma 3IBP
\subsection{Details for the determination of \eqref{eq:efpf_gamma_mix}}

Here we show that the gamma mixture of \textsc{ibp}s has a product form \textsc{efpf} with weights $V_{n,k}$'s as in~\eqref{eq:efpf_gamma_mix}.
To this end, we integrate the product form \textsc{efpf} in \eqref{eq:EFPF_product}-\eqref{eq:EFPF_IBP} with respect to the parameter $\gamma \sim \dgamma(a,b)$:
\begin{equation*}
\begin{aligned}
    \pi_n (m_1, \ldots , m_k)  &= \int_0^\infty \frac{1}{k!} \left( \frac{\gamma}{(\theta+1)_{n-1}} \right)^k
e^{ - \gamma g_n (\theta, \alpha) }\prod_{\ell=1}^k (1-\alpha)_{m_\ell-1} (\theta+\alpha)_{n-m_\ell}  \frac{b^a}{\Gamma(a)} \gamma^{a-1} e^{-\gamma b} \D \gamma \\
&= \frac{b^a}{k!\Ga(a) \{(\theta+1)_{n-1}\}^k} \int_0^\infty \gamma^{k+a-1} \exp\{-\gamma(g_n (\theta, \alpha) +b)\} \D \gamma  \prod_{\ell=1}^k (1-\alpha)_{m_\ell-1} (\theta+\alpha)_{n-m_\ell}\\
&= \frac{b^a}{k!\Ga(a) \{(\theta+1)_{n-1}\}^k} \cdot \frac{\Ga(k+a)}{(g_n (\theta, \alpha) +b)^{k+a}} \prod_{\ell=1}^k (1-\alpha)_{m_\ell-1} (\theta+\alpha)_{n-m_\ell}\\
&= \frac{b^a \, (a)_k}{k! \{(\theta+1)_{n-1}\}^k (g_n (\theta, \alpha) +b)^{k+a}} \prod_{\ell=1}^k (1-\alpha)_{m_\ell-1} (\theta+\alpha)_{n-m_\ell}.
\end{aligned}
\end{equation*}
It is easy to observe that the last expression is an \textsc{efpf} of type \eqref{eq:EFPF_product} where
\[
V_{n,k} = \frac{b^a \, (a)_k}{k! \{(\theta+1)_{n-1}\}^k (g_n (\theta, \alpha) +b)^{k+a}}.
\]

%% mixtures IBP: distributions of the number of distinct features
\subsection{Proof of Proposition~\ref{prop_k_gamma_ibp}}

The distributions of $K_n$ in \eqref{eq:Kn_P_NB} follow by specializing Theorem~\ref{thm:Kn_general} with the two specifications~\eqref{eq:EFPF_IBP} and \eqref{eq:efpf_gamma_mix} for $V_{n,k}$. Analogously, the posterior distributions for~$K_m^{(n)}$ in~\eqref{eq:Knm_Poiss} and~\eqref{eq:Kmn_NB} follow from Theorem~\ref{thm:Kmn_general_V}.

%% Asymptotic for gamma prior 3IBP
\subsection{Proof of Proposition \ref{prop_asy_km_gamma_3ibp}}

From Equation \eqref{eq:gamma_N_posterior}, it is easy to see that the posterior distribution of $\gamma\mid\bm{Z}^{(n)}$ is a gamma with parameters $(k+a,g_n (\theta, \alpha) + b)$, since $p_\gamma$ is a gamma density with parameters $(a,b)$. Thus, we also remind that the characteristic function of $\gamma\mid\bm{Z}^{(n)}$  equals
\[
\Phi_{\gamma\mid\bm{Z}^{(n)}}(t) = \left( \frac{g_n (\theta, \alpha) + b}{g_n (\theta, \alpha) + b - it}\right)^{k+a}.
\]
We now consider the case $\alpha =0$ and $\alpha \in (0,1)$ separately.\\ 
\textbf{Case $\bm{\alpha \in (0,1)}$.}
 By an application of Proposition \ref{prop_asy_anyprior}, $K_m^{(n)}/m^\alpha\mid \bm{Z}^{(n)}$ converges in distribution to $\gamma \Ga(\theta+1)/(\alpha\Ga(\theta+\alpha))\mid \bm{Z}^{(n)}$, whose characteristic function equals
\[
    \Phi_{\gamma \frac{\Ga(\theta+1)}{\alpha\Ga(\theta+\alpha)}\mid\bm{Z}^{(n)}}(t) = \Phi_{\gamma\mid\bm{Z}^{(n)}}\left(t \frac{\Ga(\theta+1)}{\alpha\Ga(\theta+\alpha)}\right) = \left( \frac{g_n (\theta, \alpha) + b}{g_n (\theta, \alpha) + b - it \frac{\Ga(\theta+1)}{\alpha\Ga(\theta+\alpha)}}\right)^{k+a}
\]
which is the characteristic function of a gamma random variable with parameters $(k+a, (g_n (\theta, \alpha) + b) \Ga(\theta+\alpha)\alpha/\Ga(\theta+1) )$, as stated.\\
\noindent 
\textbf{Case $\bm{\alpha =0}$.} By virtue of Proposition \ref{prop_asy_anyprior}, one has that $K_m^{(n)}/\log(m)\mid \bm{Z}^{(n)}$ converges in distribution to $\gamma \theta\mid \bm{Z}^{(n)}$, whose characteristic function equals
\[
    \Phi_{\gamma \theta \mid\bm{Z}^{(n)}}(t) = \Phi_{\gamma\mid\bm{Z}^{(n)}}\left(t \theta \right) = \left( \frac{g_n (\theta, 0) + b}{g_n (\theta, 0) + b - it \theta}\right)^{k+a}
\]
which is the characteristic function of a gamma random variable with parameters $(k+a, (g_n (\theta, 0) + b)/\theta )$, as stated.

%% Process formulation for the Gamma mixture of IBP
\subsection{Proof of Equation \eqref{model:ibp_gamma_process}}

We need to show the equivalence between the hierarchical formulation of the gamma mixture of \textsc{ibp}s and the formulation \eqref{model:ibp_gamma_process}, which relies on the negative binomial process. More specifically, we consider the stable-beta process with parameters $(\gamma, \alpha, \theta)$, i.e.,  $\tilde \mu \mid \gamma \sim \textsc{crm} (\rho '; P_0)$, where $\rho '$ is as in \eqref{eq:stable_beta_process}
\[
\rho '(\D s) = \gamma \frac{\Ga(1+\theta)}{\Ga(1-\alpha) \Ga(\theta+ \alpha)} s^{-\alpha -1} (1-s)^{\theta+\alpha -1} \D s, 
\]
and we choose a  gamma prior with parameters $(a,b)$ for $\gamma$. The  Laplace functional of the resulting random measure $\tilde \mu$ equals
\begin{equation*}
\begin{aligned}
    \mathcal{L}_{\tilde \mu}(g)& = \E\left[e^{-\int_\X  g(x) \tilde \mu (\D x )} \right] = \E\left[ \E\left[ e^{-\int_\X  g(x) \tilde \mu (\D x )} \mid \gamma \right] \right] = \E\left[ \exp\left\{ - \int_{0}^1 \int_\X \left(1- e^{-s g(x)} \right) \rho'(\D s) P_0(\D x) \right\} \right] \\
    & = \int_0^\infty \exp\left\{ - \gamma \int_{0}^1 \int_\X \left(1- e^{-s g(x)} \right)  \frac{\Ga(1+\theta) s^{-\alpha -1} (1-s)^{\theta+\alpha -1}}{\Ga(1-\alpha) \Ga(\theta+ \alpha) }  \D s P_0(\D x) \right\} \frac{b^a}{\Ga(a)} \gamma^{a-1} e^{-b\gamma} \D \gamma,
\end{aligned}
\end{equation*}
for any measurable  function $g: \X \to \R_+$. Thus, by integrating with respect to $\gamma$, the Laplace functional of $\tilde\mu $ boils down to
\[
\mathcal{L}_{\tilde \mu}(g) = \left( 1 + \int_{0}^1 \int_\X \left(1- e^{-s g(x)} \right) \frac1b \frac{\Ga(1+\theta)}{\Ga(1-\alpha) \Ga(\theta+ \alpha)} s^{-\alpha -1} (1-s)^{\theta+\alpha -1} \D s P_0(\D x) \right)^{-a} .
\]
Therefore, $\tilde \mu$ is distributed as a negative binomial process with parameters $(a, \rho; P_0)$, where 
\[
\rho (\D s):= \frac1b \frac{\Ga(1+\theta)}{\Ga(1-\alpha) \Ga(\theta+ \alpha)} s^{-\alpha -1} (1-s)^{\theta+\alpha -1} \D s  .
\]

%% Posterior distribution of the process
\subsection{Proof of Corollary \ref{thm:post_proc_ibp_gamma}}

It is sufficient to specialize part (ii) of Theorem \ref{thm:post_proc} and to show that the measure $\tilde \mu ' $, which appears in the posterior representation, is a negative binomial process.
In order to do this, we now compute the Laplace functional of $\tilde  \mu ' $, and we show it coincides with the Laplace functional of a negative binomial process of type~\eqref{eq:NB_Laplace_functional}.
For any measurable function~$g: \X \to \R_+$, we evaluate
\begin{equation*}
    \begin{split}
       & \E \left[ e^{-\int_\X  g(x) \tilde \mu ' (\D x )}\right]  = \E \left[ \E \left[    e^{-\int_\X  g(x) \tilde \mu ' (\D x )} \mid \gamma' \right]  \right]\\
        & \qquad = \E \left[  \exp\left\{ - \gamma' \int_\X\int_{0}^1  \left(1- e^{-s g(x)} \right)  \frac{\Ga(1+\theta )s^{-\alpha -1} (1-s)^{\theta+\alpha+n -1} \D s }{\Ga(1-\alpha) \Ga(\theta+ \alpha )}  P_0(\D x) \right\} \right],
    \end{split}
\end{equation*}
where we used the fact that~$\tilde \mu ' \mid \gamma '$ is a \textsc{crm} characterized by the intensity measure
\[
     \gamma'\frac{\Gamma(1+\theta )}{\Gamma(1-\alpha) \Gamma(\theta+\alpha)}s^{-\alpha-1}(1-s)^{n+\theta+\alpha-1} \D s .
\]
We now integrate out~$\gamma'$, which equals in distribution to the posterior~$\gamma \mid \bm{Z}^{(n)} \sim \dgamma(a + k, b + g_n(\theta, \alpha))$. Therefore, the Laplace functional under study boils down to
\begin{align}
  & \E \left[ e^{-\int_\X  g(x) \tilde \mu ' (\D x )}\right] \nonumber \\
  & \qquad = \int_0^\infty \exp\left\{ - \gamma' \int_{0}^1 \int_\X \left(1- e^{-s g(x)} \right)  \frac{\Ga(1+\theta )}{\Ga(1-\alpha) \Ga(\theta+ \alpha )} s^{-\alpha -1} (1-s)^{\theta+\alpha+n -1} \D s P_0(\D x) \right\} \nonumber\\
    &\qquad \qquad\qquad\qquad\qquad \times\frac{(g_n (\theta, \alpha) + b)^{k +a}}{\Ga(k + a)} \gamma^{k + a-1} e^{-(g_n (\theta, \alpha) + b)\gamma} \D \gamma \nonumber\\
    & \qquad = \left( 1 + \int_{0}^1 \int_\X \left(1- e^{-s g(x)} \right) \rho'(\D s) P_0(\D x) \right)^{-k -a},
    \label{eq:Laplace_final_NB_gamma_proof}
\end{align}
where we have set
\[
\rho ' (\D s ) := (g_n (\theta, \alpha) + b)^{-1} \cdot \frac{\Ga(1+\theta )}{\Ga(1-\alpha) \Ga(\theta+ \alpha )} s^{-\alpha -1} (1-s)^{\theta+\alpha+n -1}  \D s .
\]
We note that the expression in \eqref{eq:Laplace_final_NB_gamma_proof} is the Laplace functional of a negative binomial process \eqref{eq:NB_Laplace_functional} with parameters $(k+a, \rho'; P_0)$. Thus, the thesis follows.

%%%%%%%%%%%%%%%%%%%%%%%%%%%%%%%%%%%%%%%%%%%%
%%%%% Mixture of Beta-Bernoulli %%%%%%%%%%%
%%%%%%%%%%%%%%%%%%%%%%%%%%%%%%%%%%%%%%%%%%%

\section{Proofs of Section \ref{sec:BB_mixtures}}
\label{app:proofs_mixtures_BB}

%% EFPF Poisson
\subsection{Proof of Proposition \ref{prop_efpf_poiss_nb}}

We first concentrate on the proof of \eqref{eq:EFPF_BB_Poiss}. We integrate the \textsc{efpf} of a \textsc{bb} process \eqref{eq:EFPF_BB} with respect to  $N \sim \textup{Poisson} (\lambda)$:
\begin{equation*}
\begin{split}
\pi_n (m_1, \ldots , m_k) & = \sum_{N \geq k} \binom{N}{k} \left( \frac{-\alpha}{(\theta+\alpha)_n}   \right)^k   \left( \frac{(\theta+\alpha)_n}{(\theta)_n}  \right)^N
\prod_{\ell=1}^k  (1-\alpha)_{m_\ell-1} (\theta+\alpha)_{n-m_\ell}   \cdot e^{-\lambda} \frac{\lambda^N}{N!}\\
& =  \frac{1}{k!} e^{-\lambda} \lambda^k  \left( \frac{(\theta+\alpha)_n}{(\theta)_n} \right)^k  \left( \frac{-\alpha}{(\theta+\alpha)_n} \right)^k\\
& \qquad \times\prod_{\ell=1}^k  (1-\alpha)_{m_\ell-1} (\theta+\alpha)_{n-m_\ell}    \sum_{N \geq k}  \frac{1}{(N-k)!} \lambda^{N-k} \left( \frac{(\theta+\alpha)_n}{(\theta)_n} \right)^{N-k} \\
& =  \frac{1}{k!} e^{-\lambda} \lambda^k  \left( \frac{(\theta+\alpha)_n}{(\theta)_n} \right)^k  \left( \frac{-\alpha}{(\theta+\alpha)_n} \right)^k\\
& \qquad \times\prod_{\ell=1}^k  (1-\alpha)_{m_\ell-1} (\theta+\alpha)_{n-m_\ell}   \exp \left\{ \lambda \frac{(\theta+\alpha)_n}{(\theta)_n} \right\}\\
& = \frac{1}{k!} \left\{ \frac{-\alpha \lambda  }{(\theta)_n} \right\}^k\exp \left\{ - \lambda \left( 1- \frac{(\theta+\alpha)_n}{(\theta)_n} \right) \right\} \prod_{\ell=1}^k  (1-\alpha)_{m_\ell-1} (\theta+\alpha)_{n-m_\ell} .
\end{split}
\end{equation*}
It is now easy to realize that
\[
\pi_n (m_1, \ldots , m_k) = V_{n,k } \prod_{\ell=1}^k  (1-\alpha)_{m_\ell-1} (\theta+\alpha)_{n-m_\ell} 
\]
where $V_{n,k}$ is as in \eqref{eq:EFPF_BB_Poiss}.\\

Now, we move to the proof of \eqref{eq:EFPF_BB_NB}. As before, we consider the \textsc{efpf} in product form of a \textsc{bb} process, as specified in \eqref{eq:EFPF_BB}, and we integrate out $N$, distributed as a negative binomial random variable with parameters $(n_0, \mu_0)$.
Throughout the computation, we also set $p = n_0 / (\mu_0 + n_0)$ so that the negative binomial is parametrized with respect to $n_0$ and the success probability $p$.
The \textsc{efpf} equals
\begin{align*}
\pi_n (m_1, \ldots , m_k)  &= \sum_{N \geq k} \binom{N}{k}   \Big( \frac{-\alpha}{(\theta+\alpha)_n} \Big)^k  \Big(\frac{(\theta+\alpha)_n}{(\theta)_n} \Big)^N
\prod_{\ell=1}^k  (1-\alpha)_{m_\ell-1} (\theta+\alpha)_{n-m_\ell}  \\
& \qquad\qquad \times   \binom{N+n_0-1}{N} p^{n_0}  (1-p)^N \\
&=  \prod_{\ell=1}^k  (1-\alpha)_{m_\ell-1} (\theta+\alpha)_{n-m_\ell} \cdot p^{n_0}  \Big( \frac{-\alpha}{(\theta+\alpha)_n} \Big)^k \\
& \qquad\qquad \times
\frac{1}{k!(n_0-1)!}\sum_{N \geq k}  \frac{(N+n_0-1)!}{(N-k)!} \Big((1-p) \frac{(\theta+\alpha)_n}{(\theta)_n} \Big)^N ,
\end{align*}
where we have simply rearranged all the terms.
By a simple change of variable we get
\begin{equation*}
    \begin{split}
      \pi_n (m_1, \ldots , m_k)   &=  \prod_{\ell=1}^k  (1-\alpha)_{m_\ell-1} (\theta+\alpha)_{n-m_\ell} \cdot p^{n_0}  \Big( \frac{-\alpha}{(\theta+\alpha)_n} \Big)^k \\
& \qquad\qquad \times
\frac{1}{k!(n_0-1)!}\sum_{N \geq 0}  \frac{(N+k+n_0-1)!}{N!} \Big((1-p)\frac{(\theta+\alpha)_n}{(\theta)_n} \Big)^{N+k} .
    \end{split}
\end{equation*}
Since $\alpha<0$ and $\alpha + \theta >0$, we have 
\[
0<(1-p)\frac{(\theta+\alpha)_n}{(\theta)_n}<1, 
\]
thus, it is easy to rearrange all the terms and recognize the probability mass function of a negative binomial random variable in the summation over $N$. More precisely, we have
\begin{align*}
\pi_n (m_1, \ldots , m_k)  &=   \prod_{\ell=1}^k  (1-\alpha)_{m_\ell-1} (\theta+\alpha)_{n-m_\ell} \cdot p^{n_0}  \Big( \frac{-\alpha}{(\theta+\alpha)_n} \Big)^k \\
& \qquad \times
\frac{(k+n_0-1)!}{k!(n_0-1)!}\Big((1-p)\frac{(\theta+\alpha)_n}{(\theta)_n} \Big)^k
\sum_{N \geq 0}  \binom{N+k+n_0-1}{N}\Big((1-p)\frac{(\theta+\alpha)_n}{(\theta)_n} \Big)^{N} \\
&=   \prod_{\ell=1}^k  (1-\alpha)_{m_\ell-1} (\theta+\alpha)_{n-m_\ell} \cdot p^{n_0}  \Big( \frac{-\alpha}{(\theta+\alpha)_n} \Big)^k
\Big((1-p)\frac{(\theta+\alpha)_n}{(\theta)_n} \Big)^k\\
& \qquad \times
\binom{k+n_0-1}{k} \Big( 1- (1-p)\frac{(\theta+\alpha)_n}{(\theta)_n}\Big)^{-k-n_0} \\
& \qquad\qquad\times
\sum_{N \geq 0}  \binom{N+k+n_0-1}{N}\Big((1-p)\frac{(\theta+\alpha)_n}{(\theta)_n} \Big)^{N}
\Big( 1- (1-p)\frac{(\theta+\alpha)_n}{(\theta)_n}\Big)^{k+n_0} .
\end{align*}
The sum in the last expression is equal to $1$ because this is the sum of the probability masses of a negative binomial distribution with dispersion parameter
$k+n_0$ and success probability  $1 -(1-p) (\theta+\alpha)_n/(\theta)_n$. Thus, we get
\[
\begin{split}
\pi_n (m_1, \ldots , m_k)  &=  \prod_{\ell=1}^k  (1-\alpha)_{m_\ell-1} (\theta+\alpha)_{n-m_\ell} \cdot p^{n_0}  \Big( \frac{-\alpha}{(\theta+\alpha)_n} \Big)^k
\Big((1-p)\frac{(\theta+\alpha)_n}{(\theta)_n} \Big)^k\\
& \qquad \times
\binom{k+n_0-1}{k} \Big( 1- (1-p)\frac{(\theta+\alpha)_n}{(\theta)_n}\Big)^{-k-n_0} \\
& = V_{n,k} \prod_{\ell=1}^k  (1-\alpha)_{m_\ell-1} (\theta+\alpha)_{n-m_\ell} 
\end{split}
\]
which is exactly in product form and $V_{n,k}$ is as in Equation \eqref{eq:EFPF_BB_NB}.

%% Proposition number of clusters BB
\subsection{Proof of Proposition \ref{prop_pred_poiss_nb}}

The distributions of $K_n$ follow by specializing Theorem~\ref{thm:Kn_general} with the specifications of the weights  $V_{n,k}$'s for the three models. Analogously, the posterior distributions of~$K_m^{(n)}$  follow from Theorem~\ref{thm:Kmn_general_V}.

%% asymptotic resulst
\subsection{Proof of Proposition \ref{prop:asymp}}

It follows from the general Theorem \ref{prop_asy_anyprior} by considering the case $\alpha <0$. The posterior distribution of $N\mid \bm{Z}^{(n)}$ easily follows from \eqref{eq:gamma_N_posterior}, by plugging-in the correct prior probability mass function $p_N$. Specifically, for $N \sim \text{Poisson}(\lambda)$, $p_N$ is a Poisson probability mass function; for $N \sim \text{NegBinomial}(n_0, \mu_0)$, $p_N$ is a negative binomial probability mass function.

%% process formulations of \textsc{bb} mixtures
\subsection{Proofs of hierarchical representations \eqref{model:bb_poisson_process} and \eqref{model:bb_nb_process}}

Using the hierarchical representation of the \textsc{bb} model  provided in Lemma \ref{lemma:bb_process_representation}, we first note that the Laplace functional of the \textsc{bb} model with parameters $(N, \alpha, \theta)$ can be written as \begin{equation}\label{eq:laplace_bb}
\mathcal{L}_{\tilde\mu}(g) = \E\left[ e^{-\int_\X  g (x) \tilde \mu (\D x )} \right] = \left[  \int_\X\int_0^1 e^{-s g(x) } \frac{1}{B(-\alpha, \theta+\alpha)} s^{-\alpha-1} (1-s)^{\theta+\alpha-1} \D s P_0(\D x) \right]^N , 
\end{equation}
for any measurable function $g: \X \to \R_+$. 

Consequently, in order to derive representation \eqref{model:bb_poisson_process}, it is sufficient to integrate out $N \sim  \textup{Poisson} (\lambda)$ in the expression of the Laplace functional \eqref{eq:laplace_bb}: 
\begin{equation*}
\begin{aligned}
     \mathcal{L}_{\tilde\mu}(g) &= e^{-\lambda} \sum_{N=0}^\infty \frac{1}{N!}\left[\lambda \int_\X\int_0^1  e^{-s g(x) } \frac{1}{B(-\alpha, \theta+\alpha)} s^{-\alpha-1} (1-s)^{\theta+\alpha-1} \D s P_0(\D x) \right]^N \\
     &= \exp\left\{ -\lambda \left[ 1 -  \int_\X \int_0^1 e^{-s g(x) } \frac{1}{B(-\alpha, \theta+\alpha)} s^{-\alpha-1} (1-s)^{\theta+\alpha-1} \D s P_0(\D x)\right] \right\}\\
     &= \exp\left\{ -\lambda \int_\X \int_0^1  (1 - e^{-s g(x) } ) \frac{1}{B(-\alpha, \theta+\alpha)} s^{-\alpha-1} (1-s)^{\theta+\alpha-1} \D s P_0(\D x) \right\}
\end{aligned}
\end{equation*}
and this is the Laplace functional of a \textsc{crm} with intensity measure as in \eqref{model:bb_poisson_process}. Similarly, for the negative binomial prior on $N$, we end up with the Laplace functional of a negative binomial process with parameters as in \eqref{model:bb_nb_process}.

%% Proof of posterior distribution of process under BB mixtures
%
\subsection{Proof of Corollary \ref{thm:post_process_bb_mixture}}

The corollary follows from Theorem \ref{thm:post_proc} by specializing the result for $\alpha<0$ in the case of a Poisson or a negative binomial prior for $N$. 
We first concentrate on the case $N \sim \textup{Poisson} (\lambda)$. Since in the posterior representation \eqref{eq:post_proc_dist_eq}, $\tilde \mu'$ and $\tilde\mu^*$ are independent, we need only to evaluate the Laplace functional of $\tilde\mu'$.
Observe that  $\mutilde' \mid N' \dequal \sum_{j=1}^{N'} q'_j \delta_{\tilde{X}_j}$ can be seen as the random measure associated with a \textsc{bb} model with parameters $(N',\alpha,\theta +n)$, moreover $N' + k \dequal N \mid \bm{Z}^{(n)}$. Therefore, by a suitable adaptation of  Equation \eqref{eq:laplace_bb}, the Laplace functional of $\mutilde'$, for an arbitrary measurable function $g : \X \to \R_+$, equals
\begin{equation} \label{eq:laplace_BB_new_features_N_general}
\begin{aligned}
  & \mathcal{L}_{\tilde\mu'}(g) = \E \left[ \E \left[ e^{- \int_\X g (x) \tilde \mu' (\D x )} \mid N' \right] \right] \\
   & \qquad = \sum_{t=0}^\infty \left[ \int_\X \int_0^1  e^{-s g(x) } \frac{1}{B (-\alpha, \theta+n+\alpha)} s^{-\alpha-1} (1-s)^{\theta+n+\alpha-1} \D s P_0(\D x) \right]^{t}  p_{N'}(t)
\end{aligned}
\end{equation}
where $p_{N'}$ denotes the probability mass function of $N'$.
The posterior of $N'$ has been characterized in Proposition \ref{prop:asymp}, where we showed that 
$N' \sim \textup{Poisson}\left(\lambda (1 - p_n(\theta, \alpha) \right)$, with $p_n(\theta, \alpha) = 1-(\theta+\alpha)_n / (\theta)_n$.
Thus, we get
\begin{equation*}
     \mathcal{L}_{\tilde\mu'}(g)  = \exp\left\{ - \int_\X \int_0^1  (1 - e^{-s g(x) } ) \frac{\lambda}{B (-\alpha, \theta+n+\alpha)} \frac{(\theta+\alpha)_n}{(\theta)_n} s^{-\alpha-1} (1-s)^{\theta+n+\alpha-1} \D s P_0(\D x) \right\},
\end{equation*}
that is the Laplace functional of a \textsc{crm} with intensity measure 
\[
\rho'(\D s)= \lambda\cdot \frac{\Gamma(\theta)}{\Gamma(-\alpha)\Gamma(\theta+\alpha)} s^{-\alpha-1} (1-s)^{n+\theta+\alpha-1} \D s.
\]

When $N$ has a negative binomial distribution, then $p_{N'}$ in \eqref{eq:laplace_BB_new_features_N_general} coincides with a negative binomial distribution as specified in \eqref{eq:posterior_Nprime_NB}, i.e., for $t = 0,1,\ldots$,
\[
p_{N'}(t) = \binom{t + n_0 + k -1}{t} (1- p^*_n (\theta, \alpha))^t (p^*_n (\theta, \alpha))^{n_0 + k},
\]
where we have set
\[
p^*_n (\theta, \alpha) := 1- (1-p_n (\theta, \alpha))\frac{\mu_0}{\mu_0+n_0}.
\]
Thus, simple calculations show that the Laplace functional of $\tilde\mu'$ boils down to
\begin{equation*}
\begin{split}
   & \mathcal{L}_{\tilde\mu'}(g) 
    =\left\{ 1 + \int_0^1 \int_\X (1 - e^{-s g(x) } ) \frac{1}{B(-\alpha, \theta+n+\alpha)} \frac{1-p^*_{n} (\theta, \alpha)}{p^*_n (\theta, \alpha)} s^{-\alpha-1} (1-s)^{\theta+n+\alpha-1} \D s P_0(\D x) \right\}^{-(n_0 + k)}.
\end{split}
\end{equation*}
We recognize that the previous Laplace functional is the one of a negative binomial process~$\NB(n_0 + k, \rho'; P_0)$, with 
\[
 \rho'(\D s) = \frac{\Gamma(\theta)}{\Gamma(-\alpha)\Gamma(\theta+\alpha)} \cdot\frac{\frac{\mu_0}{\mu_0+n_0}}{1-\frac{\mu_0}{\mu_0+n_0}\frac{(\theta+\alpha)_n}{(\theta)_n}} s^{-\alpha-1} (1-s)^{n+\theta+\alpha-1} \D s.
\]
Thus, with simple rearrangements of the terms one may realize that the intensity $\rho'$ is the one specified in the statement of the corollary, i.e., 
\[
 \rho'(\D s) = \frac{\Gamma(\theta)}{\Gamma(-\alpha)\Gamma(\theta+\alpha)} \cdot
 \frac{1}{n_0/\mu_0+ p_n (\theta, \alpha)}s^{-\alpha-1} (1-s)^{n+\theta+\alpha-1} \D s.
\]

\section{Simulation studies} \label{app:simulation}

\subsection{Simulation study A}
We consider a true data generating process with $H=600$  total features, having occurrence probabilities $\pi_k$, $k = 1,\ldots, H$,  set as follows: $200$ features have $\pi_k = 0.005$, $200$ features have  $\pi_k = 0.01$ and $200$ features have  $\pi_k = 0.015$. 
We generate a dataset with increasing dimensions for the training set $n \in \{50, 250, 1250\}$. 
In Figure \ref{fig:rare_knr_simulation_A}, we {apply the visual approach to} model-checking for the case $n=250$. As it is evident from the plots, the mixtures of \textsc{ibp}s are definitely not a good model for such data. Conversely, we argue that the mixtures of \textsc{bb}s can be assumed to be correctly specified here. {This preference for the mixtures of \textsc{bb}s is further supported by the comparison of deviances: $D(\hat{\bm \theta}) = 16406.7$ for the \textsc{bb} model versus $D(\hat{\bm \theta}) = 16603.4$ for the \textsc{ibp} model.} Therefore, we select the mixtures of \textsc{bb}s for the inference and prediction. {For completeness, Figure \ref{fig:rare_knr_simulation_A_credible_intervals} reports the $95\%$ credible intervals of $K_1,\ldots,K_n$ and of $K_{n,r}$, $r \geq 1$, for the Poisson and two examples of negative binomial mixture of \textsc{bb}s. It is apparent that the width of the credible intervals increases as the prior variance on $N$  increases.}

\begin{figure}[tbp]
    \centering    
    \includegraphics[width = 0.49\linewidth]{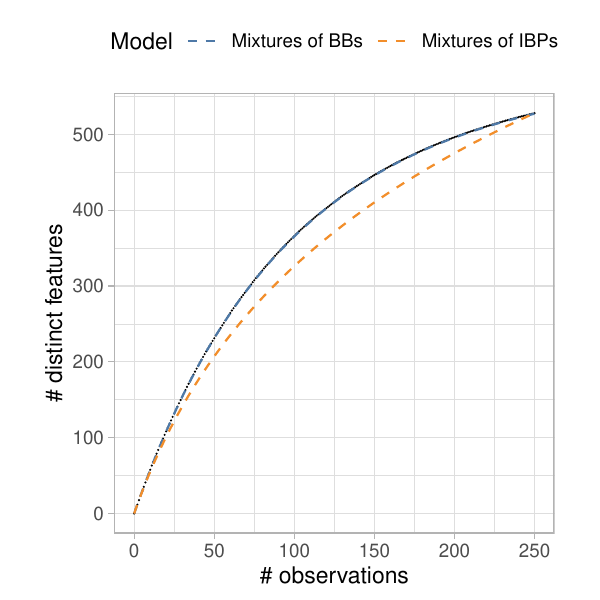}    
    \includegraphics[width = 0.49\linewidth]{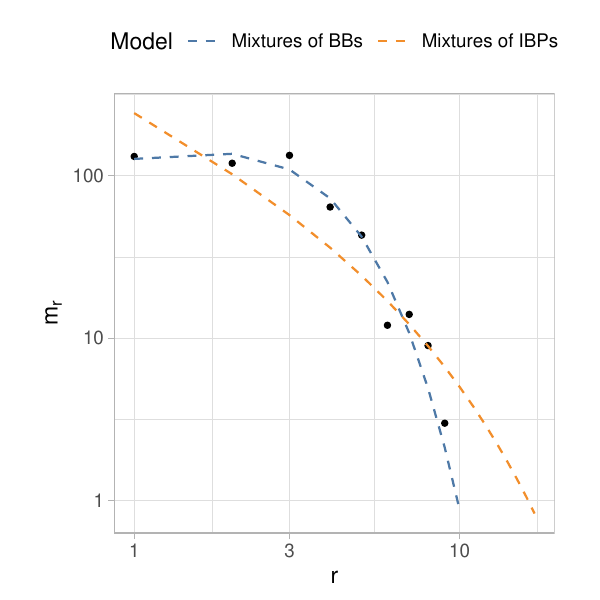}
    
    \caption{Case $n=250$. Left panel: the empirical accumulation curve (black dots) and the rarefaction curve of the models (blue and orange dashed lines). Right panel: the observed values of $K_{n,r}$ (black dots) compared with the {expected curve $\E(K_{n,r})$ of the models (blue and orange dashed lines). The right plot is in log-log scale; the orange (resp. blue) curves are identical for all the mixtures of \textsc{ibp}s (resp. \textsc{bb}s).}}
    \label{fig:rare_knr_simulation_A}
\end{figure}

\begin{figure}[tbp]
    \centering    
    \includegraphics[width = 0.49\linewidth]{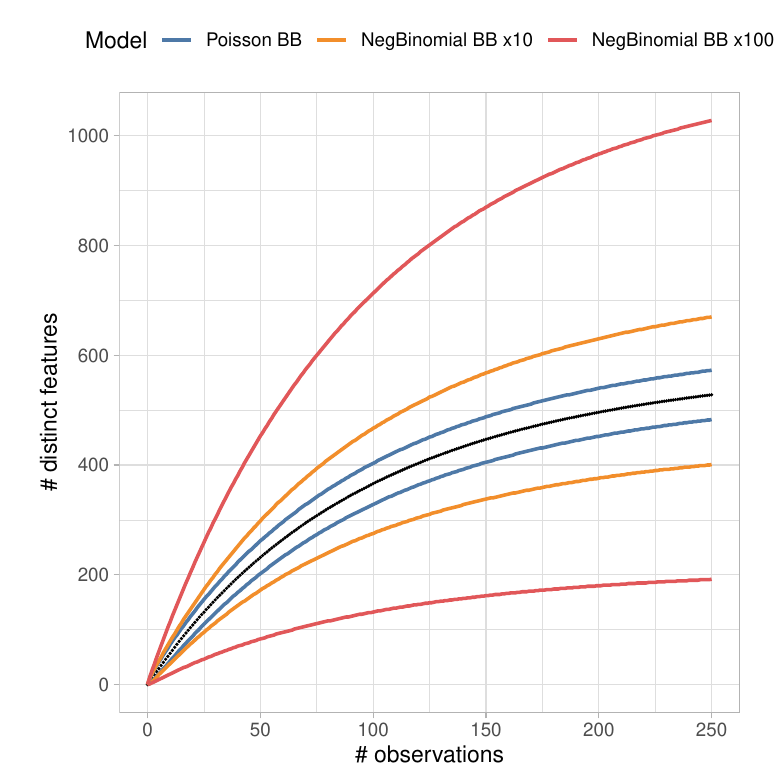}    
    \includegraphics[width = 0.49\linewidth]{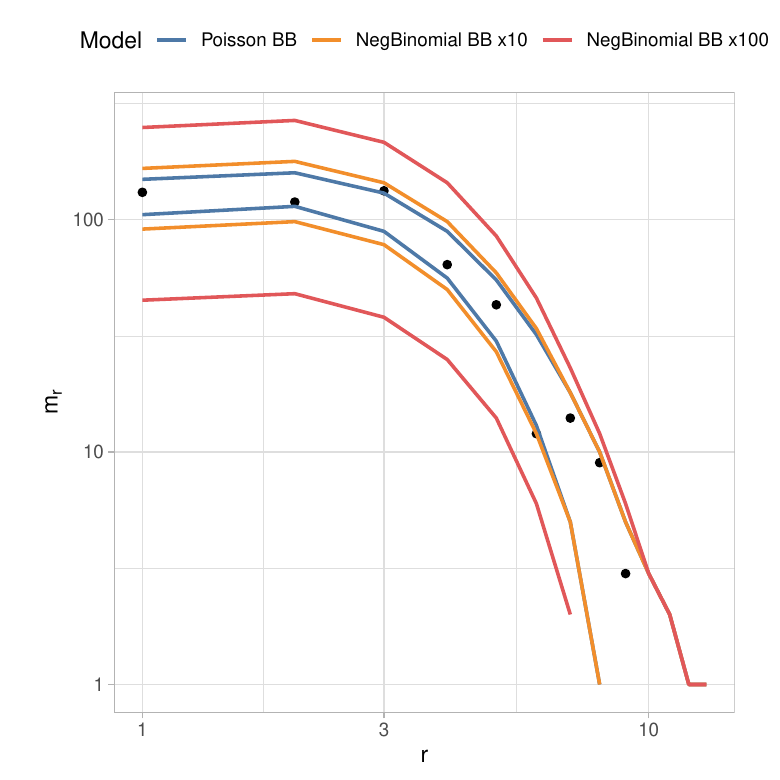}
    
    \caption{{Case $n=250$. Left panel: the empirical accumulation curve (black dots) and the credible intervals (delimited by colored lines) of $K_1,\ldots,K_n$ for the Poisson mixture of \textsc{bb}s and two examples of negative binomial mixture of \textsc{bb}s. Right panel: the observed values of $K_{n,r}$ (black dots) and the credible intervals (delimited by colored lines) of $K_{n,r}$  for the same mixtures of \textsc{bb}s. The right plot is in log-log scale.}}
    \label{fig:rare_knr_simulation_A_credible_intervals}
\end{figure}

Focusing on the prediction of the number of unseen features in an additional sample of increasing size, we compare the Poisson mixture of \textsc{bb}s, the negative binomial mixture of \textsc{bb}s, the frequentist estimator (Chao) in \cite{Cha14} and a variation of the {Good--Toulmin} estimators (\textsc{gt}) in \cite{Chak19}. For the negative binomial mixture of \textsc{bb}s, we analyse two distinct choices for the prior variance of $N$, specifically $\Var(N) = \mu_0 \times c$, with {$c \in \{10, 100\}$}. We report the comparison in Figure \ref{fig:extr_simulation_A}, for $n \in \{50, 250\}$. In the case $n=50$, the point-wise estimates produced by our mixtures of \textsc{bb}s are rather good, definitely outperforming the ones obtained
with the frequentist estimator in \cite{Cha14}, which underestimates the extrapolation curve. While the {Good--Toulmin} estimates seem pretty reliable in the case $n=50$, such an estimator shows bad predictive performance and well-known stability issues for $n=250$. Moreover, the mixtures of \textsc{bb}s allow us to quantify the uncertainty around the point estimates through credible intervals: this is a remarkable advantage of our Bayesian framework. Observe that the credible  bands contain the observed curve in the test set; for $n=50$ it is evident how the negative binomial mixtures account for larger dispersion than the Poisson mixture, having wider credible intervals. For completeness, the sanity check in the larger sample $n=1250$, where $k = 599$ features are observed, is satisfactory: the credible interval of $K_m^{(n)} + k \mid \bm{Z}^{(n)}$, for $m = 300$, is $[599, 601]${, for all the considered mixtures of \textsc{bb}s}. Finally, we remark that, while both the frequentist estimators are exclusively designed for the specific extrapolation problem, our model-based approach through mixtures of \textsc{bb}s offers a coherent and
self-contained framework for a number of inference problems, where the prediction of the number of unseen features is just one available example. 

\begin{figure}[tbp]
    \centering    
    \includegraphics[width = \linewidth]{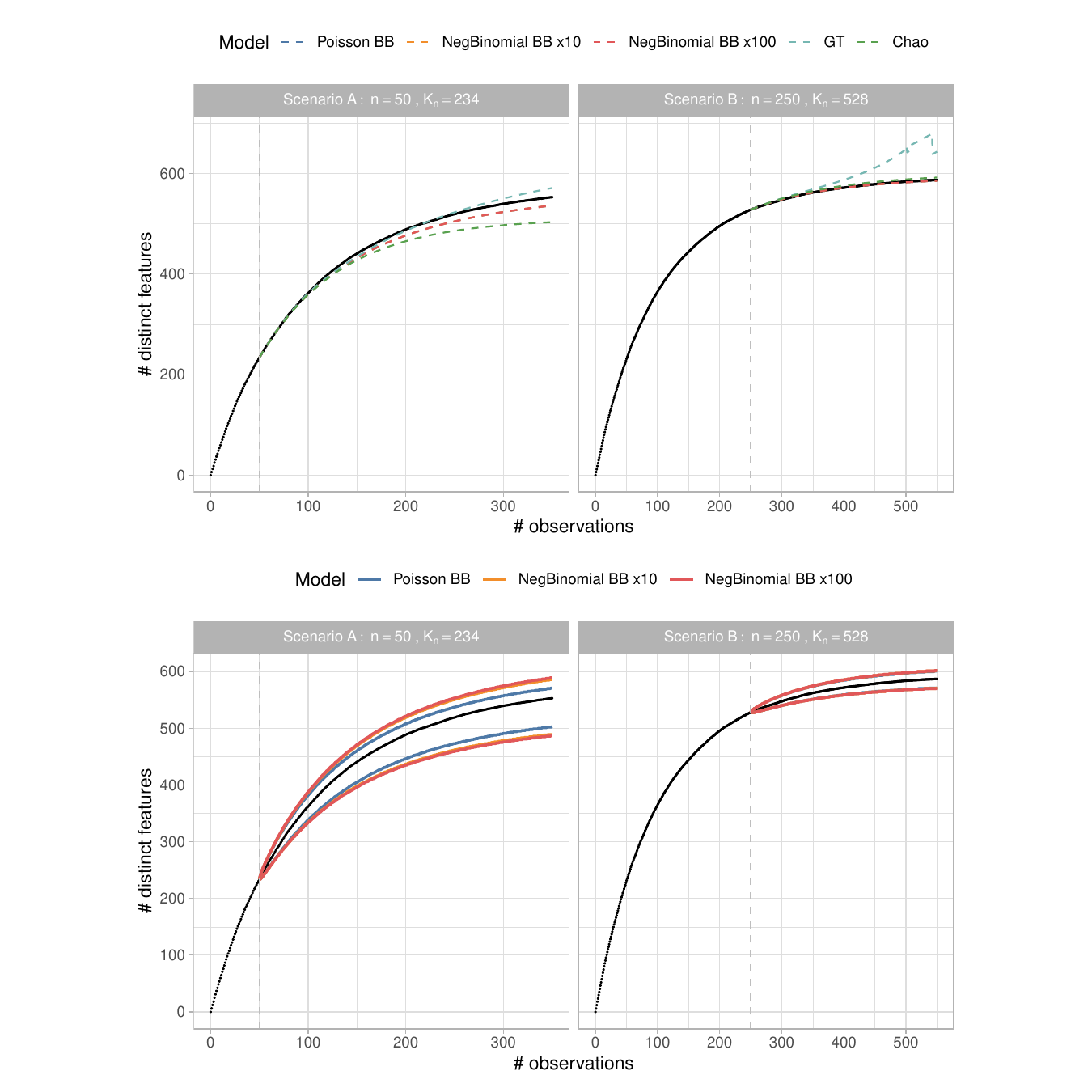}    
    
    \caption{Top row: the accumulation curve $K_m^{(n)} + k \mid \bm{Z}^{(n)}$ under the mixtures of \textsc{bb}s, the {Good--Toulmin} estimator and the Chao estimator for $n = 50$ (left) and $n = 250$ (right). Bottom row: the $95\%$ credible intervals of $K_m^{(n)} + k \mid \bm{Z}^{(n)}$, for $n = 50$ (left) and $n = 250$ (right), for the Poisson and the negative binomial mixtures of \textsc{bb}s. %The extrapolation horizon is $m=1,\ldots, 300$. The black dots represent the empirical curve, with 
    The grey vertical lines indicate the training set.}
    \label{fig:extr_simulation_A}
\end{figure}

Additionally,  mixtures of  \textsc{bb}s also allow estimation of the richness. We refer to Figure \ref{fig:richness_points_simulation_A} for an insightful illustration on the richness estimation via the Poisson mixture and two negative binomial mixtures with $\Var(N) = \mu_0 \times c$, for {$c \in \{10, 100\}$}. {For the purpose of discussion, we also consider the standard \textsc{bb} model: interpret it as a mixture of \textsc{bb}s where the prior on $N$ is a point mass on the value $N$.} Specifically, Figure \ref{fig:richness_points_simulation_A} shows, for increasing sample sizes $n \in \{50,250,1250\}$, the posterior mean of the species richness $N$, for different choices of the parameters of the prior distribution on $N$. In particular, for each model, we estimate the parameters following the usual empirical Bayes procedure described in Section \ref{sec:parameter_elic} (referred to as EB in the figure); moreover, we fix different choices for the prior mean of $N$, i.e., $\E(N) \in \{200,400,800\}$, and we estimate the parameters $\alpha$ and $\theta$ as discussed in Section \ref{sec:parameter_elic}. We observe that, as the size $n$ of the training set increases, the richness estimates shrink towards the true value $H = 600${, under proper mixtures of \textsc{bb}s}. We also highlight that when prior guesses on $N$ are far from the truth, e.g., $\E(N) \in \{200, 400, 800\}$, negative binomial mixtures of \textsc{bb}s with larger variances produce richness estimates that are closer to the true value $H$ than the ones produced by the Poisson mixture, due to the greater flexibility of the negative binomial prior, which gives less weight to prior guesses. {In contrast, under the standard \textsc{bb} model, the prior guess for $N$ deterministically fixes the richness, making it impossible to correct a poor prior guess through posterior inference.}

\begin{figure}[tbp]
    \centering    
    \includegraphics[width = \linewidth]{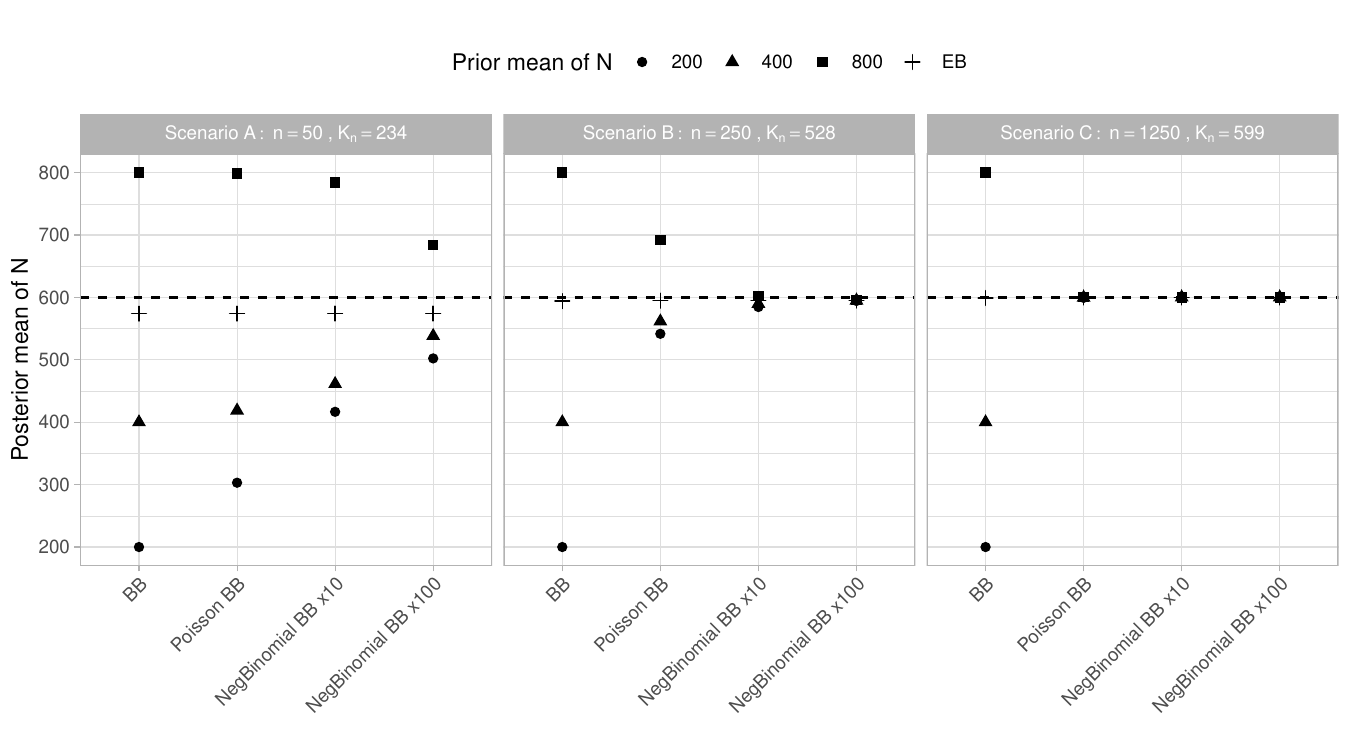}    
    
    \caption{Posterior mean of the richness  $N$, for increasing sample sizes $n \in \{50, 250, 1250\}$, and different values of the parameters of the models. The dashed lines indicate the true richness $H=600$.}
    \label{fig:richness_points_simulation_A}
\end{figure}

Finally, in Figure  \ref{fig:richness_distr_simulation_A}, we show the posterior distributions of the richness $N$, reporting (i) the empirical Bayes approach for parameters elicitation and (ii) the prior mean of $N$ equal to $ 400$. {First, we comment on the inference produced by the proper mixtures of \textsc{bb}s. In case (i),} we observe that such posterior distributions give high probability mass to the true richness $H = 600$; for smaller sample size $n=50$, we can highlight the higher dispersion of the posterior distribution of $N$ for the negative binomial mixtures. In case (ii), when the prior guess of $N$ is bad, e.g., $\E(N) = 400$, the inference gets worse as the sample size decreases, e.g., $n=50$. This is clearly expected since the posterior of $N$ gives more weight to the contribution of the prior, which brings an incorrect guess. However, it is interesting to note that {this posterior under the} negative binomial mixture with larger prior variance for $N$, i.e., {$\Var(N) = 100 \times \mu_0$, is significantly moving towards the true richness} even for $n=50$. {Second, we discuss the undesired behavior of the standard \textsc{bb} model. In case (i), although the posterior of $N$ concentrates on a value relatively close to the true one $H=600$, the lack of uncertainty in the posterior makes $H$ not plausible under the standard \textsc{bb} model. As a result, the inference is highly unreliable. 
The same issue emerges in case (ii), when the prior guess of $N$ is bad (equal to $400$). Again, note that under the standard \textsc{bb} model, the posterior of $N$ remains concentrated at $400$ and cannot learn from data as the sample size increases.}

\begin{figure}[tbp]
    \centering    
    \includegraphics[width = \linewidth]{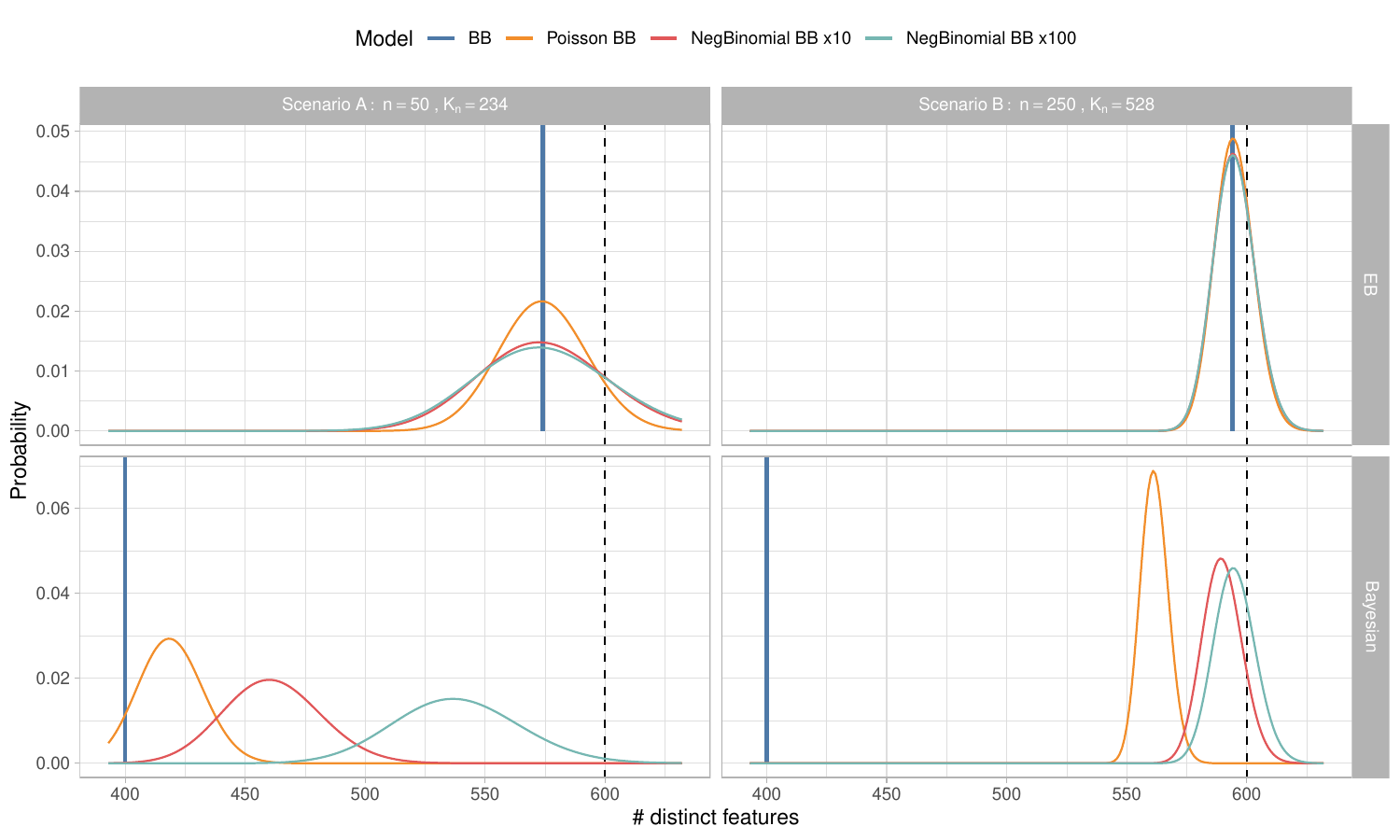}    
    
    \caption{Posterior distribution of the richness $N$ for increasing sample sizes $n \in \{50, 250\}$. Top row: parameters elicited according to the empirical Bayes approach; bottom row:  $\E(N) = 400$, with $\alpha$ and $\theta$ chosen as  in Section \ref{sec:parameter_elic}. The vertical black lines indicate the true richness $H=600$.}
    \label{fig:richness_distr_simulation_A}
\end{figure}

\subsection{Simulation study B}
We assume a generating mechanism such that the feature occurrence probabilities $\pi_k$'s are set as $\pi_k = 1/k$, for $k =1,\ldots,H$, with $H = 10^6$ total features. The choice of a large $H$ is intended to mimic a scenario where the number of features is infinite. We generate a dataset of observations and we consider  increasing dimensions of the training set $n \in \{10, 50, 250\}$. In Figure \ref{fig:rare_knr_simulation_B} we report the results of the suggested {visual}
model-checking procedure for the case $n = 50$. Differently from simulation study A, the mixtures of \textsc{bb}s are not properly fitting the data. Instead, we argue that
the mixtures of \textsc{ibp}s can be assumed to be correctly specified. {A more quantitative assessment of model fit can be obtained by comparing the deviances, with  $D(\hat{\bm \theta}) = 5398.7$ for the \textsc{bb} model versus $D(\hat{\bm \theta}) = 5094.5$ for the \textsc{ibp} model. These results are consistent with the conclusions drawn from the visual inspection of Figure \ref{fig:rare_knr_simulation_B}. }
{For completeness, Figure \ref{fig:rare_knr_simulation_B_credible_intervals} reports the $95\%$ credible intervals of $K_1,\ldots,K_n$ and of $K_{n,r}$, $r \geq 1$, for two examples of gamma mixture of \textsc{ibp}s.  We note that  the width of the credible intervals increases as the prior variance on $\gamma$  increases.}

\begin{figure}[tbp]
    \centering    
    \includegraphics[width = 0.49\linewidth]{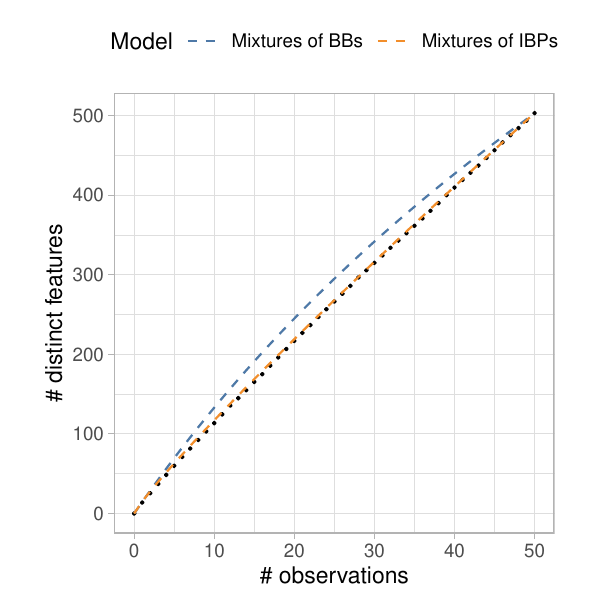}    
    \includegraphics[width = 0.49\linewidth]{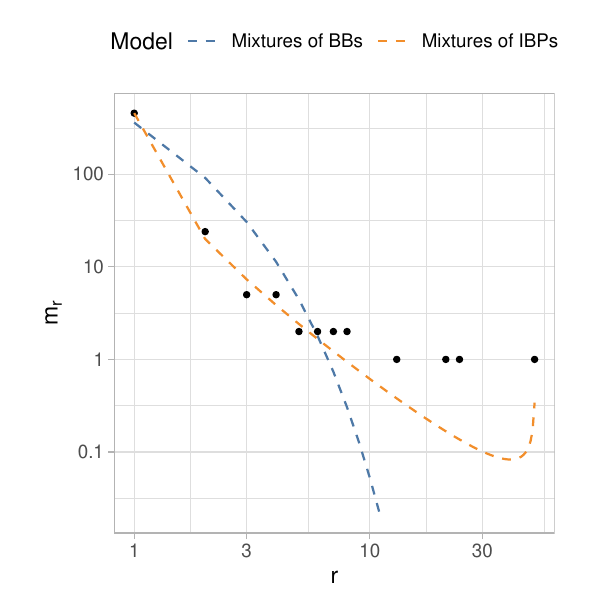}
    
    \caption{Case $n=50$.  Left panel: the empirical accumulation curve (black dots) and the rarefaction curve of the models (blue and orange dashed lines). Right panel: the observed values of $K_{n,r}$ (black dots) compared with the {expected curve $\E(K_{n,r})$ of the models (blue and orange dashed lines). The right plot is in log-log scale; the orange (resp. blue) curves are identical for all the mixtures of \textsc{ibp}s (resp. \textsc{bb}s).}}
    \label{fig:rare_knr_simulation_B}
\end{figure}

\begin{figure}[tbp]
    \centering    
    \includegraphics[width = 0.49\linewidth]{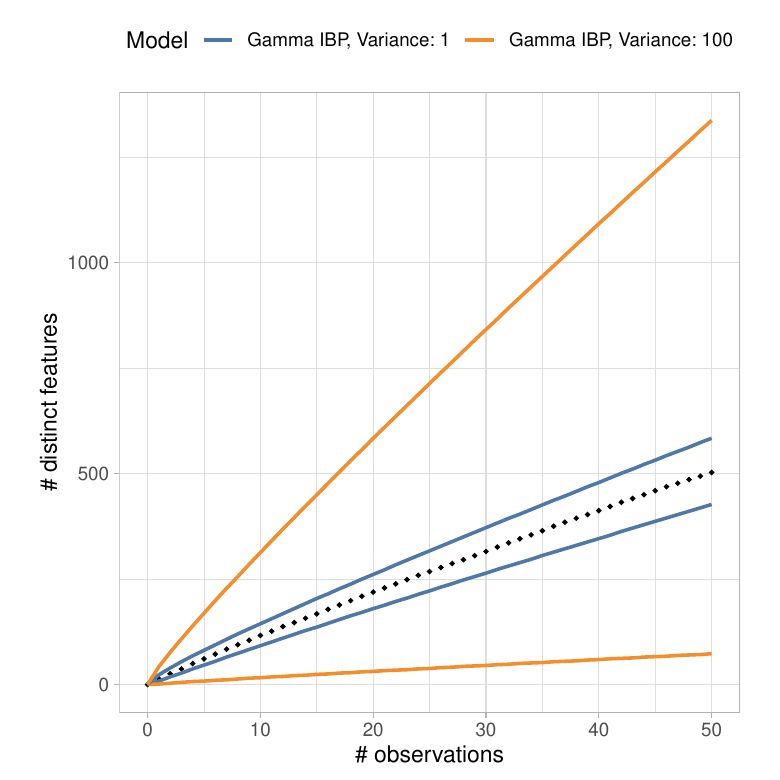}    
    \includegraphics[width = 0.49\linewidth]{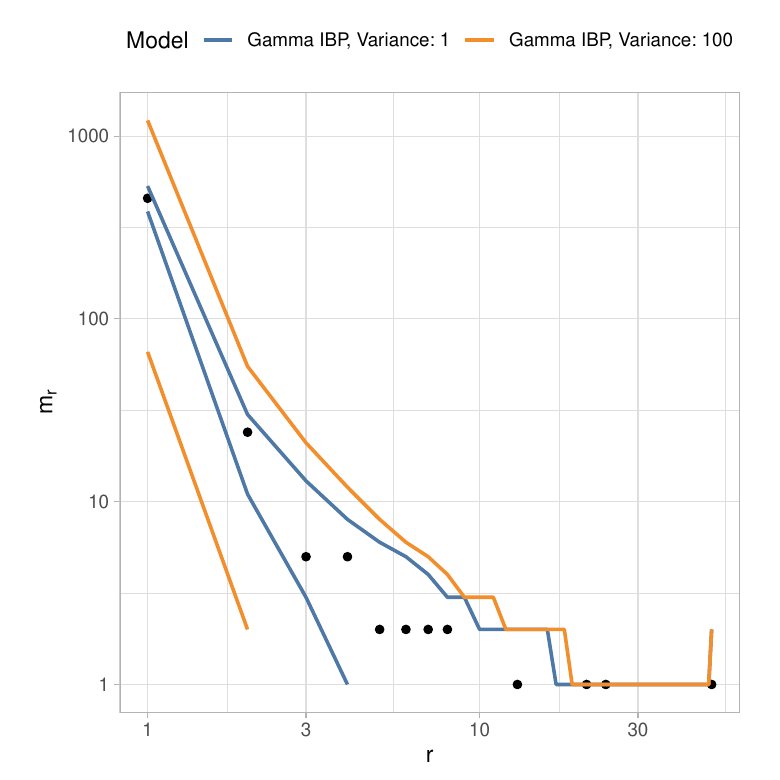}
    
    \caption{{Case $n=50$. Left panel: the empirical accumulation curve (black dots) and the credible intervals (delimited by colored lines) of $K_1,\ldots,K_n$ for two examples of gamma mixture of \textsc{ibp}s. Right panel: the observed values of $K_{n,r}$ (black dots) and the credible intervals (delimited by colored lines) of $K_{n,r}$  for the same mixtures of \textsc{ibp}s. The right plot is in log-log scale.} }
    \label{fig:rare_knr_simulation_B_credible_intervals}
\end{figure}

We address the prediction of the number of unseen features in an additional sample of increasing size, comparing the gamma mixture of \textsc{ibp}s and a variation of the {Good--Toulmin} estimators (\textsc{gt}) in \cite{Chak19}. We do not consider the estimator in \cite{Cha14} here, since it is specifically designed for situations where the assumption of finite richness is plausible. For the gamma mixture of \textsc{ibp}s, we consider two distinct choices for the prior variance of $\gamma$, i.e., {$\Var(\gamma) \in \{1, 100\}$}. {For the purpose of discussion, we also analyze the standard \textsc{ibp} model.} We report the comparison in Figure \ref{fig:extr_simulation_B}, for $n \in \{10, 50, 250\}$. Starting with the {Good--Toulmin} estimates, it is evident how the prediction curve catches the observed curve just for very limited horizons $m$. On the other hand, all the gamma mixtures are able to capture the growth rate of $K^{(n)}_m$ for increasing values of $m$. Remarkably, a larger prior variance of $\gamma$, e.g., $\Var(\gamma) = 100$, results in wider credible intervals for the prediction, allowing for a more conservative uncertainty quantification.{ Indeed, under the standard \textsc{ibp} model, the observed extrapolation curve is not always contained in the credible bands, leading to a less reliable prediction of the variability of the phenomenon. Remind that the \textsc{ibp} model corresponds to the limiting case $\Var(\gamma) \to 0$; this clearly justifies the need of the mixtures of \textsc{ibp}s.}

\begin{figure}[tbp]
    \centering    
    \includegraphics[width = \linewidth]{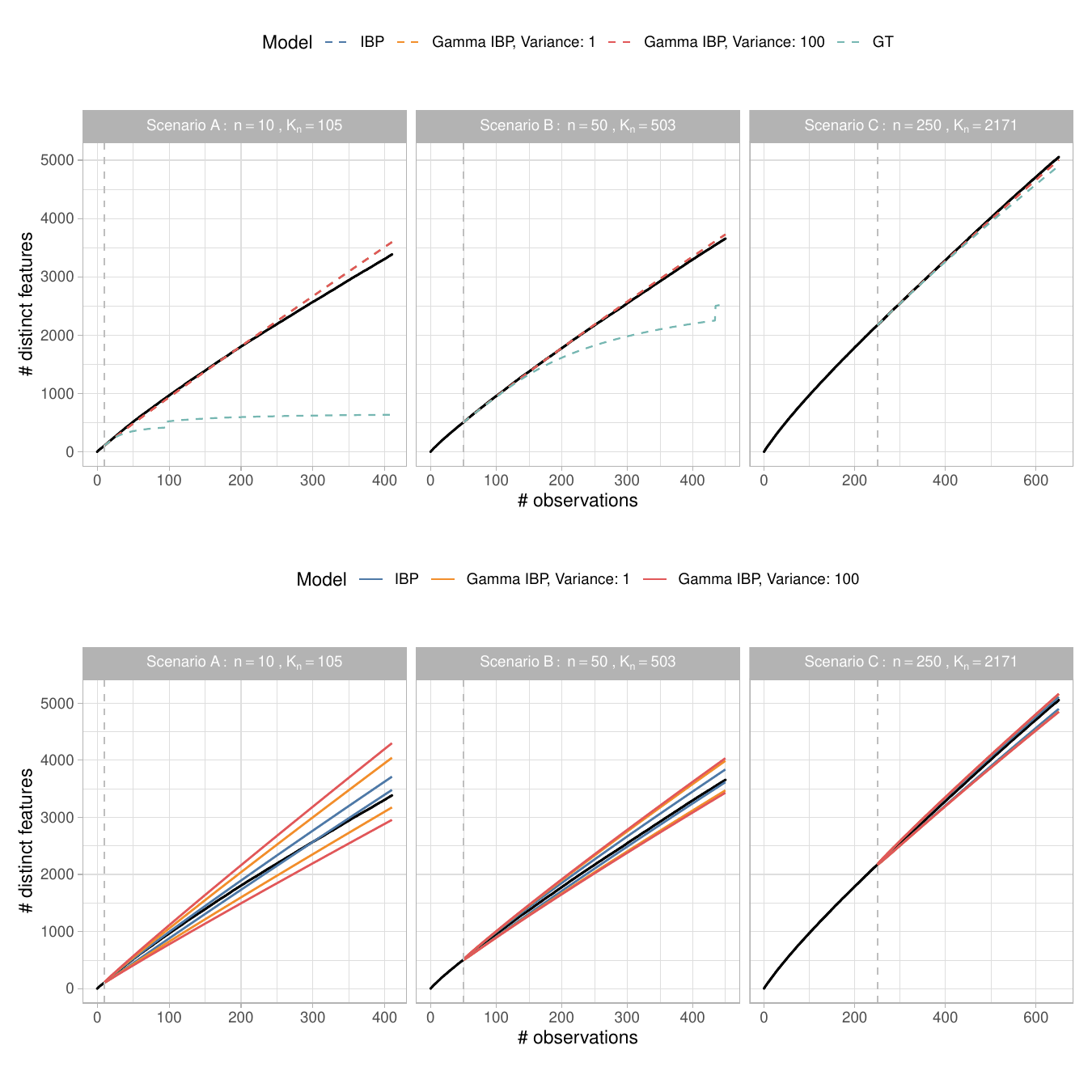}    
    
    \caption{Top row: the accumulation curve $K_m^{(n)} + k \mid \bm{Z}^{(n)}$ under the gamma mixtures of \textsc{ibp}s and the {Good--Toulmin} estimator for $n = 10$ (left), $n = 50$ (center) and $n=250$ (right). Bottom row: the $95\%$ credible intervals of $K_m^{(n)} + k \mid \bm{Z}^{(n)}$, for $n = 10$ (left), $n = 50$ (center) and $n=250$ (right), for the gamma mixtures of \textsc{ibp}s.
    }
    \label{fig:extr_simulation_B}
\end{figure}

\subsection{Simulation study C}

Here, we report a simulation study where the data are generated from the \textsc{bb} model, with total number of features equal to $H=500$. Specifically, the feature occurrence probabilities are drawn as $\pi_k \iid \dbeta(1, 100)$, $k = 1,\ldots, H$. We generate a dataset of observations and we consider increasing dimensions of the training set $n \in \{200, 1000, 5000\}$. The sample sizes are larger than the other simulation studies since the growth of the accumulation curve is rather slow. 

Figure \ref{fig:rare_knr_simulation_beta} shows the {visual} model-checking for the case $n=1000$. As expected, the mixtures of \textsc{ibp}s cannot explain the observed data. Conversely, it is reasonable to assume the mixtures of \textsc{bb}s to be correctly specified for this dataset. {This claim is further supported by the comparison of deviances, with the \textsc{bb} model yielding $D(\hat{\bm \theta}) = 53208.8$ compared to $D(\hat{\bm \theta}) = 53392.2$ for the \textsc{ibp} model.} {Based on both visual and quantitative evidence}, we focus on the mixtures of \textsc{bb}s for the inference and prediction.
{As in the previous simulation studies, Figure \ref{fig:rare_knr_simulation_beta_credible_intervals} reports the $95\%$ credible intervals of $K_1,\ldots,K_n$ and of $K_{n,r}$, $r \geq 1$, for the Poisson and two examples of negative binomial mixture of \textsc{bb}s.  }

\begin{figure}[tbp]
    \centering    
    \includegraphics[width = 0.49\linewidth]{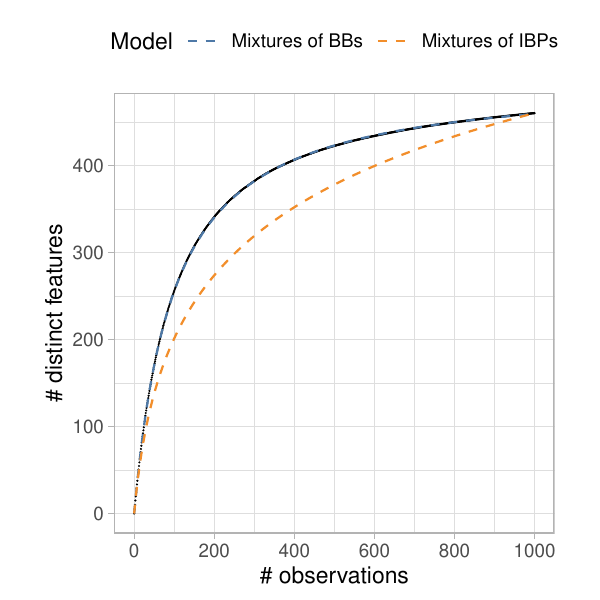}    
    \includegraphics[width = 0.49\linewidth]{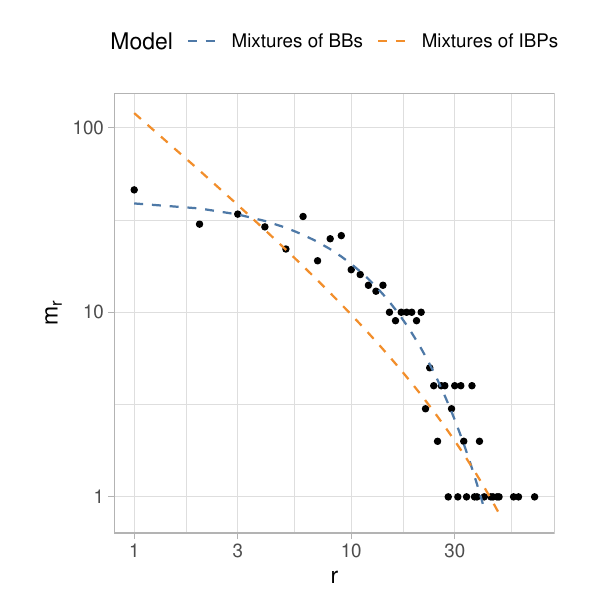}
    
    \caption{Case $n=1000$.
    Left panel: the empirical accumulation curve (black dots) and the rarefaction curve of the models (blue and orange dashed lines). Right panel: the observed values of $K_{n,r}$ (black dots) compared with the {expected curve $\E(K_{n,r})$ of the models (blue and orange dashed lines). The right plot is in log-log scale; the orange (resp. blue) curves are identical for all the mixtures of \textsc{ibp}s (resp. \textsc{bb}s).}}
    \label{fig:rare_knr_simulation_beta}
\end{figure}

\begin{figure}[tbp]
    \centering    
    \includegraphics[width = 0.49\linewidth]{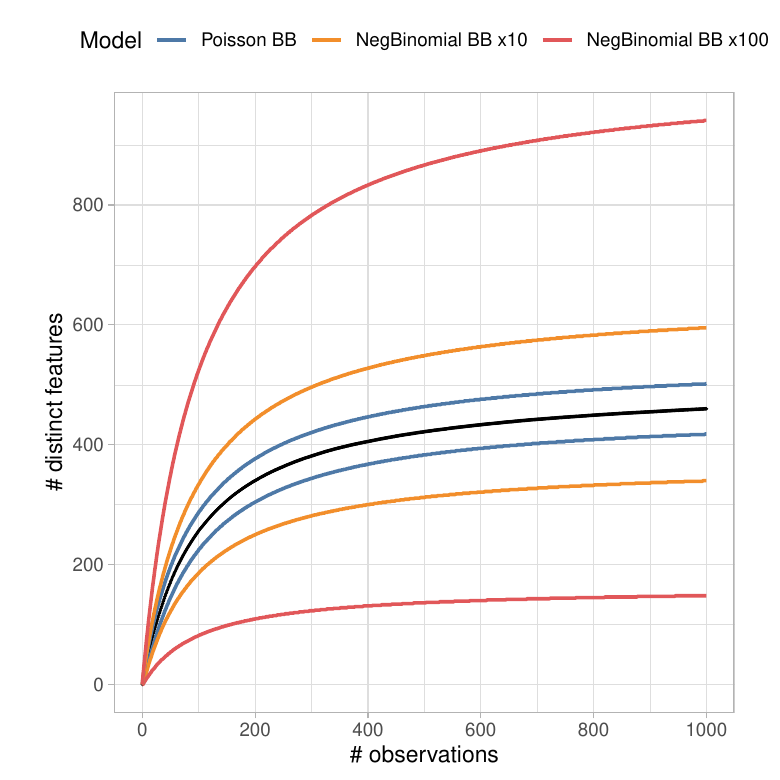}    
    \includegraphics[width = 0.49\linewidth]{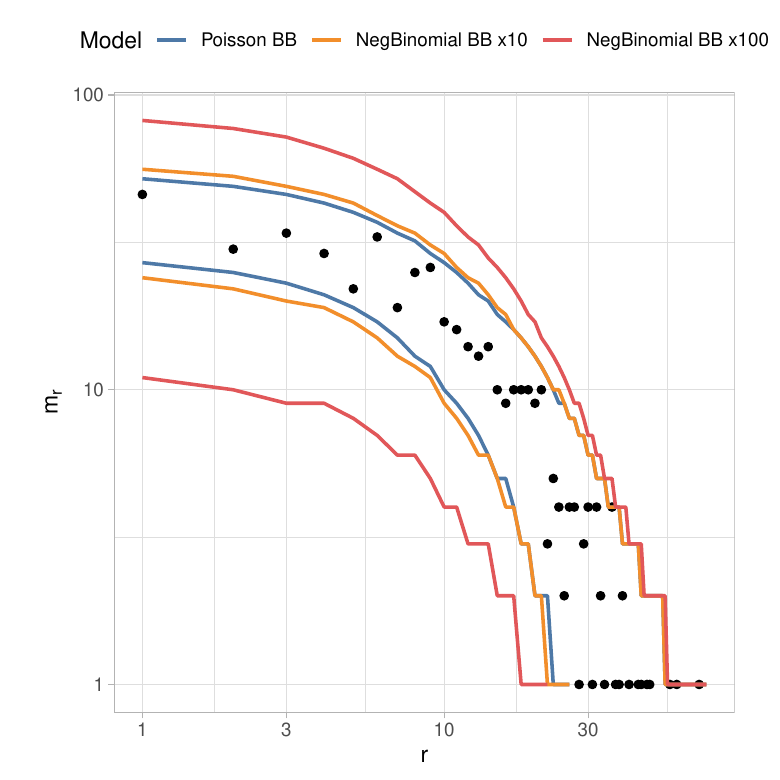}
    
    \caption{
    {Case $n=1000$. Left panel: the empirical accumulation curve (black dots) and the credible intervals (delimited by colored lines) of $K_1,\ldots,K_n$ for the Poisson mixture of \textsc{bb}s and two examples of negative binomial mixture of \textsc{bb}s. Right panel: the observed values of $K_{n,r}$ (black dots) and the credible intervals (delimited by colored lines) of $K_{n,r}$  for the same mixtures of \textsc{bb}s. The right plot is in log-log scale.}}
    \label{fig:rare_knr_simulation_beta_credible_intervals}
\end{figure}

Similarly to simulation study A, we compare the Poisson mixture of \textsc{bb}s, the negative binomial mixture of \textsc{bb}s, the frequentist estimator (Chao) in \cite{Cha14} and a variation of the {Good--Toulmin} estimators (\textsc{gt}) in \cite{Chak19}, in terms of prediction of the number of unseen features in an additional sample of increasing size. For the negative binomial mixture of \textsc{bb}s, we analyse two distinct choices for the prior variance of $N$, specifically $\Var(N) = \mu_0 \times c$, with {$c \in \{10, 100\}$}. We report the comparison in Figure \ref{fig:extr_simulation_beta}, for $n \in \{200, 1000\}$. For larger sample sizes, e.g., $n = 1000$, all the models seem to produce reliable predictions, at least for the analyzed horizon $m=1,\ldots, 500$. Significant differences are observed for the smaller sample size $n=200$, where the point-wise estimates (the expected values of $K_m^{(n)} + k$, for $m=1,\ldots, 500$) produced by our mixtures of \textsc{bb}s accurately predict the observed curve in the test set, with the credible intervals nicely quantifying the uncertainty around such estimates. Indeed, the observed curve is always contained in the credible bands. On the other hand, the {Good--Toulmin} estimates show their well-known stability issues, while the estimator in \cite{Cha14} clearly underestimates the observed curve.

\begin{figure}[tbp]
    \centering    
    \includegraphics[width = \linewidth]{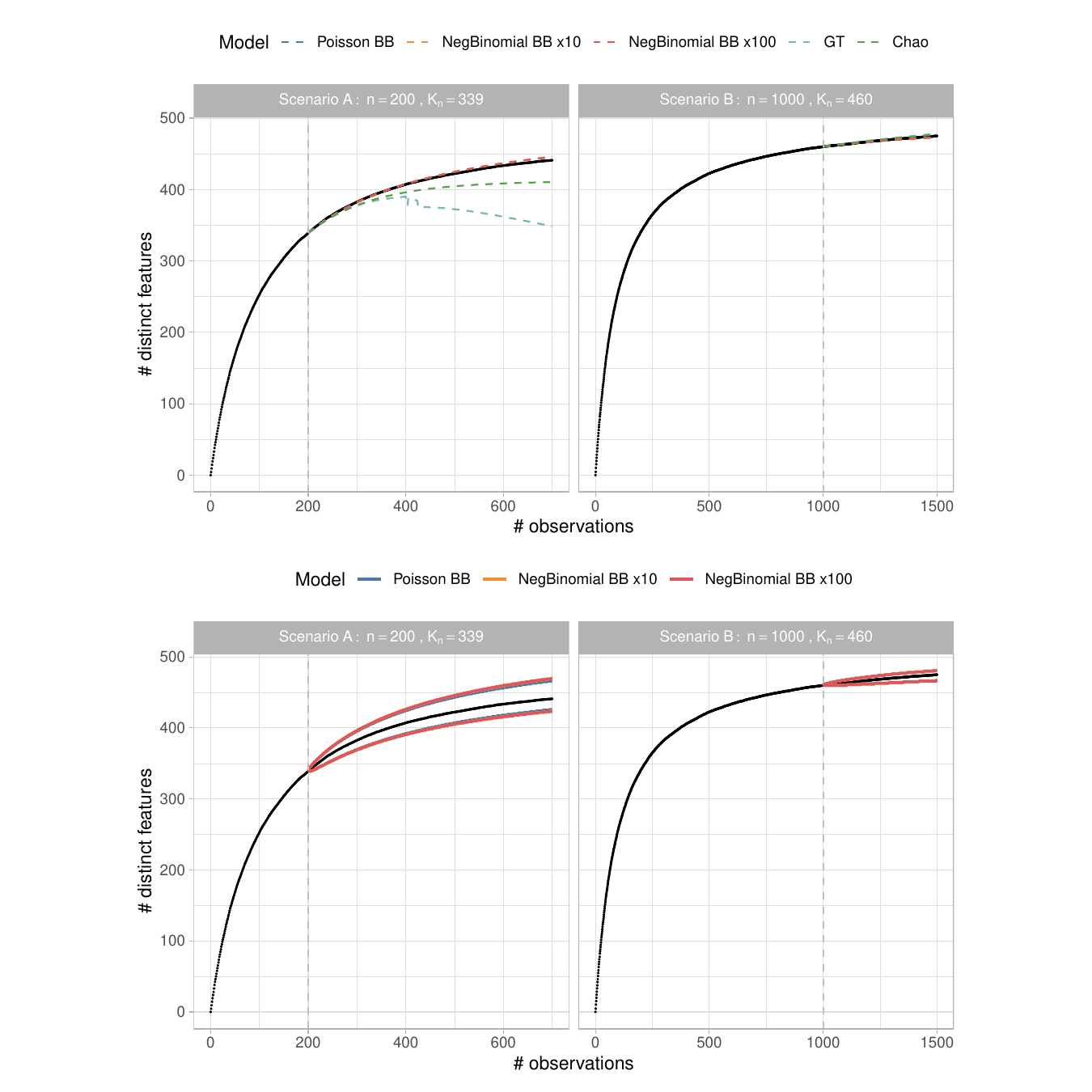}    
    
    \caption{Top row: the accumulation curve $K_m^{(n)} + k \mid \bm{Z}^{(n)}$ under the mixtures of \textsc{bb}s, the {Good--Toulmin} estimator and the Chao estimator for $n = 200$ (left) and $n = 1000$ (right). Bottom row: the $95\%$ credible intervals of $K_m^{(n)} + k \mid \bm{Z}^{(n)}$, for $n = 200$ (left) and $n = 1000$ (right), for the Poisson and the negative binomial mixtures of \textsc{bb}s. 
    The grey vertical lines indicate the training set.}
    \label{fig:extr_simulation_beta}
\end{figure}

For the analyzed dataset, the model-checking in Figure \ref{fig:rare_knr_simulation_beta} suggests the correct specification of the mixtures of \textsc{bb}s. Consequently, it is reasonable to assume the finiteness of the richness and address its estimation. Figure \ref{fig:richness_points_simulation_beta} shows an expected, still remarkable, illustration on the richness estimation via the Poisson mixture and two negative binomial mixtures with $\Var(N) = \mu_0 \times c$, with {$c \in \{10, 100\}$}. Specifically, for increasing sample sizes $n \in \{200,1000,5000\}$, it reports the posterior mean of the species richness $N$, for different choices of the parameters of the prior distribution on $N$. In particular, for each model, we estimate the parameters following the usual empirical Bayes procedure (referred to as EB in the figure); moreover, we fix different choices for the prior mean of $N$, i.e., $\E(N) \in \{200,400,800\}$, and we estimate the parameters $\alpha$ and $\theta$ as discussed in Section \ref{sec:parameter_elic}. We observe that, as the size $n$ of the training set increases, the richness estimates shrink towards the true value $H = 500$. We also highlight that when prior guesses on $N$ are substantially wrong, e.g., $\E(N) \in \{200, 400, 800\}$, negative binomial mixtures of \textsc{bb}s with larger variances produce richness estimates that are closer to the true value $H$ than the ones produced by the Poisson mixture, due to the greater flexibility of the negative binomial prior, which is able to give less weight to wrong prior guesses. {In contrast, under the standard \textsc{bb} model, the prior guess for $N$ deterministically fixes the richness, and thus posterior inference on $N$ is impossible.}

\begin{figure}[tbp]
    \centering    
    \includegraphics[width = \linewidth]{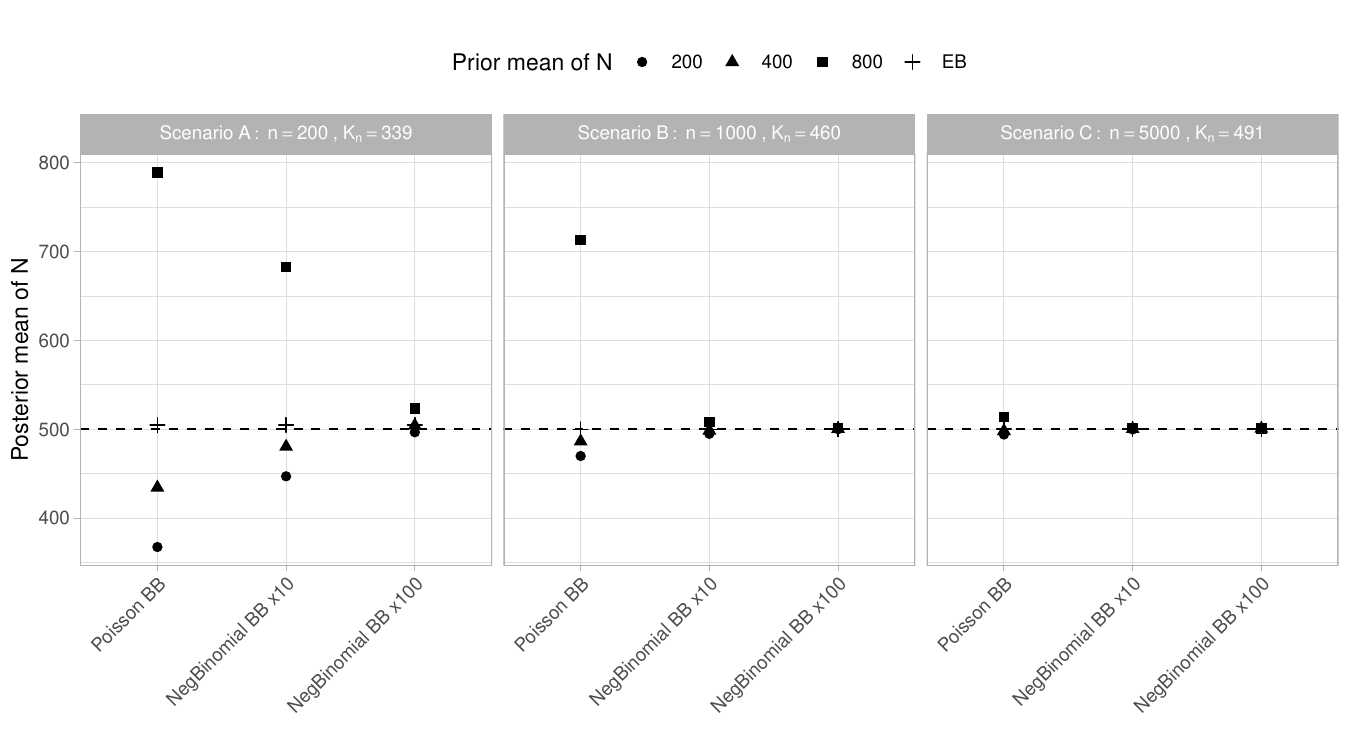}    
    
    \caption{Posterior mean of the richness  $N$, for increasing sample sizes $n \in \{200,1000,5000\}$, and different values of the parameters of the models. The dashed lines indicate the true richness $H=500$.}
    \label{fig:richness_points_simulation_beta}
\end{figure}

Finally, in Figure  \ref{fig:richness_distr_simulation_beta}, we show the posterior distributions of the richness $N$, reporting (i) the empirical Bayes approach for parameters elicitation and (ii) the prior mean of $N$ equal to $400$. Commenting on the inference produced in case (i), we observe that such posterior distributions give high probability mass to the true richness $H = 500$; for smaller sample size $n=200$, we can remark the higher dispersion of the posterior distribution of $N$ for the negative binomial mixtures. In case (ii), when the prior guess of $N$ is wrong, e.g., $\E(N) = 400$, the inference gets worse as the sample size decreases, e.g., $n=200$. This is clearly expected since the posterior of $N$ gives more weight to the contribution of the prior, which brings an incorrect guess. However, It is interesting to note how the negative binomial mixtures are perfectly catching the true richness even for $n=200$. 

\begin{figure}[tbp]
    \centering    
    \includegraphics[width = \linewidth]{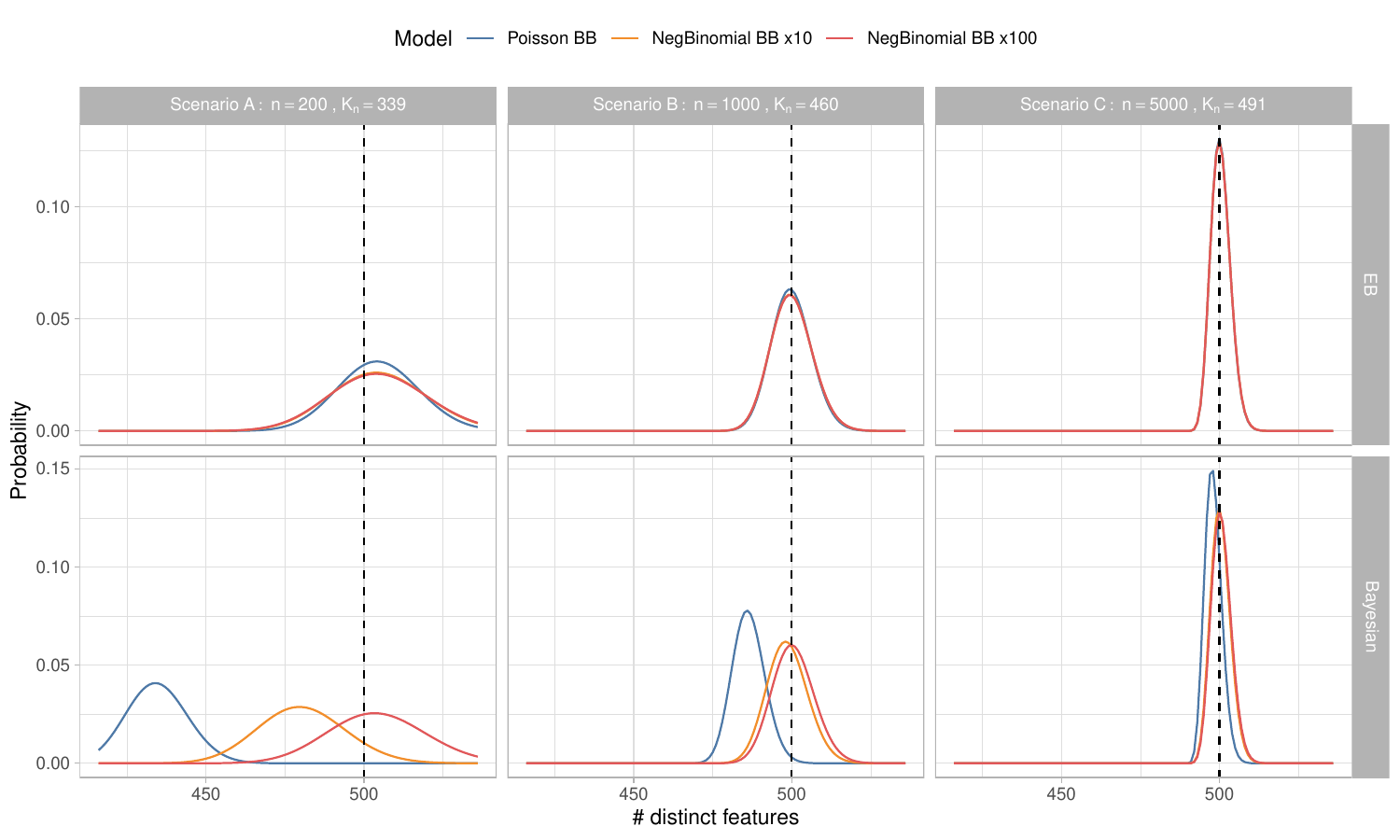}    
    
    \caption{Posterior distribution of the richness $N$ for increasing sample sizes $n \in \{200, 1000, 5000\}$. Top row: parameters elicited according to the empirical Bayes approach; bottom row: $\E(N) = 400$, with $\alpha$ and $\theta$ chosen as in Section \ref{sec:parameter_elic}. The vertical black lines indicate the true richness $H=500$.}
    \label{fig:richness_distr_simulation_beta}
\end{figure}

\section{Ecological applications: additional details and fully Bayesian approach}
\label{app:applications}

\subsection{Additional analyses}\label{app:additional_plots}

{In the present section we provide additional details and analyses for the ecological applications discussed in Section \ref{sec:real_data} of the main text. For each ecological dataset, we present: (i) the observed taxon accumulation curve, (ii) the $95\%$ credible intervals of $K_1,\ldots,K_n$ and of $K_{n,r}$, $r \geq 1$, based on some examples from the class of mixtures favored by the model selection procedures, (iii) the results of a data-holdout experiment, in which each model is trained on half of the observed data and evaluated in terms of its predictive performance on the held-out portion. }

\subsubsection*{Vascular plants in Danish forest}

{Figure \ref{fig:accumulation_plants} shows the observed taxon accumulation for the vascular plants in \cite{Mazz2016}; the plot clearly indicates that the asymptote of the curve has not yet been reached, suggesting that the species richness will exceed the observed number of species in the sample. 

As detailed in the main text, the model-checking procedures favor mixtures of \textsc{ibp}s for this dataset.
Accordingly, Figure \ref{fig:rare_knr_plants_credible_intervals} reports the $95\%$ credible intervals of $K_1,\ldots,K_n$ and of $K_{n,r}$, $r \geq 1$, for two examples of gamma mixture of \textsc{ibp}s.

For the data-holdout experiment, we randomly partition the observed data, selecting half as the training set, whose size equals $n/2 = 51$, with $n=102$. Then, we fit the Poisson and two examples of negative binomial mixture of \textsc{bb}s, along with two examples of gamma mixture of \textsc{ibp}s, on the training data. In Figure \ref{fig:extr_held_out_plants}, we compare the predictions of the extrapolation curve from all fitted models against the observed extrapolation curve in the held-out test set. The plots clearly show that the mixtures of \textsc{ibp}s effectively predict the number of unseen species in the training data which are displayed in the test data. It is also apparent  that the mixtures of \textsc{bb}s struggle to capture the true growth of the observed extrapolation curve. }

%%%%% ACCUMULATION CURVE PLANTS
\begin{figure}[tbp]
    \centering    
    \includegraphics[width = 0.5\linewidth]{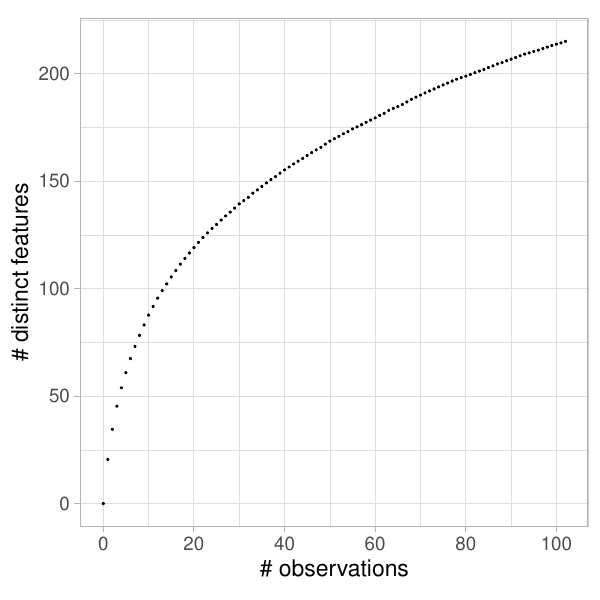}
    \caption{Taxon accumulation curve for the vascular plants in \cite{Mazz2016}.}
    \label{fig:accumulation_plants}
\end{figure}

%%%% RAREFACTION AND KNR INTERVALS PLANTS
\begin{figure}[tbp]
    \centering    
    \includegraphics[width = 0.49\linewidth]{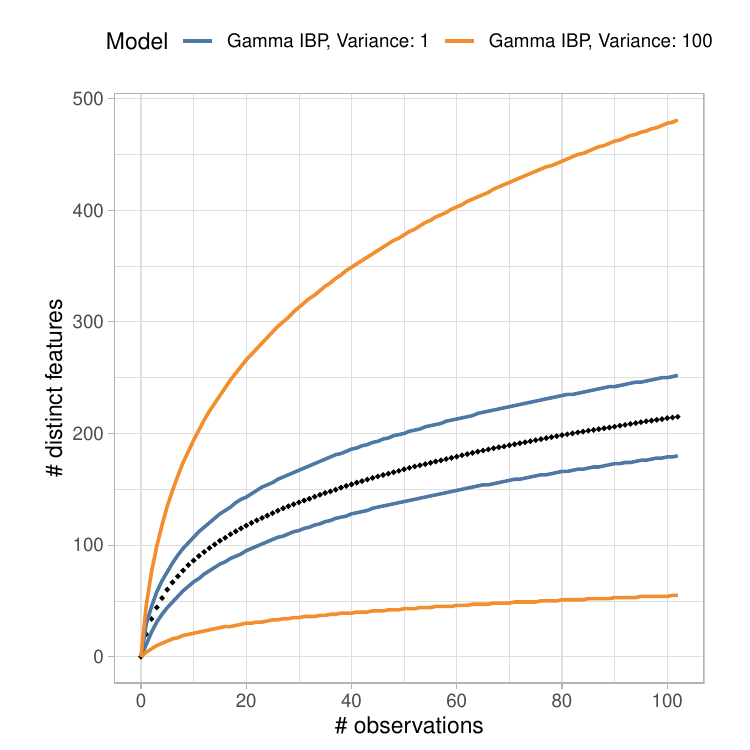}    
    \includegraphics[width = 0.49\linewidth]{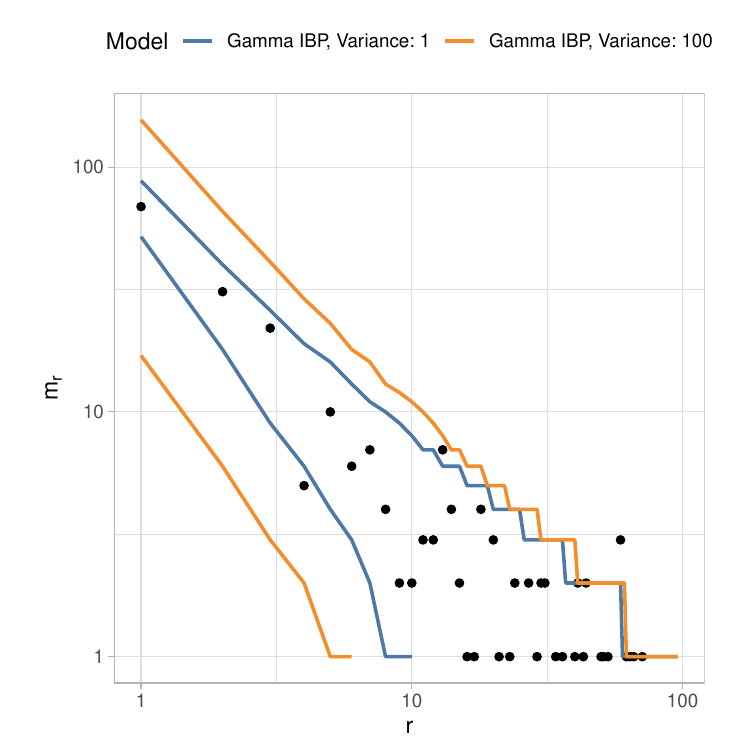}
    \caption{{
    Left panel: the empirical accumulation curve (black dots) and the credible intervals (delimited by colored lines) of $K_1,\ldots,K_n$ for two examples of gamma mixture of \textsc{ibp}s. Right panel: the observed values of $K_{n,r}$ (black dots) and the credible intervals (delimited by colored lines) of $K_{n,r}$  for the same mixtures of \textsc{ibp}s. The right plot is in log-log scale.} }
    \label{fig:rare_knr_plants_credible_intervals}
\end{figure}

%%%% EXTRAPOLATION ON HOLD-OUT PLANTS
\begin{figure}[tbp]
    \centering    
     
    \includegraphics[width = \linewidth]{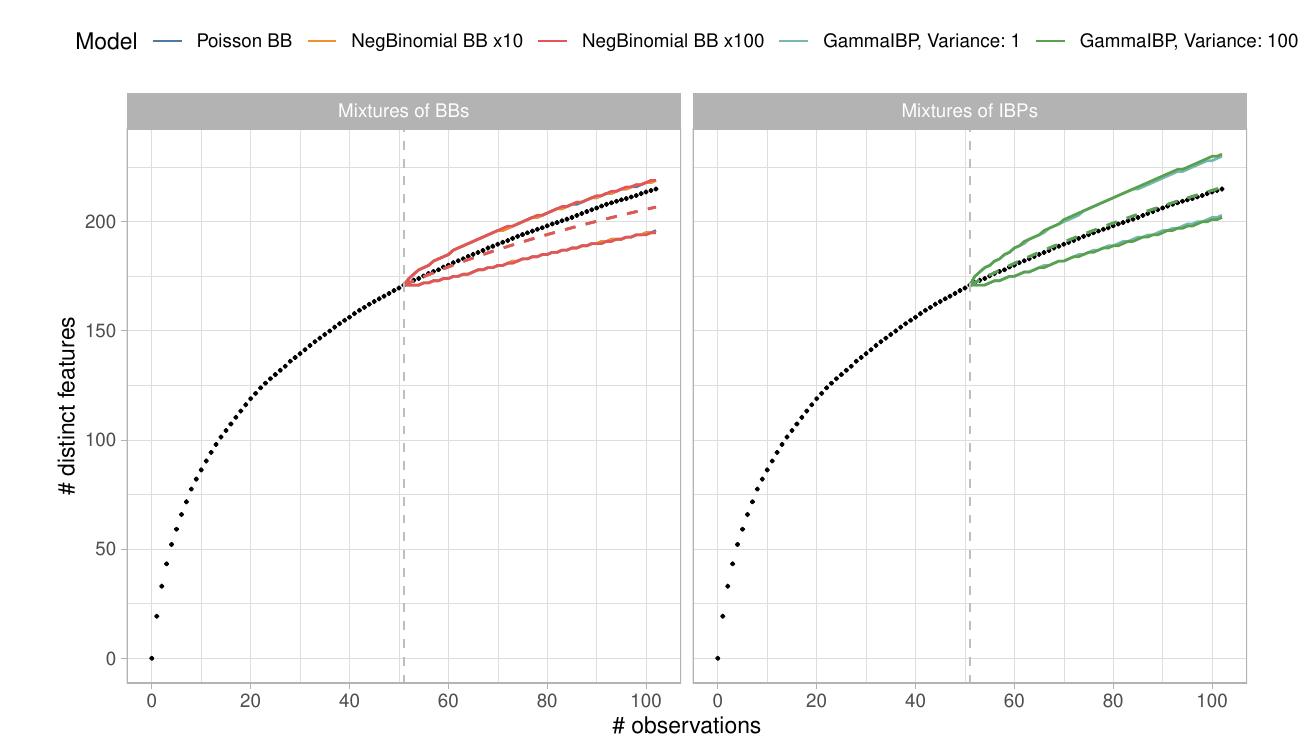}
   
    \caption{{Data-holdout analysis with training set size equal to $n/2$, where $n=102$: expected values (dashed lines) and $95\%$ credible intervals (delimited by solid lines) of $K_m^{(n/2)} + k \mid \bm{Z}^{(n/2)}$, for the mixtures of \textsc{bb}s (left) and the mixtures of \textsc{ibp}s (right). 
    The grey vertical lines indicate the training set.} }
    \label{fig:extr_held_out_plants}
\end{figure}

\subsubsection*{Trees in Barro Colorado Island}

{Figure \ref{fig:accumulation_bci} illustrates the observed taxon accumulation for the 
Barro Colorado Island dataset, freely available in the \textsc{vegan} package in \textsc{R}. The curve does not exhibit any clear asymptote, indicating that the species richness likely exceeds the observed number of species in the sample. 

As discussed in the main text, the model-checking procedures favor mixtures of \textsc{bb}s for this dataset.
Accordingly, Figure \ref{fig:rare_knr_bci_credible_intervals} reports the $95\%$ credible intervals of $K_1,\ldots,K_n$ and of $K_{n,r}$, $r \geq 1$, for the Poisson and two examples of negative binomial mixture of \textsc{bb}s.

For the data-holdout experiment, we randomly split the dataset, using half of the data as the training set, whose size equals $n/2 = 50$, with $n=100$. Then, we fit the Poisson and two examples of negative binomial mixture of \textsc{bb}s, along with two examples of gamma mixture of \textsc{ibp}s, on the training data. In Figure \ref{fig:extr_held_out_bci}, we compare the predictions of the extrapolation curve from all fitted models against the observed extrapolation curve in the held-out test set. The plots clearly show that the mixtures of \textsc{bb}s successfully predict the number of previously unseen species observed in the test data, while the mixtures of \textsc{ibp}s fail to capture this extrapolation behavior.}

%%%%% ACCUMULATION CURVE BCI
\begin{figure}[tbp]
    \centering    
    \includegraphics[width = 0.5\linewidth]{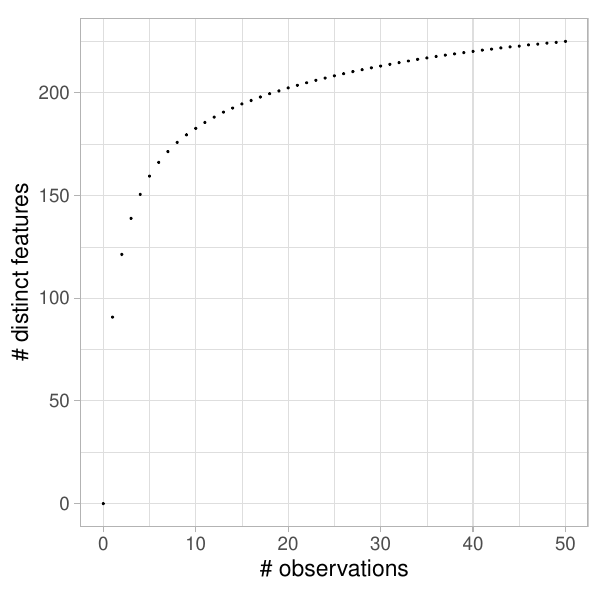}
    \caption{Taxon accumulation curve for the trees in the Barro Colorado Island data.}
    \label{fig:accumulation_bci}
\end{figure}

%%%% RAREFACTION AND KNR INTERVALS BCI
\begin{figure}[tbp]
    \centering    
    \includegraphics[width = 0.49\linewidth]{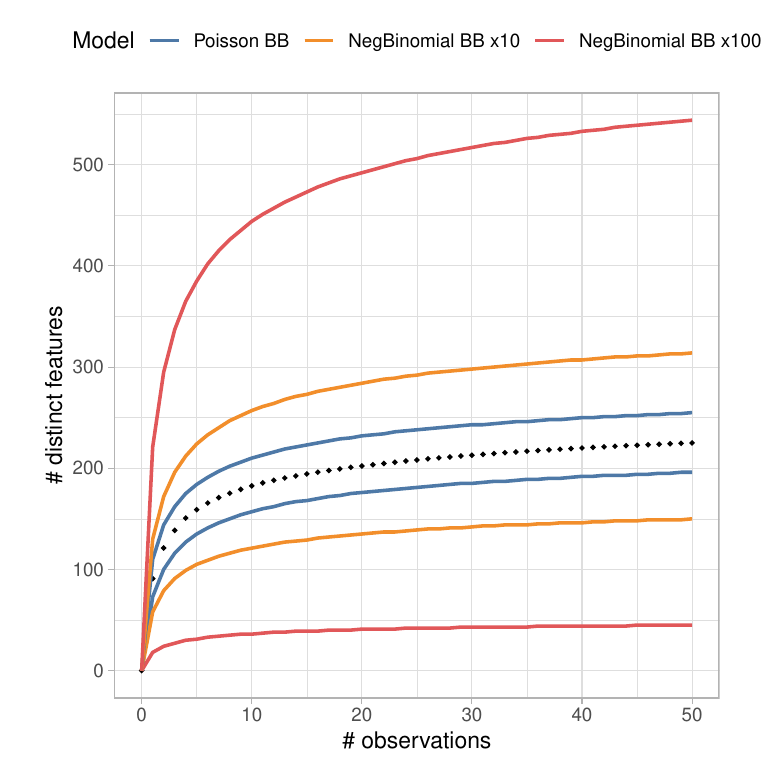}    
    \includegraphics[width = 0.49\linewidth]{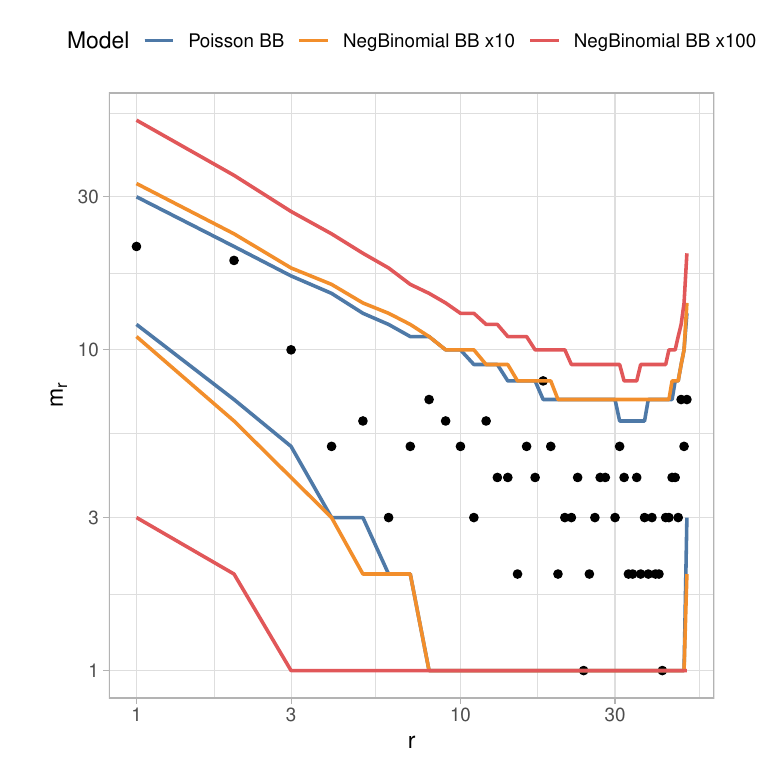}
    \caption{{Left panel: the empirical accumulation curve (black dots) and the credible intervals (delimited by colored lines) of $K_1,\ldots,K_n$ for the Poisson mixture of \textsc{bb}s and two examples of negative binomial mixture of \textsc{bb}s. Right panel: the observed values of $K_{n,r}$ (black dots) and the credible intervals (delimited by colored lines) of $K_{n,r}$  for the same mixtures of \textsc{bb}s. The right plot is in log-log scale.} }
    \label{fig:rare_knr_bci_credible_intervals}
\end{figure}

%%%% EXTRAPOLATION ON HOLD-OUT BCI
\begin{figure}[tbp]
    \centering    
    
    \includegraphics[width = \linewidth]{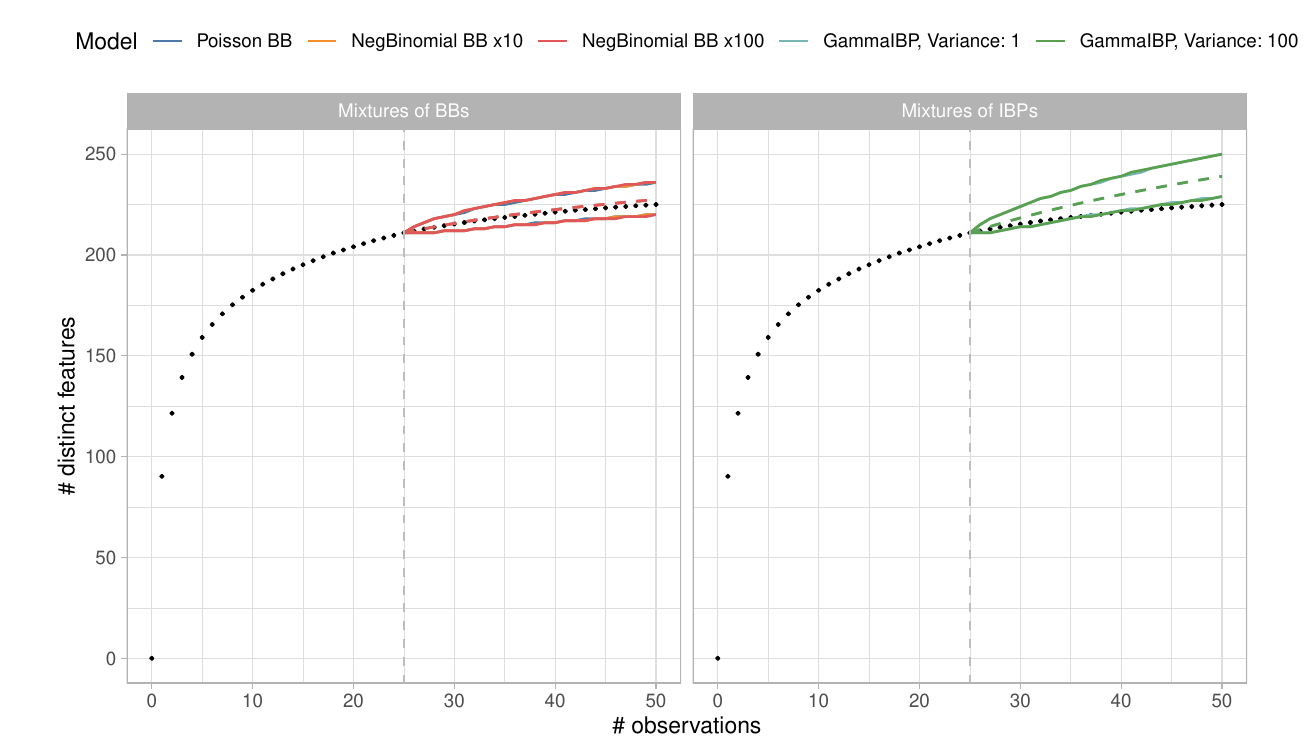}
    \caption{{Data-holdout analysis with training set size equal to $n/2$, where $n=100$: expected values (dashed lines) and $95\%$ credible intervals (delimited by solid lines) of $K_m^{(n/2)} + k \mid \bm{Z}^{(n/2)}$, for the mixtures of \textsc{bb}s (left) and the mixtures of \textsc{ibp}s (right). 
    The grey vertical lines indicate the training set.} }
    \label{fig:extr_held_out_bci}
\end{figure}

\subsection{Fully Bayesian approach: details}\label{app:application_prior}

In Section \ref{sec:parameter_elic} we describe a recommended procedure for eliciting the parameters of the models via an empirical Bayes approach. We consider such a procedure for all the discussed simulation studies and real data examples. However, one might prefer to adopt a fully Bayesian approach instead of the proposed empirical Bayes one. Specifically, the interest might be in assuming prior distributions for the parameters $\alpha$ and $\theta$, for all of the mixtures. Within the mixtures of \textsc{bb}s, we argue that the negative binomial mixture has to be regarded as the fully Bayesian version of the Poisson mixture, since the former is obtained by placing a gamma prior on the $\lambda$ parameter of the latter. 

In this section, we first discuss the prior elicitation for the parameters $\alpha$ and $\theta$; then, we illustrate the comparison with the empirical Bayes approach, in terms of inference and prediction obtained in the two real data scenarios. As far as prior elicitation is concerned, since the constraint $\theta > - \alpha$ for all of the models, it is convenient to perform a change of variable, introducing $s = \theta + \alpha$ and removing $\theta$, so that $s > 0$. Then, prior distributions on $\alpha$ and $s$ are specified. Specifically, for the mixtures of \textsc{ibp}s, characterized by $\alpha \in [0,1)$, we consider
\[
\alpha \sim \dbeta\left(a_\alpha, b_\alpha\right),\; s \sim \dgamma\left(a_s, b_s\right),
\]
with $a_\alpha, b_\alpha > 0$ and $a_s, b_s > 0$. Denoting with $\hat{\alpha}$ and $\hat{\theta}$ the empirical Bayes estimates obtained as in Section \ref{sec:parameter_elic}, we select the hyperparameters of the priors by imposing $\E(\alpha) = \hat{\alpha}$ and $\E(s) = \hat{s}$, where $\hat{s} = \hat{\alpha} + \hat{\theta}$. It follows that $a_\alpha/(a_\alpha + b_\alpha) = \hat{\alpha}$ and $a_s/b_s = \hat{s}$. Two additional equations, one involving $a_\alpha, b_\alpha$ and one involving $a_s, b_s$, are necessary to fix the four hyperparameters. For example, such equations may control the prior variances of $\alpha$ and $s$, so that the practitioner can select the degree of regularization induced by the priors. Draws from the posterior distribution of $(\alpha, s)$ can be easily obtained via any Metropolis-Hastings algorithm. Specifically, we apply the change of variables $\alpha' = \log(\alpha/(1-\alpha))$, $s' = \log(s)$, so that $(\alpha', s') \in \R^2$, and we update $(\alpha', s')$ via a pre-conditioned \textsc{mala} (Metropolis Adjusted Langevin Algorithm), since the gradient of the log-full-conditional density of $(\alpha', s')$ is available in closed form.

For the mixtures of \textsc{bb}s, having $\alpha < 0$, we consider
\[
-\alpha \sim \dgamma\left(a_\alpha, b_\alpha\right),\; s \sim \dgamma\left(a_s, b_s\right),
\]
with $a_\alpha, b_\alpha > 0$ and $a_s, b_s > 0$. We choose the hyperparameters of the priors such that $\E(\alpha) = \hat{\alpha}$ and $\E(s) = \hat{s}$. This imposes $a_\alpha/b_\alpha = - \hat{\alpha}$ and $a_s/b_s = \hat{s}$. Similarly to the mixtures of \textsc{ibp}s, one may set the degree of regularization induced by the priors by choosing the prior variances of $\alpha$ and $s$. Sampling from the posterior distribution of $(\alpha, s)$ is straightforward through any Metropolis-Hastings algorithm. We resort again to a pre-conditioned \textsc{mala} for the transformed parameters $(\alpha', s')$, where $\alpha' = \log(-\alpha)$, $s' = \log(s)$, and we stress that the gradient of the log-full-conditional density of $(\alpha', s')$ is available in closed form.

Here, we report the comparison among the fully Bayesian approach and the empirical Bayes approach described in Section \ref{sec:parameter_elic}, in terms of inference and/or prediction, concerning the two real data scenarios.

\subsubsection*{Vascular plants in Danish forest} Analyzing the vascular plants data of \cite{Mazz2016} in Section \ref{sec:vascular_plants}, we claimed the correct specification of the mixtures of \textsc{ibp}s. We consequently showed the prediction obtained via the gamma mixture of \textsc{ibp}s, using the empirical Bayes approach to select the parameters. In particular, the empirical Bayes estimates for $\alpha$ and $\theta$ result as follows: $\hat{\alpha} = 0.17$, $\hat{\theta} = 1.7$. For the fully Bayesian approach, we then select the hyperparameters by imposing: $\E(\alpha) = \hat{\alpha} = 0.17$, $\Var(\alpha) = 10^{-2}$, and $\E(s) = \hat{s} = 1.87$, $\Var(s) = 187$, so that a moderate amount of regularization is introduced. We run the \textsc{mcmc} algorithm for $5\cdot 10^4$ iterations, discarding the first $5\cdot10^3$ and keeping one every two iterations. In Figure \ref{fig:extr_plants_prior}, we report the comparison between the fully Bayesian approach and the empirical Bayes approach described in Section \ref{sec:parameter_elic}, in terms of the extrapolation curve. We observe that, for both the choices of the prior variance of $\gamma$, the expected number of unseen species which will be collected in additional samples of increasing sizes $m$ is almost identical for the two approaches. This is desirable since it confirms that the empirical Bayes approach leads to equivalent prediction than the fully Bayesian procedure. On the other hand, the credible intervals produced by the fully Bayesian approach are larger than the ones estimated with the empirical Bayes method, as naturally expected.

\begin{figure}[tbp]
    \centering    
    \includegraphics[width = \linewidth]{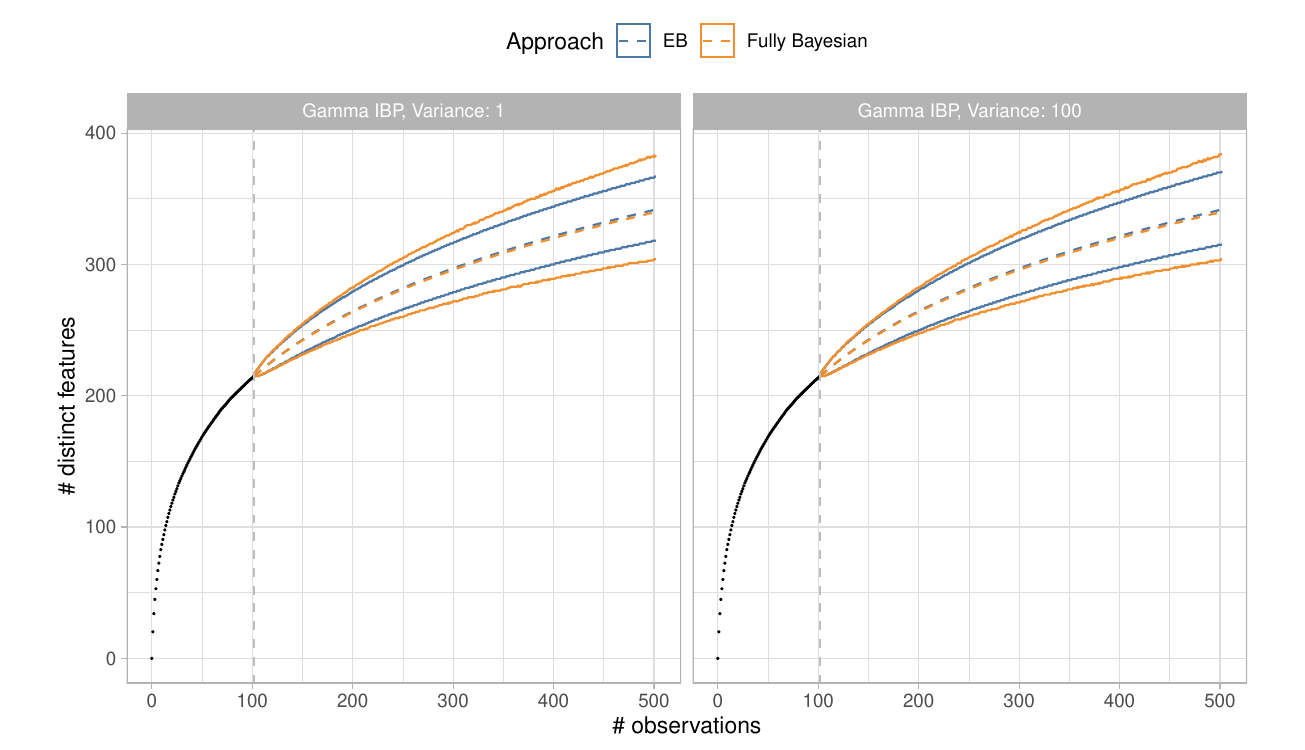}    
    
    \caption{Comparison between the empirical Bayes approach (EB) and the fully Bayesian approach in terms of expected values and the $95\%$ credible intervals for $K_m^{(n)} + k$, with $k = 215$, for the gamma mixtures of \textsc{ibp}s. The extrapolation horizon is $m=1,\ldots, 400$. The left panel shows the case {$\Var(\gamma) = 1$}, the right panel considers $\Var(\gamma) = 100$.}
    \label{fig:extr_plants_prior}
\end{figure}

\subsubsection*{Trees in Barro Colorado Island} Differently than for the vascular plants of \cite{Mazz2016}, the tree data investigated in Section \ref{sec:bci_data} are suitably modelled via the mixtures of \textsc{bb}s (refer to Figure \ref{fig:rare_knr_bci}). In the main text, we report the inference and prediction produced by the Poisson and negative binomial mixtures of \textsc{bb}s, with parameters selected via the empirical Bayes approach discussed in Section \ref{sec:parameter_elic}. In particular, we get $\hat{\alpha} = - 0.29$ and $\hat{\theta} = 0.95$ as empirical Bayes estimates for the parameters. In the fully Bayesian approach, we select the hyperparameters so that $\E(\alpha) = \hat{\alpha} = -0.29$, $\Var(\alpha) \approx 300$, and $\E(s) = \hat{s} = 0.66$, $\Var(s) \approx 650$. We run the \textsc{mcmc} algorithm for $5\cdot 10^4$ iterations, discarding the first $5\cdot10^3$ and keeping one every two iterations. In Figure \ref{fig:extr_bci_prior}, we report the comparison between the fully Bayesian approach and the empirical Bayes approach, in terms of the extrapolation curve. The expected values for the extrapolation curve are similar between the two approaches, with larger credible intervals in the fully Bayesian case. This behavior is expected and desired. Moreover, Figure \ref{fig:richness_bci_prior} shows the posterior distribution of $N$, the species richness, under the different approaches. It is evident that the fully Bayesian approach leads to much more dispersed posterior distributions for $N$. In particular, we remind that the empirical Bayes procedure estimates an expected species richness of $296.17$, with credible interval $[278, 316]$, for both the choices of the prior variance of $N$. The fully Bayesian approach produces an expected species richness of $306.24$ with credible interval $[262, 375]$, for the case $\Var(N) = 10 \times \mu_0$, while it leads to an expected species richness of {$316.44$ with credible interval $[261, 425]$, for the case $\Var(N) = 100 \times \mu_0$}.

\begin{figure}[tbp]
    \centering    
    \includegraphics[width = \linewidth]{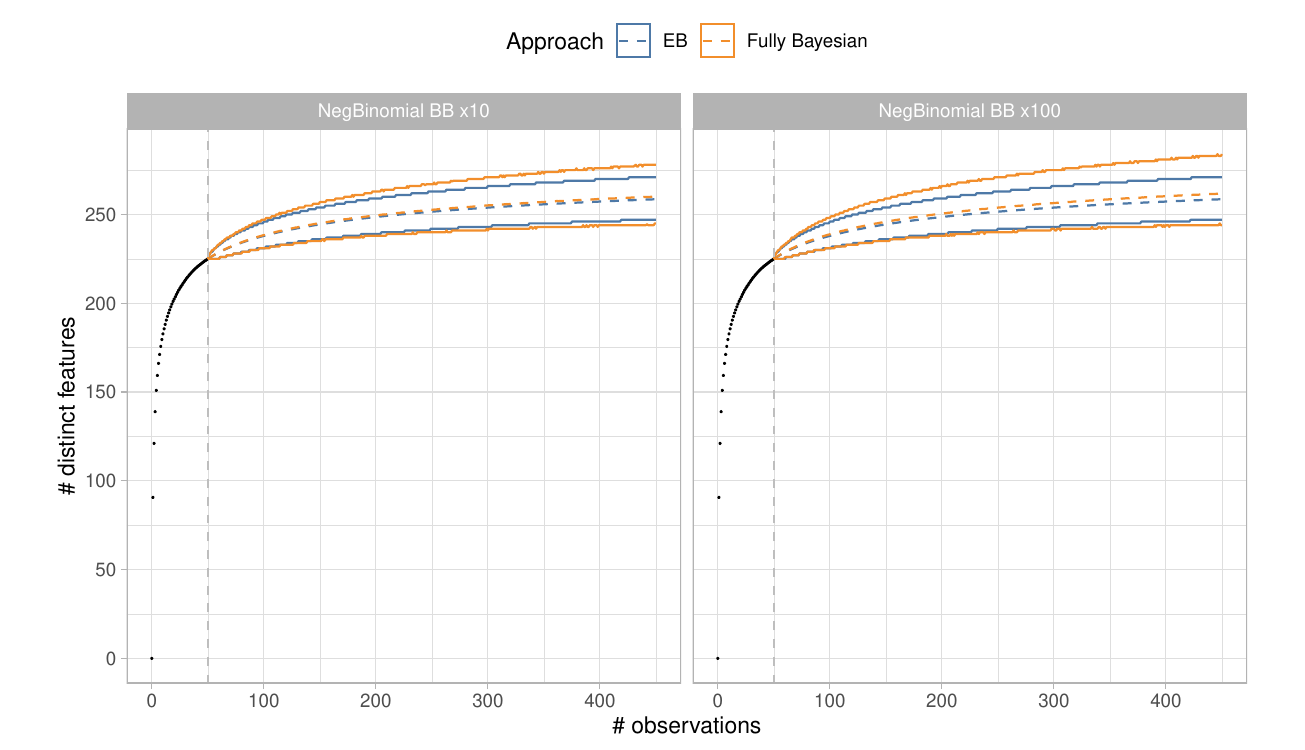}    
    
    \caption{Comparison between the empirical Bayes approach (EB) and the fully Bayesian approach in terms of expected values and the $95\%$ credible intervals for $K_m^{(n)} + k$, with $k = 225$, for the negative binomial mixtures of \textsc{bb}s. The extrapolation horizon is $m=1,\ldots, 400$. The left panel shows the case $\Var(N) = 10 \times \mu_0$, the right panel considers {$\Var(N) = 100 \times \mu_0$}.}
    \label{fig:extr_bci_prior}
\end{figure}

\begin{figure}[tbp]
    \centering    
    \includegraphics[width =  \linewidth]{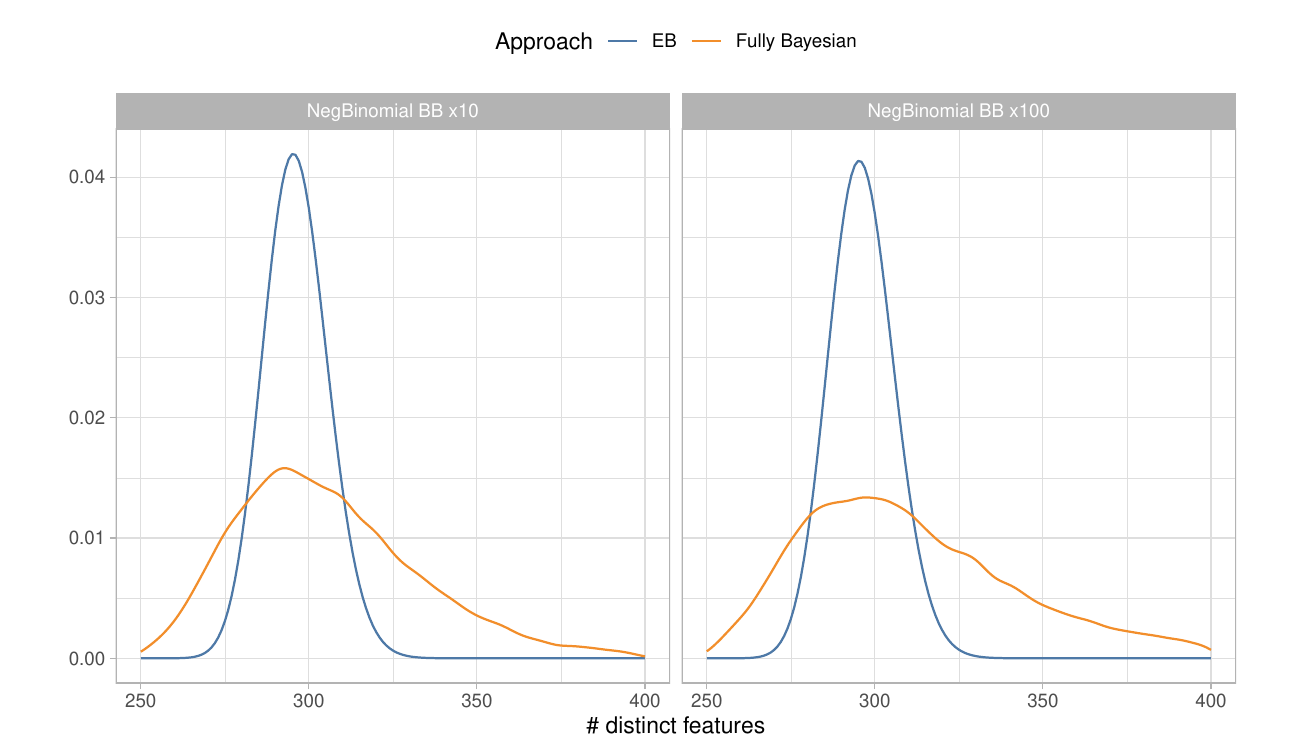}    
    
    \caption{Comparison between the empirical Bayes approach (EB) and the fully Bayesian approach in terms of the posterior distribution of the species richness $N$, for the negative binomial mixtures of \textsc{bb}s. The left panel shows the case $\Var(N) = 10 \times \mu_0$, the right panel considers {$\Var(N) = 100 \times \mu_0$}.}
    \label{fig:richness_bci_prior}
\end{figure}

{
\section{Discussion on computational complexity}

In this section, we present a brief discussion on some computational aspects related to the proposed methodology. In general, the computational procedure of our methodology is highly efficient, due to the availability of closed form expressions for the \textsc{efpf}s.
More specifically, the computational effort depends on the chosen parameter elicitation strategy, namely: (i) the empirical Bayes approach described in Section \ref{sec:parameter_elic}; (ii) the fully Bayesian approach presented in the real data analysis of Section \ref{app:application_prior}.

Concerning (i), i.e., the empirical Bayes approach (Section \ref{sec:parameter_elic}), computation is limited to optimizing the parameters of the mixtures of \textsc{bb}s or the mixtures of \textsc{ibp}s. Specifically, for the mixtures of \textsc{bb}s, the parameters are estimated by maximizing the \textsc{efpf} of the \textsc{bb} model as given in Equations \eqref{eq:EFPF_product}--\eqref{eq:EFPF_BB}. The computational complexity of each evaluation of the \textsc{efpf} of the \textsc{bb} model, as a function of the sample size $n$ and the number of observed features $k$, scales as $k$, since the complexity of the term in \eqref{eq:EFPF_BB} is independent of both $n$ and $k$. Similarly, for the mixtures of \textsc{ibp}s, parameter estimation involves maximizing the \textsc{efpf} of the \textsc{ibp} model from Equations \eqref{eq:EFPF_product}--\eqref{eq:EFPF_IBP}. Here, each evaluation scales as $n + k$, where the term in Equation \eqref{eq:EFPF_IBP} contributes to the dependence on $n$.
For both optimization problems, we use the  \texttt{nlminb} function from \texttt{R}'s \texttt{stats} package. This results in extremely fast computations across all settings.

Concerning (ii), i.e.,  the fully Bayesian approach (Section \ref{app:application_prior}), \textsc{mcmc} algorithms are employed for the updating of the parameters $\alpha$ and $\theta$. In particular, for the mixtures of \textsc{bb}s, we apply a pre-conditioned \textsc{mala} to the transformed parameters $(\alpha', s')$, with $\alpha' = \log(-\alpha)$ and $s' = \log(s)$. The gradient of the log-full-conditional density of $(\alpha', s')$ is available in closed form. Similarly, for the gamma mixture of \textsc{ibp}s, we apply the change of variables $\alpha' = \log(\alpha / (1 - \alpha))$ and $s' = \log(s)$, again resorting to a pre-conditioned \textsc{mala}, where the gradient of the log-full-conditional density of $(\alpha', s')$ is available in closed form.
Thus, in all cases, the \textsc{mcmc} algorithms involve only two transformed parameters and benefits from analytic gradients, ensuring both speed and stability. 

To quantify the efficiency of the \textsc{mcmc} implementations, we provide below the computational details which concern the real data examples of Section \ref{app:application_prior}. In particular, for the gamma mixture of \textsc{ibp}s on the \emph{Vascular plants in Danish forest}, under the choice $\operatornamewithlimits{Var}(\gamma) = 100$, we run $5\cdot 10^4$ iterations, discarding the first $5\cdot10^3$ and keeping one every two iterations. This yields effective sample sizes of $8071.03$ for $\alpha$ and $11396.61$ for $\theta$. The total runtime is $36.71$ seconds on a Windows 11 system with an Intel Core i7-1165G7 CPU @ 2.80 GHz and 8 GB RAM. Finally, for the negative binomial mixture of \textsc{bb}s on the \emph{Trees in Barro Colorado Island}, using $\operatornamewithlimits{Var}(N) = 10 \times \mu_0$ and the same \textsc{mcmc} settings, we obtain effective sample sizes of $11456.89$ for $\alpha$ and $11453.94$ for $\theta$. The total runtime is $30.23$ seconds using the same machine.
}
\end{document}